\documentclass{scrartcl}

\usepackage[T1]{fontenc}	
\usepackage[utf8]{inputenc}	
\usepackage{lmodern} 

\usepackage{amsmath}	
\usepackage{amsthm} 	
\usepackage{amssymb} 

\usepackage{graphicx} 
\usepackage{tikz}	
\usetikzlibrary{arrows,automata,shadows}
\usetikzlibrary{matrix}
\usetikzlibrary{chains}
\usetikzlibrary{positioning}
    \everymath{\displaystyle}
    \tikzstyle{mathbox} = [inner sep=0pt, anchor=base]
    \tikzstyle{every picture}+=[remember picture]
    \usetikzlibrary{calc}
    \usetikzlibrary{positioning}
\usepackage{hyperref}
\hypersetup{
colorlinks=false,
linkbordercolor={1 1 1}, 
citebordercolor={1 1 1}, 
urlbordercolor={1 1 1} 	 
}

\def\tens#1{\ensuremath{\mathsf{#1}}}	
\def\vec#1{\ensuremath{\mathchoice	
                     {\mbox{\boldmath$\displaystyle\mathbf{#1}$}}
                     {\mbox{\boldmath$\textstyle\mathbf{#1}$}}
                     {\mbox{\boldmath$\scriptstyle\mathbf{#1}$}}
                     {\mbox{\boldmath$\scriptscriptstyle\mathbf{#1}$}}}}
\newtheorem{theorem}{Theorem}[section]
\newtheorem{lemma}{Lemma}[section]
\newtheorem{definition}{Definition}[section]

\title{PageRank for evolving link structures
}

\author{
 Christopher Engstr{\"o}m,\\
Division of Applied Mathematics\\
School of Education, Culture and Communication (UKK)\\
M{\"a}lardalen University,\\
christopher.engstrom@mdh.se\\[0.5cm] 
Sergei Silvestrov\\
Division of Applied Mathematics\\
School of Education, Culture and Communication (UKK)\\
M{\"a}lardalen University\\
sergei.silvestrov@mdh.se}

\begin{document}

\maketitle

\date{}


\begin{abstract}
In this article we will look at the PageRank algorithm used as part of the ranking process of different Internet pages in search engines by for example Google. This article has its main focus in the understanding of the behavior of PageRank as the system dynamically changes either by contracting or expanding such as when adding or subtracting nodes or links or groups of nodes or links. In particular we will take a look at link structures consisting of a line of nodes or a complete graph where every node links to all others.

We will look at PageRank as the solution of a linear system of equations and do our examination in both the ordinary normalized version of PageRank as well as the non-normalized version found by solving the linear system. We will see that it is possible to find explicit formulas for the PageRank in some simple link structures and using these formulas take a more in-depth look at the behavior of the ranking as the system changes. 
\end{abstract}

\section{Introduction}
\label{intro}
PageRank is a method in which we can rank nodes in different link structures such as Internet pages on the Web in order of "importance" given the link structure of the complete system. It is important that the method is extremely fast since there is a huge number of Internet pages. It is also important that the algorithm returns the most relevant results first since very few people will look through more than a couple of pages when doing a search in a search engine. \cite{Brin1998107} 

While PageRank was originally constructed for use in search engines, there are other uses of PageRank or similar methods, for example in the EigenTrust algorithm for reputation management to decrease distribution of unauthentic files in P2P networks. \cite{Kamvar:2003:EAR:775152.775242}

Calculating PageRank is usually done using the Power method which can be implemented very efficiently, even for very large systems. The convergence speed of the Power method and it's dependence on certain parameters have been studied to some extent. For example the Power method on a graph structure such as that created by the Web will converge with a convergence rate of $c$, where $c$ is one of the parameters used in the definition \cite{ilprints582}, and the problem is well conditioned unless $c$ is very close to $1$ \cite{ilprints597}. However since the number of pages on the Web is huge, extensive work has been done in trying to improve the computation time of PageRank even further. One example is by aggregating webpages that are "close" and are expected to have a similar PageRank as in \cite{5399514}. Another method used to speed up calculations is found in \cite{Kamvar200451} where they do not compute the PageRank of pages that have already converged in every iteration. Other methods to speed up calculations include removing "dangling nodes" before computing PageRank and then calculate them at the end or explore other methods such as using a power series formulation of PageRank \cite{FAndersson_art_PR}.

There are also work done on the large scale using PageRank and other measures in order to learn more about the Web, for example looking at the distribution of PageRank both theoretically and experimentally such as in \cite{Dhyani:2003:DVS:942051.942054}.

While the theory behind PageRank is well understood from Perron-Frobenius theory for non-negative irreducible matrices \cite{berman1994nonnegative,gantmacher1959theory,lancaster1969theory} and the study of Markov chains \cite{NorrisMC,Ryden2000m}, how PageRank is affected from changes in the the system or parameters is not as well known. 


In this artcle we start by giving a short introduction on PageRank and some notation and definitions used throughout the article. 
We will look at PageRank as the solution to a linear system of equations and what we can learn using this representation. Looking at some common graph structures we want to gain a better understanding of the changes in PageRank as the graph structure changes. This could for example be used in finding good approximations of PageRank of certain structures in order to speed up calculations further. 
We will look at both the "ordinary" normalized version of PageRank as well as a non-normalized version we get by solving the linear system. We will see how this non-normalized version corresponds to the probabilities of a random walk through the graph and how we can use this to find the PageRank of some systems using this perspective rather than solving the system or computing the dominant eigenvector. 
Mainly two different structures, first a simple line in Sect.~\ref{SimpleLine_sec} and later a complete graph in Sect.~\ref{CCompGraph} will be examined. In both cases we will see that we can find explicit expressions for the PageRank depending on the number of nodes. In both cases of the "ordinary" PageRank as well as a non-normalized version expressions for the PageRank will be found for both the structure itself as well as the PageRank after doing some simple modifications. The last graph structure we will look at is when we combine the simple line with the complete graph by adding a link between them in Sect.~\ref{conn_line_complete}.
In Sect.~\ref{closer_look_formula} and Sect.~\ref{formula_norm} we will take a closer look at the found formulas for some of the examples mainly by looking at partial derivatives of the PageRank. We will see one of the possible reasons why $c$ is usually choosen to be around $c \approx 0.85$. PageRank for some nodes increases extremely fast while for some other nodes decreases extremely fast for larger $c$, while for lower $c$ the difference in PageRank between nodes is smaller the lower $c$ gets and the initial weight vector have a much larger influence on the final ranking. 
Last we take a short look at what happens when changing the weight vector $\vec{V}$ present in the PageRank formulation as well as giving a short comparison of the differences and similarities between normalized and non-normalized PageRank. 

\section{Calculating PageRank} \label{sec: calcP}

Starting with a number of nodes (Internet pages) and the non-negative matrix $\tens{A}$ with every element $a_{ij} \ne 0$ corresponding to a link from node $i$ to node $j$. The value of element $a_{ij} = 1/n$ where $n$ is the number of outgoing links from node $i$. An example of a graph and corresponding matrix can be seen in Fig.~\ref{exGraph}. 
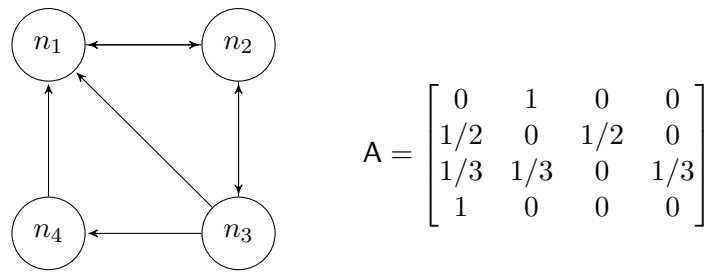
\begin{figure} [!hbt]
\begin{center}
 \begin{tikzpicture}[->,shorten >=1pt,auto,node distance=2.5cm,on grid,>=stealth',
every state/.style={circle,,draw=black}]
\node[state] (A) {$n_1$};
\node[state] (B) [ right=of A] {$n_2$};
\node[state] (C) [below=of B] {$n_3$};
\node[state] (D) [below=of A] {$n_4$};
\path (A) edge node {} (B)
(B) edge node {} (A)
	edge node {} (C)
(C) edge node {} (A)
	edge node {} (B)
	edge node {} (D)
(D) edge node {} (A);
\node [right=1.5cm, yshift = 1.25cm ,text width=4cm] at (C){
 \begin{equation*} \tens{A} = 
\begin{bmatrix}
0 & 1 & 0 & 0  \\
1/2 & 0 & 1/2 & 0 \\
1/3 & 1/3 & 0 & 1/3 \\
1 & 0 & 0 & 0\\
  \end{bmatrix} 
\end{equation*}
};
\end{tikzpicture}
\end{center}
\caption{Directed graph and corresponding matrix system matrix $\tens{A}$}
\label{exGraph}
\end{figure}

Note that we by convention do not allow a node to link to itself. We also need that no nodes have zero outgoing links (dangling nodes) resulting in a row with all zeros. For now we assume that none of these dangling nodes are present in the link matrix. This means that every row will sum to one in the link matrix $\tens{A}$.

The PageRank vector $\vec{R}$ we want for ranking the nodes (pages) is the eigenvector corresponding to the eigenvalue one of matrix $\tens{M}$: 
\begin{displaymath}\tens{M} = c\tens{A}^\top +(1-c) \vec{u}\vec{e}^\top\end{displaymath}
where $0<c<1$, usually $c = 0.85$, $\tens{A}$ the link matrix, $\vec{e}$ a column vector of the same length as the number of nodes ($n$) filled with ones and $\vec{u}$ is a column vector of the same length with elements $u_i$, $0\le u_i \le1$ such that $||\vec{u}||_1 = 1$. For $\vec{u}$ we will usually use the uniform vector (all elements equal) with $u_i = 1/n$ where $n$ is the number of nodes. The result after calculating the PageRank of the example matrix for the system in Fig.~\ref{exGraph} can be seen below: 
\begin{displaymath}\vec{R} \approx \left [ \begin{array}{c}
0.3328 \\
0.3763\\
0.1974 \\
0.0934  \end{array} \right ]\end{displaymath}
This can be seen as a random walk where we start in a random node depending on the weightvector $\vec{u}$. Then with a probability $c$ we go to any of the nodes linked to from that node and with a probability $1-c$ we instead go to a random (in the case of uniform $\vec{u}$) new node. The PageRank vector can be seen as the probability that you after a long time is located in the node in question.\cite{FAndersson_art_PR} More on why an eigenvector with eigenvalue $1$ always exists can be seen in for example \cite{A-25bilEig}.

\paragraph{Role of $c$.} Looking at the formula it is not immediately obvious why we demand $0<c<1$ and what role $c$ holds. We can easily see what happens at the limits, if $c=0$ the PageRank is decided only by the initial weights $\vec{u}$. However if $c = 1$ the weights have no role and the algorithm used for calculating PageRank might not even converge. As $c$ increases, nodes further and further away have an impact on the PageRank of individual nodes. And the opposite for low $c$, the lower $c$ is the more important is the immediate surrounding of a node in deciding its PageRank. The parameter  $c$ is also a very important factor in how fast the algorithms used to calculate PageRank converges, the higher $c$ is the slower the algorithm will converge.

\paragraph{Handling of dangling nodes.} If $\tens{A}$ contains dangling nodes, corresponding row no longer sums to one and there therefor will probably not be any eigenvector with eigenvalue equal to one. The method we use in order to fix this is to instead assume that the dangling nodes link to all nodes equally (or according to some other desired distribution). This gives us: $\tens{T} = \tens{A}+\vec{g}\vec{w}^\top$, where $\vec{g}$ is a column vector with elements equal to one for a dangling node and zero for all other nodes. Here $\vec{w}$ is the distribution according to how we make the dangling nodes link to other nodes (usually uniform or equal to $\vec{u}$). In this work we always use $\vec{w}=\vec{u}$ to simplify calculations. 

There are other ways to handle dangling nodes, for example by adding one new node linking only to itself and let all dangling nodes link to this node. Assuming $\vec{w}=\vec{u}$ these methods should be essentially the same apart from implementation \cite{Bianchini:2005:IP:1052934.1052938}.

\section{Notation and definitions} \label{sec: not}
Here we give some notes on the notation used through the rest of the article in order to clarify which variation of PageRank is used as well as some overall notation and the definition of some common important link structures. We will repeatedly use the $L^1$ norm in comparing the size of different vectors or (parts of) matrices.

First some overall notation:
\begin{itemize}
	\item	$S_G$:	The system of nodes and links for which we want to calculate PageRank, contains the system matrix $\tens{A}_G$ as well as a weight vector $\vec{v}_G$. Subindex $G$ can be either a capital letter or a number in the case of multiple systems.
	\item	$n_G$:	The number of nodes in system $S_G$.
	\item	$\tens{A}_G$:	System matrix where a zero element $a_{ij}$ means there is no link from node $i$ to node $j$. Non-zero elements are equal to $1/r_i$ where $r_i$ is the number of links from node $i$. Size $n_G\times n_G$.
	\item	$\vec{v}_G$:	Non-negative weight vector, not necessary with sum one. Size $n_G\times 1$.
	\item $\vec{u}_G$: 	The weight vector $\vec{v}_G$ normalized such that $||\vec{u}_G||_1 = 1$. We note that $\vec{u}_G$ is proportional to $\vec{v}_G$ ($\vec{u}_G \propto \vec{v}_G$). Size $n_G\times 1$.
	\item	$c$:	Parameter $0<c<1$ for calculating PageRank , usually $c=0.85$.
	\item $\vec{g}_G$: Vector with elements equal to one for dangling nodes and zero for all other in $S_G$. Size $n_G\times 1$.
	\item $\tens{M}_G$: Modified system matrix, $\tens{M}_G = c(\tens{A}_G+\vec{g}_G\vec{u}_G^\top)^\top+(1-c)\vec{u}_G\vec{e}^\top$ used to calculate PageRank, where $\vec{e}$ is the unit vector. Size $n_G\times n_G$.
	\item $S$:	Global system made up of multiple disjoint subsystems $S = S_1 \cup S_2 \ldots \cup S_N$, where $N$ is the number of subsystems.
	\item $\vec{V}$:	Global weight vector for system $S$, $\vec{V} = [\vec{v}_1^\top ~ \vec{v}_2^\top ~ \ldots ~ \vec{v}_N^\top]^\top$, where $N$ is the number of subsystems.
\end{itemize}
In the cases where there is only one possible system the subindex $G$ will often be omitted. For the systems making up $S$ we define disjoint systems in the following way. 
\begin{definition}
Two systems $S_1$, $S_2$ are disjoint if there are no paths from any nodes in $S_1$ to $S_2$ or from any nodes in $S_2$ to $S_1$.
\end{definition}

From earlier we saw how we could calculate PageRank for a system $S$, now we make the assumption that $\vec{w}=\vec{u}$ both since it simplifies calculations, but also since using two different weight vectors for essentially the same thing seems like it could create more problems and unexpected behavior than what you actually could gain from it.

We will use three different ways to define the different versions of PageRank using the notation: 
$$\vec{R}^{(t)}_{G}[S_H\rightarrow S_I, S_J\rightarrow S_K. \ldots]$$ 
where $t$ is the type of PageRank used, $S_G \subseteq S$ is the nodes in the global system $S$ for which $\vec{R}$ is the PageRank. Often $S_G = S$ and we write it as $\vec{R}^{(t)}_{S}$. In the last part within brackets we write possible connections between otherwise disjoint subsystems in $S$, for example an arrow to the right means there are links from the left system to the the right system. How many and what type of links however needs to be specified for every individual case. In more complicated examples there may be arrows pointing in two directions or a number above the arrow notifying how many links we have between the systems. 

We will sometimes give the formula for a specific node $j$ in this case it will be noted as $\vec{R}^{(t)}_{G,j}[S_H\rightarrow S_I,S_J\rightarrow S_K. \ldots]$. When it is obvious which system to use (for example when only one is specified) and there are no connections between systems $S_G$ as well as the brackets with connections between systems will usually be omitted resulting in $\vec{R}^{(t)}_{j}$.  It should be obvious when this is the case. When normalizing the resulting elements such that their sum equal to one we get the traditional PageRank:
\begin{definition}
$\vec{R}^{(1)}_{G}$ for system $S_G$ is defined as the eigenvector with eigenvalue one to the matrix $\tens{M}_G = c(\tens{A}_G+\vec{g}_G\vec{u}_G^\top)^\top+(1-c)\vec{u}_G\vec{e}^\top$. 
\end{definition}
Note that we always have $||\vec{R}^{(1)}||_1 = 1$ and that non-zero elements in $\vec{R}^{(1)}_{G}$ are all positive. The fact that $||\vec{R}^{(1)}||_1 = 1$ is generally not the case in our other versions of PageRank. When instead setting up the resulting equation system and solving it we get the second definition, the result is multiplied with $n_G$ in order to get multiplication with the one vector in case of uniform $\vec{u}_G$.  
\begin{definition}
$\vec{R}^{(2)}_{G}$ for system $S_G$ is defined as $\vec{R}^{(2)}_{G} = (\tens{I}-c\tens{A}_G^\top)^{-1}n_G\vec{u}_G $ 
\end{definition}
We note that generally $||\vec{R}^{(2)}||_1 \neq 1$ as well as $\vec{R}^{(2)}_{G} \neq n_G \vec{R}^{(1)}_{G}$ unless there are no dangling nodes in the system. However the two versions of PageRank are proportional to each other ($\vec{R}^{(2)}_G \propto \vec{R}^{(1)}_G$). Last we have the third way to define PageRank which we define in order to make it possible to use the power method but still be able to compare PageRank between different subsystems $S_G, S_H,\ldots $ without any additional computations as well as simplifying the work when updating the system.
\begin{definition}
$\vec{R}^{(3)}_{G}$ for system $G$ is defined as:
$$\vec{R}^{(3)}_{G} = \frac{\vec{R}^{(1)}_{G}||\vec{v}_G||_1}{d_G}$$
$$d_G = 1 - \sum{c\tens{A}_G^\top} \vec{R}^{(1)}_{G}$$
where $\vec{v}_G$ is the part of the global weight vector $\vec{V}$ belonging to the nodes in system $S_G$
\cite{CEngstrom:MThesis}.
\end{definition}
A closer look at $\vec{R}^{(3)}_{G}$ is left for a later article. The definition of $\vec{R}^{(3)}_{G}$ is included here only for completeness.  
\begin{definition}
A simple line is a graph with $n_L$ nodes where node $n_L$ links to node $n_{L-1}$ which in turn links to node $n_{L-2}$ all the way until node $n_2$ link to node $n_1$.
\end{definition}
The link matrix $\tens{A}_L$ and graph for system $S_L$ consisting of a simple line with $5$ nodes can be seen in Fig.~\ref{GraphLine}: 
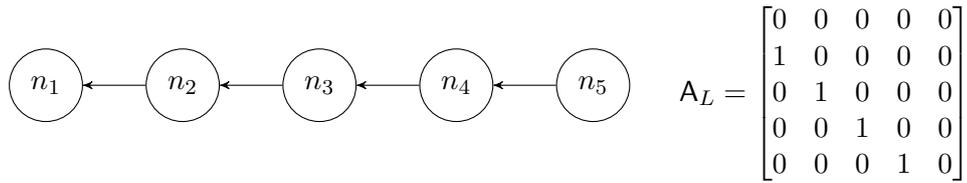
\begin{figure} [!hbt]
\begin{center}
 \begin{tikzpicture}[->,shorten >=0pt,auto,node distance=1.8cm,on grid,>=stealth',
every state/.style={circle,,draw=black}]
\node[state] (A) {$n_1$};
\node[state] (B) [ right=of A] {$n_2$};
\node[state] (C) [right=of B] {$n_3$};
\node[state] (D) [right=of C] {$n_4$};
\node[state] (E) [right=of D] {$n_5$};
\path (B) edge node {} (A)
(C) edge node {} (B)
(D) edge node {} (C)
(E) edge node {} (D);
\node [right=1cm,yshift = 0.1cm, text width=3cm] at (E){
 \begin{equation*} \tens{A}_L = 
\begin{bmatrix}
0 & 0 & 0 & 0 & 0 \\
1 & 0 & 0 & 0 & 0\\
0 & 1 & 0 & 0 & 0\\
0 & 0 & 1 & 0 & 0\\
0 & 0 & 0 & 1 & 0
  \end{bmatrix} 
\end{equation*}
};
\end{tikzpicture}
\end{center}
\caption{The simple line with 5 nodes and corresponding system matrix}
\label{GraphLine}
\end{figure}
\begin{definition}
A complete graph is a group of nodes in which all nodes in the group links to all other nodes in the group. 
\end{definition}
The link matrix $\tens{A}_G$ for system $S_G$ consisting of a complete graph with $5$ nodes can be seen in Fig.~\ref{GraphComplete}
\begin{figure} [!hbt]
\begin{center}
 \begin{tikzpicture}[->,shorten >=0pt,auto,node distance=2.1cm,on grid,>=stealth',
every state/.style={circle,,draw=black,,minimum size=17pt}]
\foreach \name/\angle/\text in {A/162/n_1, B/90/n_2, C/18/n_3, D/-54/n_4, E/-126/n_5}
\node[state, xshift=6cm,yshift=.5cm] (\name) at (\angle:1.5cm) {$\text$};
\foreach \from/\to in {A/B,A/C,A/D,A/E,B/C,B/D,B/E,C/D,C/E,D/E}
\draw [<->] (\from) -- (\to);
\node [right=1cm ,yshift = -0.75cm, text width=4cm] at (C){
 \begin{equation*} \tens{A}_G = \frac{1}{4}
\begin{bmatrix}
0 & 1 & 1 & 1 & 1 \\
1 & 0 & 1 & 1 & 1\\
1 & 1 & 0 & 1 & 1\\
1 & 1 & 1 & 0 & 1\\
1 & 1 & 1 & 1 & 0
  \end{bmatrix} 
\end{equation*}
};
\end{tikzpicture}
\end{center}
\caption{A complete graph with five nodes and corresponding system matrix}
\label{GraphComplete}
\end{figure}
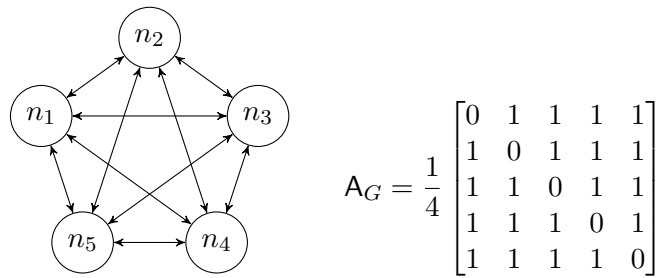

\section{Changes in PageRank when modifying some common structures}
Looking at some simple structures and how PageRank changes as we change them, the goal is to learn something in how and why the rank changes as it does. This in an attempt to answer questions such as: How do I connect my two sites or within my one site in such a way that I won't get any undesired results? In all these examples we will assume uniform $\vec{u}$ (which means we can multiply the inverse $(\tens{I}-c\tens{A}^\top)^{-1}$ with the one vector in order to get $\vec{R}^{(2)}$). 

In this article we will look at two methods two calculate PageRank ($\vec{R}^{(2)}$), while maybe not as useful for calculating PageRank for large systems we use these ways in hope that we can learn something about the behavior of different typical systems or structures within a system.
From earlier we have:
\begin{equation}
\vec{R}^{(1)} = \tens{M}\vec{R}^{(1)}  = (c(\tens{A}+\vec{g}\vec{u}^\top)^\top+(1-c)\vec{u}\vec{e}^\top)\vec{R}^{(1)}
\end{equation}
Calculating the dominant eigenvector $\vec{R}^{(1)}$ is the same as solving the linear system: 
\begin{equation}
\vec{R}^{(1)} = \tens{M}\vec{R}^{(1)} \Leftrightarrow (c\tens{A}^\top-\tens{I})\vec{R}^{(1)} = -(c\vec{u}\vec{g}^\top+(1-c)\vec{u}\vec{e}^\top)\vec{R}^{(1)}
\end{equation}
Since every column of $\vec{u}\vec{g}^\top$ is either equal to $\vec{u}$ or zero and all columns equal to $\vec{u}$ for $\vec{u}\vec{e}^\top$ we can see that $-(c\vec{u}\vec{g}^\top+(1-c)\vec{u}\vec{e}^\top)\vec{R}^{(1)}$ will be proportional to $\vec{u}$. This can be written as: $(c\tens{A}^\top-\tens{I})\vec{R}^{(1)} = k\vec{u}$. 

We choose $k=-n$ in order to get $k\vec{u}$ equal to the one vector in the case of uniform $\vec{u}$, the minus sign to get positive rank and solving the system we get: 
\begin{equation}
\vec{R}^{(2)} = (\tens{I}-c\tens{A}^\top)^{-1}n\vec{u}
\end{equation}
To get the rank to sum to one it is a simple matter of normalizing the result. 
$\vec{R}^{(1)} = \vec{R}^{(2)}/||\vec{R}^{(2)}||_1$ \cite{Bianchini:2005:IP:1052934.1052938}. We note the similarity with this formulation of PageRank (solution to $\vec{R}^{(2)} = c\tens{A}^T\vec{R}^{(2)} +n\vec{u}$) with the one for the potential of a Markov chain with a discounted cost  (solution to $\vec{R}^{(2)} = \alpha \tens{A}\vec{R}^{(2)} +c$), where $0 <\alpha<1$ is the discount factor and $c$ is a cost vector. \cite{NorrisMC} 

Note that we do not need to take any care of the dangling nodes when calculating the PageRank in this way although it is a lot slower than using the Power method or other conventional methods of calculating PageRank. Although we do not need to change $\tens{A}$ for dangling nodes, the result when doing so is changed (but still proportional to $\vec{R}^{(1)}$). We will never change $\tens{A}$ for dangling nodes when solving the linear system and only use the version defined above. 
Note that while solving the equation system is slow it could be possible to get to this non-normalized version of PageRank using another PageRank algorithm, such as using a power series formulation as in \cite{FAndersson:MThesis}. 

The following theorem explains how PageRank ($\vec{R}^{(2)}$) can be computed and how it can be interpret from a probabilistic viewpoint using random random walks on a graph and hitting probabilities. 
\begin{theorem}\label{thm:prProb}
Consider a random walk on a graph described by $c\tens{A}$ described as before. We walk to a new node with probability $c$ and stop with probability $1-c$.

PageRank $\vec{R}^{(2)}$ of a node when using uniform $\vec{u}$ can be written:
\begin{equation}
\vec{R}^{(2)}_j = \left(\sum_{e_i \in S,e_i\neq e_j}{P(e_i \rightarrow e_j)}+1\right) \left(\sum_{k=0}^{\infty}{(P(e_j \rightarrow e_j))^k}\right)
\end{equation}
where $P(e_i \rightarrow e_j)$ is the probability to hit node $e_j$ in a random walk starting in node $e_i$ described as above. This can be seen as the expected number of visits to $e_j$ if we do multiple random walks, starting in every node once.
\end{theorem}
\begin{proof}
$(c\tens{A}^\top)^k_{ij}$ is the probability to be in node $e_i$ starting in node $e_j$ after $k$ steps. Multiplying with the unit vector $\vec{e}$ (vector with all elements equal to one) therefor gives the sum of all the probabilities to be in node $e_i$ after $k$ steps starting in every node once. The expected total number of visits is the sum of all probabilities to be in node $e_i$ for every step starting in every node:
\begin{equation}
\vec{R}^{(2)}_j =\left( \left(\sum_{k=0}^{\infty}{(c\tens{A}^\top)^k}\right ) \vec{e}\right)_j 
\end{equation}
$\sum_{k=0}^{\infty}{(c\tens{A}^\top)^k}$ is the Neumann series of $(\tens{I}-c\tens{A}^\top)^{-1}$ which is guaranteed to converge since $c\tens{A}^\top$ is non-negative and have column sum $<1$. If $\vec{u}$ is uniform we get by the definition:
\begin{equation}
\begin{split}
\vec{R}^{(2)} &= (\tens{I}-c\tens{A}^\top)^{-1}n\vec{u} = (\tens{I}-c\tens{A}^\top)^{-1}\vec{e} = \left(\sum_{k=0}^{\infty}{(c\tens{A}^\top)^k}\right ) \vec{e}\\
 \Rightarrow \vec{R}^2_j &=  \left(\sum_{e_i \in S,e_i\neq e_j}{P(e_i \rightarrow e_j)}+1\right) \left(\sum_{k=0}^{\infty}{(P(e_j \rightarrow e_j))^k}\right) 
\end{split}
\end{equation}
\qed 
\end{proof}

\subsection{Changes in the simple line} \label{SimpleLine_sec}
Using the simple line as defined earlier we recall that we had the link matrix with an image of the system in Fig.~\ref{GraphLine}
\begin{displaymath} \tens{A} = \left[ \begin{array}{ccccc}
0 & 0 & 0 & 0 & 0 \\
1 & 0 & 0 & 0 & 0\\
0 & 1 & 0 & 0 & 0\\
0 & 0 & 1 & 0 & 0\\
0 & 0 & 0 & 1 & 0
\end{array} \right]
\end{displaymath}
By setting up the system of equations we get the inverse $(\tens{I}-c\tens{A}^\top)^{-1}$ as: 
\begin{displaymath} (\tens{I}-c\tens{A}^\top)^{-1} = \left[ \begin{array}{ccccc}
1 & c & c^2 & c^3 & c^4 \\
0 & 1 & c & c^2 & c^3\\
0 & 0 & 1 & c & c^2\\
0 & 0 & 0 & 1 & c \\
0 & 0 & 0 & 0 & 1
\end{array} \right]
\end{displaymath}
Note that this needs only to be multiplied with $n\vec{u}$ or a multiple of $\vec{u}$ for us to get a meaningful ranking. This gives us $\vec{R}^{(2)}$ (for uniform $\vec{u}$):
\begin{displaymath}\vec{R}^{(2)} = [1+c+c^2+c^3+c^4,1+c+c^2+c^3,1+c+c^2,1+c,1]^\top\end{displaymath}
If wanted to get the common normalized ranking $\vec{R}^{(1)}$ we need to normalize the result to sum to one. Looking at the elements $a_{ij}$ of $(\tens{I}-c\tens{A}^\top)^{-1}$ and considering the example with a random walk through the graph, we can see the value of every element $a_{ij}$ as the probability to get from node $e_j$ to node $e_j$. In the case where the link matrix contain nodes with paths back to itself we will later see that it is actually not the probability to get there but the sum of all probabilities to get from $e_j$ to $e_i$ corresponding to Theorem \ref{thm:prProb}. We can motivate this further by looking at the same line but adding a link back from the first node to the second node. 

\subsubsection{The simple line with node one linking to node two}
Letting node one link to node two in the earlier example gives us the graph in Fig.~\ref{GraphLine2}.
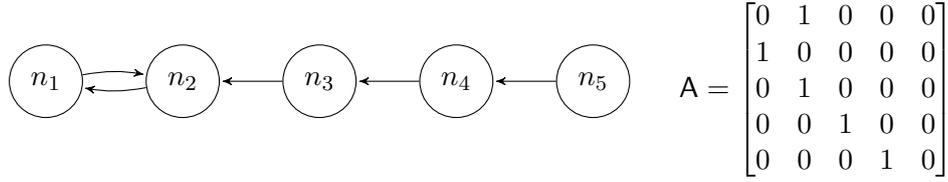
\begin{figure} [!hbt]
\begin{center}
 \begin{tikzpicture}[->,shorten >=1pt,auto,node distance=1.8cm,on grid,>=stealth',
every state/.style={circle,,draw=black}, bend angle=10]
\node[state] (A) {$n_1$};
\node[state] (B) [ right=of A] {$n_2$};
\node[state] (C) [right=of B] {$n_3$};
\node[state] (D) [right=of C] {$n_4$};
\node[state] (E) [right=of D] {$n_5$};
\path (A) edge [bend left] node {} (B)
(B) edge [bend left] node {} (A)
(C) edge node {} (B)
(D) edge node {} (C)
(E) edge node {} (D);
\node [right=1cm,yshift = 0.1cm ,text width=3cm] at (E){
 \begin{equation*} \tens{A} = 
\begin{bmatrix}
0 & 1 & 0 & 0 & 0 \\
1 & 0 & 0 & 0 & 0\\
0 & 1 & 0 & 0 & 0\\
0 & 0 & 1 & 0 & 0\\
0 & 0 & 0 & 1 & 0
  \end{bmatrix} 
\end{equation*}
};
\end{tikzpicture}
\end{center}
\caption{Simple line where the first node links to the second and corresponding system matrix}
\label{GraphLine2}
\end{figure}
The resulting inverse can be written: 
\begin{displaymath} (\tens{I}-c\tens{A}^\top)^{-1} = \left[ \begin{array}{ccccc}
s & sc & sc^2 & sc^3 & sc^4 \\
sc & s & sc & sc^2 & sc^3\\
0 & 0 & 1 & c & c^2\\
0 & 0 & 0 & 1 & c \\
0 & 0 & 0 & 0 & 1
\end{array} \right]
\end{displaymath}
Where $s = \sum_{k=0}^\infty{c^{2k}} = \frac{1}{1-c^2}$ is the sum of all the probabilities of getting from node $1$ or $2$ back to itself.  

From this we can see that the following observations seem to be true.
\begin{itemize}
\item[(i)]	The sum of a column $c_j$ is at most $\sum_{k=0}^\infty{c^k} = \frac{1}{1-c}$ when using uniform $\vec{u}$, with equality if there are no paths to any dangling node from node $j$ and node $j$ is not a dangling node itself. 
\item[(ii)] A diagonal element is equal to one if the node have no paths leading back to itself.
\item[(iii)]	Setting one element in $u_i$ to zero only effects the influence of a random walk starting in the corresponding node. 			
\item[(iiii)]	Every non zero element in the same row can be written as the diagonal element on the same line times the sum of probabilities of getting from all other nodes to the node corresponding to the current line.  			
\item[(iiiii)] Each element $e_{ij}$ of $(\tens{I}-c\tens{A}^\top)^{-1}$ contains the sum of probabilities of all paths starting in node $j$ and ending in node $i$. When doing a random walk by choosing a random link with probability $c$ and stopping with probability $1-c$.   
\end{itemize}
Which is consistent with the statement that the normalized PageRank $\vec{R}^{(1)}_j$ of a node is the probability that a surfer that starts in a random node (page) and keeps clicking links with probability $c$ or starts at a new random page with probability (1-c) is in a given node. However here we can explicitly see all the probabilities and their influence on the ranking. \cite{A-25bilEig}

\subsubsection{Removing a link between two nodes}
When removing a link between two nodes in the simple line we end up with two smaller disjoint lines instead. We note that these could be calculated separately and we would still have the same relation between them. This is interesting since when using the "Power method" or straight calculating $\vec{R}^{(1)}$ this is not possible since more nodes in a system obviously means a lower mean rank since we in that case normalize the result to one.

Especially in the inverse $(\tens{I}-c\tens{A}^\top)^{-1}$ we see that when we remove one link, we remove all the elements in the upper right corresponding to paths from nodes above the removed link to all the ones below it. An example of what the new inverse looks like when removing the link between the third node and the second node in Fig.~\ref{GraphLine} can be seen below:
\begin{displaymath} (\tens{I}-c\tens{A}^\top)^{-1} = \left[ \begin{array}{ccccc}
1 & c & 0 & 0 & 0 \\
0 & 1 & 0 & 0 & 0\\
0 & 0 & 1 & c & c^2\\
0 & 0 & 0 & 1 & c \\
0 & 0 & 0 & 0 & 1
\end{array} \right]
\end{displaymath}
With PageRank: $\vec{R}^{(2)}= [1+c, 1, 1+c+c^2,1+c,1]$ and normalizing constant $N = 5+3c+c^2$, when using a uniform $\vec{u}$

\subsubsection{Adding a new node pointing at one node in the simple line}
A more interesting example is when looking at what happens when we add a single new node, linking to one other node in the simple line. 
Since we make no changes in the line that part of the inverse will stay the same. We will however add a new row and column. The non diagonal element of the new column can be found immediately as $c$ times the column corresponding to the node our new node links to. This since we got the probability $c$ to get to that node instead of $1$ when we start in it. At last we need to add the one at the new element in the diagonal. An example of what the inverse looks like after adding a new node pointing at node $3$ in Fig.~\ref{GraphLine} can be seen below.

\begin{displaymath} (\tens{I}-c\tens{A}^\top)^{-1} = \left[ \begin{array}{cccccc}
1 & c & c^2 & c^3 & c^4 & c^3 \\
0 & 1 & c & c^2 & c^3 & c^2\\
0 & 0 & 1 & c & c^2 & c\\
0 & 0 & 0 & 1 & c  & 0\\
0 & 0 & 0 & 0 & 1 & 0 \\
0 & 0 & 0 & 0 & 0 & 1 
\end{array} \right]
\end{displaymath}
From this easy example we can immediately get an expression for the PageRank of a simple line with one or more added nodes linking to any of the nodes in the simple line. 
\begin{theorem}\label{thm:line_add_node}
The PageRank of a node $e_i$ belonging to the line in a system containing a simple line with one outside node linking to one of the nodes in the line when using uniform weight vector $\vec{u}$ can be written: 
\begin{equation}
\begin{split}
\vec{R}^{(2)}_{i} &= \sum_{k=0}^{n_L-i}{c^k}+b_{ij} = \frac{1-c^{n_L-i+1}}{1-c}+b_{ij}\\
b_{ij} &= \left\{ \begin{split} c^{j+1-i},~  j \geq i \\
0,~  j < i \end{split} \right. 
\end{split}
\end{equation}
where $n_L$ is the number of nodes in the line and the new node link to node $j$. The new node has rank $1$. After normalization we get the PageRank of node $i$ as: 
\begin{equation}
\vec{R}^{(1)}_i= \frac{\frac{1-c^{n_L-i+1}}{1-c}+b_{ij}}{n_L+1+(n_L-1)c+(n_L-2)c^2+\ldots + c^{n_L-1}+\frac{1-c^j}{1-c}}
\end{equation}
where $\vec{R}^{(1)}_i,\vec{R}^{(2)}_i$ is the PageRank of one of the nodes in the original line, $L$ is the number of nodes on the line, $j$ is the number of the node linked to be the new node.

Additionally adding new nodes linking to the line means adding additional $b_{ij}$ parts and adding the corresponding part $\frac{1-c^j}{1-c}$ to the normalizing constant. 
\end{theorem}

\begin{proof}
From Theorem~\ref{thm:prProb} PageRank for a node when using uniform $\vec{u}$ can be written as: 
$$\vec{R}^{(2)}_i = \left(\sum_{e_j \in S,e_j\neq e_i}{P(e_j \rightarrow e_i)}+1\right) \left(\sum_{k=0}^{\infty}{(P(e_i \rightarrow e_i))^k}\right)$$
where $P(e_j \rightarrow e_i)$ is the probability to hit node $e_i$ starting in node $e_j$. When we consider a random walk on a graph given by $c\tens{A}$ described as before. We walk to a new node with probability $c$ and stop with probability $1-c$.

The probability of getting to any node $e_i$ in the line from any other node $e_j$ in the line once is: 
\begin{equation}
P(e_j\rightarrow e_i) = c^{j-i}, \quad j > i 
\end{equation}
and zero otherwise. Summation over all $j > i$ gives
\begin{equation}
\sum_{e_j \in S,e_j\neq e_i}{P(e_j \rightarrow e_i)}+1=\sum_{k=1}^{n_L-i}{c^k}+1 = \frac{1-c^{n_L-i+1}}{1-c}
\end{equation}
where $L$ is the number of nodes in the line. With the first part shown we need to show that the single outside node linking to node $e_j$ adds $b_{ij} =c^{j+1-i},  ~ j \ge i$. We get this probability in the same way by instead looking at the line created by the first $j$ nodes plus the extra node added linking to node $j$. We get the probability to reach node $e_j$ as $c$ and then $c^2$ for the next and so on. If $ i > j$, $e_i$ does not belong to this line, and we obviously cannot reach it from $e_j$ hence $b_{ij} = 0, ~ i>j$.

Last the PageRank of the "outside" node linking to a node in the line is obviously $1$ since no node links to it. 
The normalized PageRank is found by dividing $\vec{R}^{(2)}$ with $||\vec{R}^{(2)}||_1$.
\qed 
\end{proof}

We also give a proof using matrices but first we will need the following lemma for blockwise inversion used repeatedly throughout the article.. We note that we label the blocks from $\tens{B}$ to $\tens{E}$ rather than from $\tens{A}$ to $\tens{D}$ in order to avoid confusion with the system matrix $\tens{A}$. 

\begin{lemma}\label{lem:blockwise}
\begin{equation}
\begin{bmatrix}
\tens{B}& \tens{C}\\
\tens{D}& \tens{E} \end{bmatrix}^{-1} =  \begin{bmatrix}
(\tens{B}-\tens{C}\tens{E}^{-1}\tens{D})^{-1} & -(\tens{B}-\tens{C}\tens{E}^{-1}\tens{D})^{-1}\tens{C}\tens{E}^{-1} \\
-\tens{E}^{-1}\tens{D}(\tens{B}-\tens{C}\tens{E}^{-1}\tens{D})^{-1} &~~~ \tens{E}^{-1}+\tens{E}^{-1}\tens{D}(\tens{B}-\tens{C}\tens{E}^{-1}\tens{D})^{-1}\tens{C}\tens{E}^{-1} 
\end{bmatrix}
\end{equation}
where $\tens{B},\tens{E}$ is square and $\tens{E},(\tens{B}-\tens{C}\tens{E}^{-1}\tens{D})$ are nonsingular.
\end{lemma}

\begin{proof}
To prove the Lemma it is enough to show that:
\begin{equation}
 \begin{bmatrix}
\tens{B}& \tens{C}\\
\tens{D}& \tens{E} \end{bmatrix}  \begin{bmatrix}
(\tens{B}-\tens{C}\tens{E}^{-1}\tens{D})^{-1} & -(\tens{B}-\tens{C}\tens{E}^{-1}\tens{D})^{-1}\tens{C}\tens{E}^{-1} \\
-\tens{E}^{-1}\tens{D}(\tens{B}-\tens{C}\tens{E}^{-1}\tens{D})^{-1} &~~~ \tens{E}^{-1}+\tens{E}^{-1}\tens{D}(\tens{B}-\tens{C}\tens{E}^{-1}\tens{D})^{-1}\tens{C}\tens{E}^{-1} 
\end{bmatrix}= I 
\end{equation}
Looking at the result blockwise we get: 
\begin{equation}
\begin{split}
\tens{B}(\tens{B}-\tens{C}\tens{E}^{-1}\tens{D})^{-1} -\tens{C}\tens{E}^{-1}\tens{D}(\tens{B}-\tens{C}\tens{E}^{-1}\tens{D})^{-1}\\
= (\tens{B}-\tens{C}\tens{E}^{-1}\tens{D})(\tens{B}-\tens{C}\tens{E}^{-1}\tens{D})^{-1} = \tens{I}
\end{split}
\end{equation}
\begin{equation}
\begin{split}
-\tens{B}(\tens{B}-\tens{C}\tens{E}^{-1}\tens{D})^{-1}\tens{C}\tens{E}^{-1} +\tens{C}(\tens{E}^{-1}+\tens{E}^{-1}\tens{D}(\tens{B}-\tens{C}\tens{E}^{-1}\tens{D})^{-1}\tens{C}\tens{E}^{-1})  \\
= \tens{C}\tens{E}^{-1} - (\tens{B}-\tens{C}\tens{E}^{-1}\tens{D})(\tens{B}-\tens{C}\tens{E}^{-1}\tens{D})^{-1}\tens{C}\tens{E}^{-1} = \tens{C}\tens{E}^{-1} -\tens{I}\tens{C}\tens{E}^{-1} = \tens{0}
\end{split}
\end{equation}
\begin{equation}
\begin{split}
\tens{D}(\tens{B}-\tens{C}\tens{E}^{-1}\tens{D})^{-1} -\tens{E}\tens{E}^{-1}\tens{D}(\tens{B}-\tens{C}\tens{E}^{-1}\tens{D})^{-1} \\
= \tens{D}(\tens{B}-\tens{C}\tens{E}^{-1}\tens{D})^{-1}-\tens{D}(\tens{B}-\tens{C}\tens{E}^{-1}\tens{D})^{-1} = \tens{0}
\end{split}
\end{equation}
\begin{equation}
\begin{split}
-\tens{D}(\tens{B}-\tens{C}\tens{E}^{-1}\tens{D})^{-1}\tens{C}\tens{E}^{-1} +\tens{E}(\tens{E}^{-1}+\tens{E}^{-1}\tens{D}(\tens{B}-\tens{C}\tens{E}^{-1}\tens{D})^{-1}\tens{C}\tens{E}^{-1})\\
= -\tens{D}(\tens{B}-\tens{C}\tens{E}^{-1}\tens{D})^{-1}\tens{C}\tens{E}^{-1} + \tens{I} + \tens{D}(\tens{B}-\tens{C}\tens{E}^{-1}\tens{D})^{-1}\tens{C}\tens{E}^{-1} = \tens{I}
\end{split}
\end{equation}
This gives:
\begin{equation}
\begin{split}
\left[  \begin{array}{cc}
\tens{B}& \tens{C}\\
\tens{D}& \tens{E} \end{array} \right]  \left[ \begin{array}{cc}
(\tens{B}-\tens{C}\tens{E}^{-1}\tens{D})^{-1} & -(\tens{B}-\tens{C}\tens{E}^{-1}\tens{D})^{-1}\tens{C}\tens{E}^{-1} \\
-\tens{E}^{-1}\tens{D}(\tens{B}-\tens{C}\tens{E}^{-1}\tens{D})^{-1} &~~~ \tens{E}^{-1}+\tens{E}^{-1}\tens{D}(\tens{B}-\tens{C}\tens{E}^{-1}\tens{D})^{-1}\tens{C}\tens{E}^{-1} 
\end{array} \right]  \\
= \left[\begin{array}{cc}
\tens{I}& \tens{0}\\
\tens{0}& \tens{I} \end{array} \right] = \tens{I} 
\end{split}
\end{equation}
Furthermore we need that $\tens{E}$ and $(\tens{B}-\tens{C}\tens{E}^{-1}\tens{D})$ is nonsingular in order for the matrix to be invertible. \cite{B_Dennis_Matrix}
\qed 
\end{proof}
When using Lemma \ref{lem:blockwise} we will denote the individual blocks off the inverse matrix as described in Definition \ref{def:blockInv}
\begin{definition}\label{def:blockInv}
Given a block matrix $\tens{M}$  we denote the inverse as:
\begin{equation} \tens{M}  = \begin{bmatrix}
\tens{M}_{1,1} & \tens{M}_{1,2} & \ldots & \tens{M}_{1,n} \\
\tens{M}_{2,1} & \tens{M}_{2,2} & \ldots & \tens{M}_{2,n} \\
\vdots & \vdots & \ddots &\vdots \\
\tens{M}_{m,1} & \tens{M}_{m,2} & \ldots & \tens{M}_{m,n} \\
\end{bmatrix}, ~~ \tens{M}^{-1} = \begin{bmatrix}
\tens{M}_{1,1}^{\text{inv}} & \tens{M}_{1,2}^{\text{inv}} & \ldots & \tens{M}_{1,n}^{\text{inv}} \\
\tens{M}_{2,1}^{\text{inv}} & \tens{M}_{2,2}^{\text{inv}} & \ldots & \tens{M}_{2,n}^{\text{inv}} \\
\vdots & \vdots & \ddots &\vdots \\
\tens{M}_{m,1}^{\text{inv}} & \tens{M}_{m,2}^{\text{inv}} & \ldots & \tens{M}_{m,n}^{\text{inv}} \\
\end{bmatrix}
\end{equation}
\end{definition}
We can now give a matrix proof of Theorem \ref{thm:line_add_node} as well.
\begin{proof}[Proof of Theorem \ref{thm:line_add_node}]
We let $\tens{B}$ be the part of the matrix $(\tens{I}-c\tens{A}^\top) $ corresponding to the nodes in the line which gives: 
\begin{equation}(\tens{I}-c\tens{A}^\top) = \begin{bmatrix}
\tens{B} & \tens{C} \\
\tens{0} & \tens{1}
\end{bmatrix} 
\end{equation}
We write 
\begin{equation}
(\tens{I}-c\tens{A}^\top)^{-1} = \begin{bmatrix}
\tens{B}^{\text{inv}} & \tens{C}^{\text{inv}} \\
\tens{D}^{\text{inv}} & \tens{E}^{\text{inv}}
\end{bmatrix}
\end{equation}
Using Lemma \ref{lem:blockwise} for blockwise inverse we get $\tens{B}^{\text{inv}} = (\tens{B}-\tens{C}\tens{E}^{-1}\tens{D})^{-1} = \tens{B}^{-1}$. Since $\tens{B}$ is the matrix for the simple line found earlier we get: 
\begin{equation}
\tens{B}^{\text{inv}} =  (\tens{I}-c\tens{A}^\top)^{-1} = \left[ \begin{array}{ccccc}
1 & c & c^2 & \ldots & c^{L-1}\\
0 & 1 & c & \ldots & c^{L-2}\\
0 & 0 & 1 & \ldots & c^{L-3}\\
\vdots & \vdots & \ddots & \ddots & \vdots \\
0 & \ldots & \ldots & 0 & 1
\end{array} \right]
\end{equation}
where $L$ is the total number of nodes in the line. $\tens{C} = [0 \ \ldots \ c \ 0 \ldots \ 0]^\top$ where the non-zero element $c$ is at position $j$ gives:
\begin{equation}
\tens{C}^{\text{inv}} = -\tens{B}^{\text{inv}}\tens{C}\tens{E}^{-1} =  -\tens{B}^{\text{inv}}\tens{C} = [c^j \ c^{j-1} \ \ldots \ c \ 0 \ldots \ 0]^\top
\end{equation}
Last, since $\tens{D} = \tens{0}$ we get $\tens{D}^{\text{inv}}=\tens{0},\ \tens{E}^{\text{inv}} = 1$. Since the weight vector $\vec{u}$ is uniform we get the PageRank of a node as the sum of corresponding row in $(\tens{I}-c\tens{A}^\top)^{-1}$. For the nodes in the line we get PageRank:
\begin{equation}
\begin{split}
\vec{R}^{(2)}_{i} &= \sum_{k=0}^{n_L-i}{c^k}+b_{ij} = \frac{1-c^{n_L-i+1}}{1-c}+b_{ij} \\
b_{ij} &= \left\{ \begin{split}
c^{j+1-i},~j \geq i\\
0,~j < i
\end{split} \right.
\end{split}
\end{equation}
where the sum is the sum of the first $n_L$ values and $b_{ij}$ is the value on the last column. For the last row we obviously get sum $1$. 

We get the normalized PageRank $\vec{R}^{(1)}$ by dividing $\vec{R}^{(1)} = \vec{R}^{(2)}/||\vec{R}^{(2)} ||_1$. 

\qed 
\end{proof}

\subsection{Changes in a complete graph}\label{CCompGraph}
Complete graphs or similar structures are common both as parts of a site and as a way between different sites to try and gain a better rank. An image of a complete graph with five nodes can be seen in Fig.~\ref{GraphComplete}. We recall that the system matrix for this system is: 
\begin{displaymath} \tens{A} = \frac{1}{4}\left[ \begin{array}{ccccc}
0 & 1 & 1 & 1 & 1 \\
1 & 0 & 1 & 1 & 1\\
1 & 1 & 0 & 1 & 1\\
1 & 1 & 1 & 0 & 1\\
1 & 1 & 1 & 1 & 0
\end{array} \right]\end{displaymath}
Using this we get the inverse of this system as: 

\begin{displaymath} (\tens{I}-c\tens{A}^\top)^{-1} = \left[ \begin{array}{ccccc}
\frac{3c-4}{c^2+3c-4} & \frac{-c}{c^2+3c-4} & \frac{-c}{c^2+3c-4} &\frac{-c}{c^2+3c-4} &\frac{-c}{c^2+3c-4} \\
\frac{-c}{c^2+3c-4} & \frac{3c-4}{c^2+3c-4} & \frac{-c}{c^2+3c-4} &\frac{-c}{c^2+3c-4} &\frac{-c}{c^2+3c-4} \\
\frac{-c}{c^2+3c-4} & \frac{-c}{c^2+3c-4} & \frac{3c-4}{c^2+3c-4} &\frac{-c}{c^2+3c-4} &\frac{-c}{c^2+3c-4} \\
\frac{-c}{c^2+3c-4} & \frac{-c}{c^2+3c-4} & \frac{-c}{c^2+3c-4} &\frac{3c-4}{c^2+3c-4} &\frac{-c}{c^2+3c-4} \\
\frac{-c}{c^2+3c-4} & \frac{-c}{c^2+3c-4} & \frac{-c}{c^2+3c-4} &\frac{-c}{c^2+3c-4} &\frac{3c-4}{c^2+3c-4} \\
\end{array} \right]
\end{displaymath}
After normalization we will obviously end up with $\vec{R}^{(1)}_i =1/5$ as PageRank for every node $i$. However since there is not any dangling nodes in the complete graph all the nodes will have maximum influence on the PageRank of the system. Additionally since they only point to each other they will not share any of it with the outside in the case of a bigger link matrix with a part of it being a complete graph. This makes a complete graph similar to a dangling node in that it will not increase the rank of anyone else, but with the addition of having a higher rank in itself since it can increase its own rank to a certain extent.   

Trying to find an expression for the elements in the inverse $(\tens{I}-c\tens{A}^\top)^{-1}$ for the complete graph we formulate the following lemma: 

\begin{lemma}\label{lem:complete}
The diagonal element $a_d$ of the inverse  $(\tens{I}-c\tens{A}^\top)^{-1}$ of the complete graph with $n$ nodes is: 
\begin{equation}
 a_{d} = \frac{(n-1)-c(n-2)}{(n-1)-c(n-2)-c^2}
\end{equation}
The non diagonal elements $a_{ij}$ can be written as: 
\begin{equation}
a_{ij} = \frac{c}{(n-1)-c(n-2)-c^2}
\end{equation}
\end{lemma}

\begin{proof}
The diagonal element is the sum of the probabilities of all paths to node $e_d$ from itself. This can be written as a geometric sum: ${a_d = \sum_{k=0}^{\infty}{P(e_d\rightarrow e_d)}^k}$, where $P(e_d\rightarrow e_d)$ is the probability of getting from node $e_d$ to node $e_d$. The probability $P(e_d\rightarrow e_d)$ can be written as: 
\begin{equation}
\begin{split}
P(e_d\rightarrow e_d) &= \frac{c^2}{n-1}+\frac{c^3(n-2)}{(n-1)^2}+\frac{c^4(n-2)^2}{(n-1)^3}+\ldots \\
&= \frac{c^2}{n-1} \sum_{k=0}^{\infty}{\left( \frac{c(n-2)}{n-1}\right)^k} = \frac{c^2}{(n-1)-c(n-2)}
\end{split}
\end{equation}
This gives: 
\begin{equation}
a_d = \sum_{k=0}^{\infty}{\left(\frac{c^2}{(n-1)-c(n-2)}\right)^k} = \frac{(n-1)-c(n-2)}{(n-1)-c(n-2)-c^2} 
\end{equation}
For non-diagonal elements $e_{ij}$ we get $e_{ij}=P(e_i\rightarrow e_j)a_d$, where $P(e_i\rightarrow e_j)$ is the probability of getting from node $e_i$ to node $e_j$ where $e_i \neq e_j$. This probability can be written as: 
\begin{equation}
\begin{split}
P(e_i\rightarrow e_j) &= \frac{c}{n-1}+\frac{c^2(n-2)}{(n-1)^2}+\frac{c^3(n-2)^2}{(n-1)^3}+\ldots \\
&= \frac{c}{n-1} \sum_{k=0}^{\infty}{\left( \frac{c(n-2)}{n-1}\right)^k} =  \frac{c}{(n-1)-c(n-2)}
\end{split}
\end{equation}
This gives:
\begin{equation}
a_{ij} = \frac{c}{(n-1)-c(n-2)} \frac{(n-1)-c(n-2)}{(n-1)-c(n-2)-c^2} = \frac{c}{(n-1)-c(n-2)-c^2}
\end{equation}
\qed 
\end{proof}

We give a matrix proof of Lemma~\ref{lem:complete} as well:

\begin{proof}[Proof of Lemma~\ref{lem:complete}]
We consider a general matrix $\tens{A}$ of the form: 
\begin{displaymath}
\tens{A} = \begin{bmatrix}
1 & a & a & \ldots & a\\
a & 1 & a & \ldots & a\\
a & a & 1 & \ddots & a\\
\vdots & \vdots & \ddots & \ddots & \vdots \\
a & a & a & \ldots & 1\\
\end{bmatrix} 
\end{displaymath}
We use Gauss-Jordan elimination to find the inverse $\tens{A}^{-1}$: 
\begin{displaymath}
\begin{bmatrix}
1 & a & a & \ldots & a & 1 & 0 & 0 & \ldots & 0 \\
a & 1 & a & \ldots & a & 0 & 1 & 0 & \ldots & 0 \\
a & a & 1 & \ddots & a &0 & 0 & 1 & \ddots & 0 \\
\vdots & \vdots & \ddots & \ddots & \vdots &  \vdots & \vdots & \ddots & \ddots & 0 \\
a & a & a & \ldots & 1 & 0 & 0 & 0 & \ldots & 1
\end{bmatrix}
 \end{displaymath}
We add $-a r_1$ where $r_1$ is the first row to every other row to eliminate the elements below $1$ on the first column. 
\begin{displaymath}
\begin{bmatrix}
1 & a & a & \ldots & a & 1 & 0 & 0 & \ldots & 0 \\
0 & 1-a^2 & a-a^2  & \ldots & a-a^2  & -a & 1 & 0 & \ldots & 0 \\
0 & a-a^2 & 1-a^2 & \ddots & a-a^2  &-a & 0 & 1 & \ddots & 0 \\
\vdots & \vdots & \ddots & \ddots & \vdots &  \vdots & \vdots & \ddots & \ddots & 0 \\
0 & a-a^2 & a-a^2 & \ldots & 1-a^2 & -a & 0 & 0 & \ldots & 1
\end{bmatrix} \end{displaymath}
Next we eliminate the values to the right of the $1$ on the first row. We add $-k\sum_{i=2}^{n}{r_i}$, where $r_i$ is row $i$ to the first row giving the equation: 
\begin{equation}
\begin{split}
a = -k (1-a^2 + (n-2)(a-a^2))\\
\Rightarrow k = \frac{-a}{ (1-a^2 + (n-2)(a-a^2))}
\end{split}
\end{equation}
This gives:
$$\begin{bmatrix}
1 & 0 & 0 & \ldots & 0 & 1-(n-1)ak & k & k & \ldots & k \\
0 & 1-a^2 & a-a^2  & \ldots & a-a^2  & -a & 1 & 0 & \ldots & 0 \\
0 & a-a^2 & 1-a^2 & \ddots & a-a^2  &-a & 0 & 1 & \ddots & 0 \\
\vdots & \vdots & \ddots & \ddots & \vdots &  \vdots & \vdots & \ddots & \ddots & 0 \\
0 & a-a^2 & a-a^2 & \ldots & 1-a^2 & -a & 0 & 0 & \ldots & 1
\end{bmatrix} $$  

We are now done calculating the first row of the inverse $\tens{A}^{-1}$. We get the other rows using the same calculations if we start with another pivot element. For the inverse matrix we get diagonal elements $d = 1-(n-1)ak$ and for all other elements $e = k$, where $n$ is the total number of rows giving a inverse like below: 
\begin{displaymath}
\tens{A}^{-1} = \begin{bmatrix}
 1-(n-1)ak & k & k & \ldots & k \\
 k & 1-(n-1)ak & k & \ldots & k \\
 k & k & 1-(n-1)ak & \ddots & \vdots \\
  \vdots & \vdots & \ddots & \ddots & k \\
k& k & \ldots  & k&  1-(n-1)ak  
\end{bmatrix} 
\end{displaymath} 
Calculating for $a = -c/(n-1)$ as for a complete graph gives: 
\begin{equation}
k = \frac{-a}{ (1-a^2 + (n-2)(a-a^2))} = \frac{c}{ (n-1)-(n-2)c-c^2}
\end{equation}
\begin{equation}
\begin{split}
d &= 1-(n-1)ak = \frac{ (n-1)-(n-2)c-c^2 -(n-1)(-c)/(n-1)c}{ (n-1)-(n-2)c-c^2}\\
 &=\frac{(n-1)-(n-2)c}{ (n-1)-(n-2)c-c^2}
\end{split}
\end{equation}
And the proof is complete.
\qed 
\end{proof}

Using this we immediately get the PageRank (before normalization) of elements in a complete graph with uniform $\vec{u}$:

\begin{theorem} \label{thm:complete}
Given a complete graph with $n>1$ nodes, PageRank $\vec{R}^{(2)}$ before normalization can be written as:
\begin{equation}
\vec{R}^{(2)}_i = \frac{1}{1-c}
\end{equation}
\end{theorem}

\begin{proof}
From Lemma \ref{lem:complete} We already have the inverse $(\tens{I}-c\tens{A}^\top)^{-1}$, We then find the PageRank by summation of any row of the matrix (since all rows have equal sum). 
\begin{equation}
\begin{split}
\vec{R}^{(2)}_i = a_d+(n-1)a_{ij}, \quad i \neq j \\
=\frac{(n-1)-c(n-2)+c(n-1)}{(n-1)-c(n-2)-c^2} = \frac{c+(n-1)}{(n-1)-c(n-2)-c^2} = \frac{1}{1-c}
\end{split}
\end{equation}
\qed 
\end{proof}

We do note that since we have no dangling nodes all the probability from a node in the complete graph is distributed within the complete graph. Also the size of the graph is irrelevant for the individual nodes as long as none are linked to from outside sources and it consists of at least two nodes. In the $\vec{R}^{(1)}$ sense the size obviously changes the result since we would increase the overall number of nodes in the system by increasing the size of the complete graph. Two things is important to note however: The higher ones own PageRank before joining the complete graph (probability of getting there from outside nodes) the more gain there is by joining a small complete graph in order to maximize the probability of returning to itself. In the same way if a node have a very low rank it gains much by joining a large complete graph of nodes with higher rank than itself.

\subsubsection{Adding a link out of a complete graph}

If we want to see how the complete graph changes when adding one link from one node (node one) out of the complete graph we end up with the following system matrix for the nodes in the complete graph: 

\begin{displaymath} (\tens{I}-c\tens{A}^\top) = \left[ \begin{array}{ccccc}
1 & -c/4 & -c/4 & -c/4 & -c/4 \\
-c/5& 1& -c/4 & -c/4 & -c/4 \\
-c/5& -c/4 & 1 & -c/4 &-c/4 \\
-c/5& -c/4 & -c/4 & 1 & -c/4 \\
-c/5& -c/4 & -c/4 & -c/4 & 1 \\
\end{array} \right]
\end{displaymath}
After taking the inverse and multiplying with $-1$ we get: 
\begin{displaymath} (\tens{I}-c\tens{A}^\top)^{-1} = \end{displaymath}
\begin{displaymath}\left[ \begin{array}{ccccc}
\frac{15c-20}{s} & \frac{-5c}{s} & \frac{-5c}{s} &\frac{-5c}{s} &\frac{-5c}{s} \\
\frac{-4c}{s} &\frac{12c^2+40c-80}{(c+4)s} & -\frac{4c(5+c)}{(c+4)s} &-\frac{4c(5+c)}{(c+4)s}&-\frac{4c(5+c)}{(c+4)s} \\
\frac{-4c}{s} &-\frac{4c(5+c)}{(c+4)s} & \frac{12c^2+40c-80}{(c+4)s} &-\frac{4c(5+c)}{(c+4)s} &-\frac{4c(5+c)}{(c+4)s} \\
\frac{-4c}{s} &-\frac{4c(5+c)}{(c+4)s} &-\frac{4c(5+c)}{(c+4)s} &\frac{12c^2+40c-80}{(c+4)s} &-\frac{4c(5+c)}{(c+4)s} \\
\frac{-4c}{s} &-\frac{4c(5+c)}{(c+4)s} &-\frac{4c(5+c)}{(c+4)s} &-\frac{4c(5+c)}{(c+4)s} &\frac{12c^2+40c-80}{(c+4)s} \\
\end{array} \right]
\end{displaymath}
where $s = 4c^2+15c-20$

We find the expression for the PageRank in a complete graph  with one node linking out to be the following assuming uniform $\vec{u}$. 

\begin{theorem}\label{thm:complete_out}
The PageRank of the nodes in a complete graph with the first node linking out of the complete graph, the PageRank can be written as:
\begin{equation}\label{eq:complete_out_1}
\vec{R}^{(2)}_1 = \frac{n(n-1)+nc}{n(n-1)-(n-1)c^2-n(n-2)c}
\end{equation}
\begin{equation}
\vec{R}^{(2)}_i = {\frac { \left( c+n \right)  \left( n-1 \right) }{n(n-1)-(n-1)c^2-n(n-2)c }},\quad n \geq i > 1 
\end{equation}
where $n$ is the number of nodes in the complete graph and node one links out of the complete graph.
\end{theorem}

\begin{proof}
We start by looking at the PageRank as a probability, we let $e_1$ be the node linking out. The probability to get from $e_1$ back to itself is:
\begin{equation}
\begin{split}
P(e_1 \rightarrow e_1 ) &= \frac{c(n-1)}{n}\frac{c}{n-1}+\frac{c(n-1)}{n}\frac{c}{n-1}\frac{c(n-2)}{n-1}\\
 &~~+ \frac{c(n-1)}{n}\frac{c}{n-1}\left( \frac{c(n-2)}{n-1}\right) ^2 + \ldots \\
&= \frac{c^2}{n}\sum_{k = 0}^\infty{\left( \frac{c(n-2)}{n-1} \right)^k} = \frac{c^2}{n}\frac{n-1}{(n-1)-c(n-2)}
\end{split}
\end{equation}
And we get the sum of all probabilities from $e_1$ back to itself as: 
\begin{equation}
\begin{split}
\sum_{k = 0}^\infty{(P(e_1 \rightarrow e_1 ))^k } &= \sum_{k = 0}^\infty{\left(\frac{c^2}{n}\frac{n-1}{(n-1)-c(n-2)}\right)^k} \\
&= \frac{n((n-1)-c(n-2))}{n((n-1)-c(n-2))-c^2(n-1)} = \tens{B}^{inv}
\end{split}
\end{equation}
We remember that on the diagonal of  $(\tens{I}-c\tens{A}^\top)$, we have the sums of probabilities of nodes going back to themselves. So if we divide the matrix $(\tens{I}-c\tens{A}^\top)$ in blocks:
$$(\tens{I}-c\tens{A}^\top) = \begin{bmatrix}
B & C \\
D & E
\end{bmatrix} $$
And inverse matrix: 
$$(\tens{I}-c\tens{A}^\top)^{-1} = \begin{bmatrix}
\tens{B}^{inv} & \tens{C}^{inv} \\
\tens{D}^{inv} & \tens{E}^{inv}
\end{bmatrix} $$
We note that $\tens{B}^{inv}$ is not the inverse of $\tens{B}$ but the part of the inverse $(\tens{I}-c\tens{A}^\top)^{-1}$ corresponding to block $\tens{B}$. We let  $\tens{B} = [1]$ corresponding to the node linking out and we get $\tens{B}^{inv}$ as above. 

For the elements $C^{inv}_i, i \neq 1$ of $\tens{C}^{inv} $ we find them as 
\begin{equation}
\begin{split}
C^{inv}_i &=  \sum_{k = 0}^\infty{(P(e_i \rightarrow e_1 ))^k }\sum_{k = 0}^\infty{(P(e_1 \rightarrow e_1 ))^k } \\
 &= \frac{c}{n-1}\sum_{k = 0}^\infty{\left( \frac{c(n-2)}{n-1} \right)^k}\tens{B}^{inv}  
 = \frac{cn}{n((n-1)-c(n-2))-c^2(n-1)} 
\end{split}
\end{equation}
Since $\tens{E}$ and $\tens{D}\tens{B}^{-1}\tens{C}$ are both symmetric and have every non-diagonal element equal as well as all diagonal elements equal, the inverse $\tens{E}^{inv} = (\tens{E} -\tens{D}\tens{B}^{-1}\tens{C})^{-1}$ should be the same as well. Especially every row and column have the same sum.
From Lemma \ref{lem:blockwise} for blockwise inversion we get:
\begin{equation}
C^{inv}_i = -\frac{-c}{n-1} \sum_{k=1}^{n-1}{E^{inv}_{ki}}
\end{equation}
\begin{equation}
D^{inv}_i = -\frac{-c}{n} \sum_{k=1}^{n-1}{E^{inv}_{ik}} = -\frac{-c}{n} \sum_{k=1}^{n-1}{E^{inv}_{ki}}
\end{equation}
\begin{equation}
\Rightarrow \left \{ \begin{split}
D^{inv}_i  = \frac{(n-1)C^{inv}_i}{n} \\
\sum_{k=1}^{n-1}{E^{inv}_{ik}} = \frac{(n-1)C^{inv}_i}{c} \end{split}\right.
\end{equation}
We get the PageRank as: 
\begin{equation}
\begin{split}
\vec{R}^{(2)}_1 = \tens{B}^{inv} +  (n-1) \tens{C}^{inv}  = \frac{n((n-1)-c(n-2))}{n((n-1)-c(n-2))-c^2(n-1)} \\
+\frac{(n-1)cn}{n((n-1)-c(n-2))-c^2(n-1)} = \frac{n(n-1)+nc}{n(n-1)-(n-1)c^2-n(n-2)c}
\end{split}
\end{equation}
\begin{equation}
\begin{split}
\vec{R}^{(2)}_i &= D^{inv} +\sum_{k=1}^{n-1}{E^{inv}_{ik}} = \frac{(n-1)C_inv_i}{n}+ \frac{(n-1)C^{inv}_i}{c}\\
&= \frac{(n-1)C^{inv}_i(c+n)}{nc}  = \frac { \left( c+n \right)  \left( n-1 \right) }{n(n-1)-(n-1)c^2-n(n-2)c}
\end{split}
\end{equation}
And the proof is complete.
\qed 
\end{proof}

We give a matrix proof of Theorem~\ref{thm:complete_out} as well:
\begin{proof}[Proof of Theorem~\ref{thm:complete_out}]
We consider the square matrix $\tens{A}$ with $n$ rows.
$$\tens{A} = \begin{bmatrix}
1 & a & a & \ldots & a\\
b & 1 & a & \ldots & a\\
b & a & 1 & \ddots & a\\
\vdots & \vdots & \ddots & \ddots & \vdots \\
b & a & a & \ldots & a
\end{bmatrix} $$ 
Where $a = -c/(n-1),\ b = -c/(n)$. We divide the matrix in blocks:
$$\tens{A} = \begin{bmatrix}
\tens{B} & \tens{C}\\
\tens{D} & \tens{E}
\end{bmatrix}$$
Where $\tens{B} = [1],\ \tens{C} = [a \ a \ \ldots \ a],\ \tens{D} = [b \ b \ \ldots \ b]^\top$ and $\tens{E}$ have looks like the matrix for a complete graph but is of size $(n-1)\times (n-1)$:
$$\tens{E} = \begin{bmatrix}
1 & a & a & \ldots & a\\
a & 1 & a & \ldots & a\\
a & a & 1 & \ddots & a\\
\vdots & \vdots & \ddots & \ddots & \vdots \\
a & a & a & \ldots & a
\end{bmatrix} $$
In the same way as in the proof of Lemma \ref{lem:complete} we find the elements of $\tens{B},\tens{C}$ by choosing the top left element as pivot element. This gives
\begin{equation}
k_A = \frac{-a}{(1-ab)+(n-2)(a-ab)}  
\end{equation}
We write $\tens{A}^{-1}$ as blocks:
$$\tens{A}^{-1} = \begin{bmatrix}
\tens{B}^{inv} & \tens{C}^{inv}\\
\tens{D}^{inv} & \tens{E}^{inv}
\end{bmatrix} $$
and get: $\tens{B}^{inv} = 1-(n-1)bk_A$ and $C_i^{inv} = k_A$. 

From the matrix proof of Lemma \ref{lem:complete} we get the non-diagonal elements $E_e$ and diagonal elements $E_d$ of $E^{-1}$ as
\begin{equation}
E_e =k_D = \frac{-a}{(1-a^2)+(n-3)(a-a^2)} = \frac{(n-1)c}{(n-1)^2 -(n-3)(n-1)c -(n-2)c^2}
\end{equation}
\begin{equation}
E_d = 1-(n-2)ak_D = \frac{(n-1)^2 -(n-3)(n-1)c}{(n-1)^2 -(n-3)(n-1)c -(n-2)c^2}
\end{equation}
From Lemma \ref{lem:blockwise} we then get: 
\begin{equation}
\tens{B}^{inv} = (\tens{B}-\tens{C}\tens{E}^{-1}\tens{D})^{-1} = 1-(n-1)bk_A
\end{equation}
\begin{equation}
\begin{split}
\tens{C}^{inv} = -(B-CE^{-1}D)^{-1}CE^{-1} \\
\Rightarrow C^{inv}_i = -(B-CE^{-1}D)^{-1}b(E_d+(n-2)E_e) = k_A
\end{split}
\end{equation}
\begin{equation}
\begin{split}
\tens{D}^{inv} = -E^{-1}D(B-CE^{-1}D)^{-1}\\
 \Rightarrow D^{inv}_i = -a(E_d+(n-2)E_e)(B-CE^{-1}D)^{-1} = \frac{bk_A}{a}
\end{split}
\end{equation}
\begin{equation}
\tens{E}^{inv} = E^{-1} +E^{-1}D(B-CE^{-1}D)^{-1}CE^{-1} 
\end{equation}
\begin{equation}
\Rightarrow \left \{ \begin{split}
 E^{inv}_d &= E_d + b(E_d+(n-2)E_e) (\tens{B}-\tens{C}\tens{E}^{-1}\tens{D})^{-1} a(E_d+(n-2)E_e) \\
&= E_d- b(E_d+(n-2)E_e) k_A \\
 E^{inv}_e &= E_e + b(E_d+(n-2)E_e) (\tens{B}-\tens{C}\tens{E}^{-1}\tens{D})^{-1} a(E_d+(n-2)E_e)\\
&= E_e- b(E_d+(n-2)E_e) k_A 
\end{split} \right.
\end{equation}
We replace $a = -c/(n-1)$ and $b = -c/n$ as for our complete graph and get inverse: 
\begin{displaymath}
(\tens{I}-c\tens{A}^\top)^{-1}
 \end{displaymath}
\begin{displaymath}
=\begin{bmatrix}
1-(n-1)bk_A & k_A & k_A & \ldots & k_A\\
\frac{bk_A}{a} & 1-(n-2)ak_D & k_D & \ldots & k_D\\
\frac{bk_A}{a} & k_D & 1-(n-2)ak_D & \ddots & k_D\\
\vdots & \vdots & \ddots & \ddots & \vdots \\
\frac{bk_A}{a} & k_D & k_D & \ldots & 1-(n-2)ak_D
\end{bmatrix} 
\end{displaymath}
For the PageRank of the node linking out we get: 
\begin{equation}
\begin{split}
\vec{R}_1^{(2)} &= \tens{B}^{inv}+(n-1)C^{inv}_i = 1-(n-1)bk_A + (n-1)k_A = 1-(n-1)(b-1)k_A \\
& = \frac{(1-ab)+(n-2)(a-ab) +(n-1)(b-1)a}{(1-ab)+(n-2)(a-ab)} \\&= \frac{(1-ab)-(a-ab)-(n-1)a }{(1-ab)+(n-2)(a-ab)} 
 = \frac{n(n-1) +cn}{n(n-1)-n(n-2)c-(n-1)c^2}
\end{split}
\end{equation}
For all other nodes we get PageRank: 
\begin{equation}
\begin{split}
\vec{R}^{(2)}_i &=  D^{inv}_i + E^{inv}_d +(n-2)E^{inv}_e =  \frac{bk_A}{a} + E_d- b(E_d+(n-2)E_e) k_A \\
&~~+ (n-2)E_e- (n-2)b(E_d+(n-2)E_e) k_A\\
&=E_d+(n-2)E_e -(n-1)b(E_d+(n-2)E_e)k_A +(b/a)k_A \\
&= \frac{1-a}{(1-a^2)+(n-3)(a-a^2)}+ \frac{-b}{(1-ab)+(n-2)(a-ab)}\\
&~~+ \frac{(n-1)ab(1-a)}{\left( (1-a^2)+(n-3)(a-a^2 \right)\left( (1-ab)+(n-2)(a-ab)\right) }\\
&=\frac{1-b}{1-ab+(n-2)(a-ab)} = \frac{(n-1)(n+c)}{n(n-1)-n(n-2)c-(n-1)c^2}
\end{split}
\end{equation}
And the proof is complete.

\qed 
\end{proof}

Just looking at the expression it is hard to say how the PageRank changes after linking out. We can however note a couple of things: First of all the PageRank is lower than for the complete graph (since we now have a chance to escape the graph). But more interesting, when comparing the node that links out with the others we formulate the following theorem:

\begin{theorem}
In a complete graph not linked to from the outside but with one link out, the node that links out will have the highest PageRank in the complete graph.
\end{theorem}

\begin{proof}
Using the expression for PageRank in a complete graph with one link out we want to prove $\vec{R}^{(2)}_k > \vec{R}^{(2)}_i$ where $\vec{R}^{(2)}_k$ is the PageRank for the node linking out and $\vec{R}^{(2)}_i$ is the PageRank of all the other nodes. 
\begin{equation}
\begin{split}
\vec{R}^{(2)}_k  &> \vec{R}^{(2)}_i \\
\Leftrightarrow\frac{n(n-1)+nc}{n(n-1)-(n-1)c^2-n(n-2)c}  &>
  \frac { \left( c+n \right)  \left( n-1 \right) }{n(n-1)-(n-1)c^2-n(n-2)c}  \\
 \Leftrightarrow n(n-1)+nc &> ( c+n )  ( n-1 )\\
\Leftrightarrow n^2+nc-n &> n^2+nc-n-c
\end{split}
\end{equation}
Where $0<c<1$ and $n>1$ is the number of nodes in the complete graph. This is obviously true and the proof is complete.
\qed 
\end{proof}

Apart from the knowledge that it is the node that links out of a complete graph that loses the least from it we can also see that as the number of nodes in the complete graph increases the difference between them decreases since we have a factor $n^2$ in the denominator compared to only a difference of $c$ in the nominator. 

\subsubsection{Effects of linking to a complete graph}
In the case of a link to a complete graph without a link back from the complete graph we can easily guess the result. From earlier we know that for a node linking to one other node in a link matrix with no change of getting back to itself the column corresponding to the node linking out is $c$ times the column of the node it links to. Additionally we need to add a one to the diagonal element for that column. 

The fact that there is no probability (or a very low if it is only close to complete) to escape the complete graph and give any advantage to other parts of the system means the complete graph as a whole get maximum benefit from the links to it. Looking at how the additional probability $c/(1-c) = c +c^2 +c^3 +\ldots +c^{\infty}$ get distributed within the complete graph we realize that the node linked to gains all of the initial $c^1$, then loses a part $c^2$ distributed among all other nodes in the complete graph, after that the rest is close to evenly distributed between all the nodes in the complete graph. As such we see that the node linked to is the node which gains the most from the link (which is what we would expect).

\subsection{Connecting the simple line with the complete graph}\label{conn_line_complete}
Here we will look at what happens when we connect a complete graph with a simple line in various ways. This way we can get some information on what type of structure is most effective in getting a high PageRank and see how they interact with each other. 

\subsubsection{Connecting the simple line with a link from a node in the complete graph to a node in the line}
Looking at the system where we let one node in a complete graph link to one node in a simple line we get a system similar to the case where we added a single node to the line (complete graph with one node). An example of what the system could look like can be seen in Fig.~\ref{GraphLineFromGraph}. We have the two systems $S_L$, $S_G$ as the original systems for the simple line and complete graph respectively. We want to find the new PageRank of these nodes after creating our new system $S$ by adding a link from the first node in the complete graph $e_{G,1}$ to node $e_{L,j}$ in the simple line. When using $n_L=5,n_G=5,j=3$ we get the system with $(\tens{I}-c\tens{A}^\top)$ seen below:
\begin{figure} [!hbt]
\begin{center}
 \begin{tikzpicture}[->,shorten >=0pt,auto,node distance=1.5cm,on grid,>=stealth',
every state/.style={circle,,draw=black}]
\foreach \name/\angle/\text in {A/162/n_{10}, B/90/n_6, C/18/n_7, D/-54/n_8, E/-126/n_9}
\node[state, xshift=6cm,yshift=.5cm] (\name) at (\angle:1.5cm) {$\text$};
\node[state] (H) [above=of B] {$n_3$};
\node[state] (G) [left=of H] {$n_2$};
\node[state] (I) [right=of H] {$n_4$};
\node[state] (F) [left =of G] {$n_1$};
\node[state] (J) [right=of I] {$n_5$};
\foreach \from/\to in {A/B,A/C,A/D,A/E,B/C,B/D,B/E,C/D,C/E,D/E}
\draw [<->] (\from) -- (\to);
\path (G) edge node {} (F)
(H) edge node {} (G)
(I) edge node {} (H)
(J) edge node {} (I)
(B) edge node {} (H);
\end{tikzpicture}
\end{center}
\caption{Simple line with one link from a complete graph to one node in the line}
\label{GraphLineFromGraph}
\end{figure}
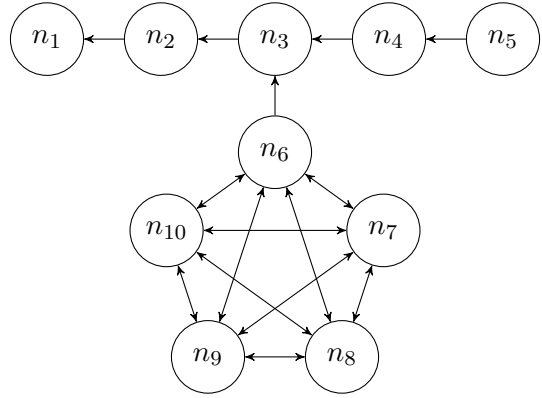
\begin{displaymath}
 I-c\tens{A}^\top = \left[ \begin{array}{cccccccccc}
1 & -c & 0 & 0 & 0 &0 & 0 & 0 & 0 & 0 \\
0 & 1 & -c & 0 & 0&0 & 0 & 0 & 0 & 0 \\
0 & 0 & 1 & -c & 0&-c/5 & 0 & 0 & 0 & 0 \\
0 & 0 & 0 & 1 & -c&0 & 0 & 0 & 0 & 0 \\
0 & 0 & 0 & 0 & 1&0 & 0 & 0 & 0 & 0 \\
0&0 & 0 & 0 & 0 &1 & -c/4 & -c/4 & -c/4 & -c/4 \\
0&0 & 0 & 0 & 0 &-c/5& 1& -c/4 & -c/4 & -c/4 \\
0&0 & 0 & 0 & 0 &-c/5& -c/4 & 1 & -c/4 & -c/4 \\
0&0 & 0 & 0 & 0 &-c/5& -c/4 & -c/4 & 1 & -c/4 \\
0&0 & 0 & 0 & 0 &-c/5& -c/4 & -c/4 & -c/4 & 1 
\end{array} \right]
\end{displaymath}
We find the inverse as: 
\begin{displaymath}
(\tens{I}-c\tens{A}^\top)^{-1}= \end{displaymath}
\begin{displaymath}
\left[ \begin{array}{cccccccccc}
1 & c &c^2&c^3&c^4&\frac{c^3(3c-4)}{s} & -\frac{c^4}{s} & -\frac{c^4}{s} & -\frac{c^4}{s} & -\frac{c^4}{s} \\
0 & 1 & c &c^2&c^3&\frac{c^2(3c-4)}{s} & -\frac{c^3}{s} & -\frac{c^3}{s} &-\frac{c^3}{s} & -\frac{c^3}{s} \\
0 & 0 & 1 & c &c^2&\frac{c(3c-4)}{s} & -\frac{c^2}{s} & -\frac{c^2}{s} & -\frac{c^2}{s} & -\frac{c^3}{s} \\
0 & 0 & 0 & 1 & c &0 & 0 & 0 & 0 & 0 \\
0 & 0 & 0 & 0 & 1 &0 & 0 & 0 & 0 & 0 \\
0 & 0 & 0 & 0 & 0 &\frac{15c-20}{s} & \frac{-5c}{s} & \frac{-5c}{s} &\frac{-5c}{s} &\frac{-5c}{s} \\
0 & 0 & 0 & 0 & 0 &\frac{-4c}{s} &\frac{t}{(c+4)s} & -\frac{4c(5+c)}{(c+4)s} &-\frac{4c(5+c)}{(c+4)s}&-\frac{4c(5+c)}{(c+4)s} \\
0 & 0 & 0 & 0 & 0 &\frac{-4c}{s} &-\frac{4c(5+c)}{(c+4)s} & \frac{t}{(c+4)s} &-\frac{4c(5+c)}{(c+4)s} &-\frac{4c(5+c)}{(c+4)s} \\
0 & 0 & 0 & 0 & 0 &\frac{-4c}{s} &-\frac{4c(5+c)}{(c+4)s} &-\frac{4c(5+c)}{(c+4)s} &\frac{t}{(c+4)s} &-\frac{4c(5+c)}{(c+4)s} \\
0 & 0 & 0 & 0 & 0 &\frac{-4c}{s} &-\frac{4c(5+c)}{(c+4)s} &-\frac{4c(5+c)}{(c+4)s} &-\frac{4c(5+c)}{(c+4)s} &\frac{t}{(c+4)s} \\
\end{array} \right]
\end{displaymath}
where $s = 4c^2+15c-20$ and $t = 12c^2+40c-20$.
Using what we earlier learned by doing changes to the simple line and complete graph separately we can note a couple of things. For the nodes in the complete graph we get the same expression as when adding a link out of the complete graph (since there are no link back to the complete graph). 

Assuming uniform $\vec{u}$ the PageRank in the simple line after adding the link from the complete graph $\vec{R}^{(2)}_{L}[S_G \rightarrow S_L]$ can still be written in about the same way:

\begin{theorem}\label{thm:line_comp_fg}
Observing the nodes in a system $S$ made up of two systems, a simple line $S_L$ with $n_L$ nodes and a complete graph $S_G$ with $n_G$ nodes where we add one link from node $e_g$ in the complete graph to node $e_j$ in the simple line. Assuming uniform $\vec{u}$ we get the PageRank $\vec{R}^{(2)}_{L,i}[S_G \rightarrow S_L]$ for the nodes in the line after the new link and $\vec{R}^{(2)}_{G,i}[S_G \rightarrow S_L]$ for the nodes in the complete graph after the new link as: 
\begin{equation}
\begin{split}
\vec{R}^{(2)}_{L,i}[S_G \rightarrow S_L] = \sum_{k=0}^{n_L-i}{c^k}+b_{ij} = \frac{1-c^{n_L-i+1}}{1-c}+b_{ij}\\
b_{ij} = -c^{j+1-i}\frac{c+(n_G-1)}{(n_G-1)c^2+n_G(n_G-2)c-n_G(n_G-1)},~~ j \geq i \\
b_{ij} = 0,~~ j < i 
\end{split}
\end{equation}
\noindent For the nodes in the complete graph we get:
\begin{equation}
\vec{R}^{(2)}_{G,1}[S_G \rightarrow S_L] = -\frac{n_G(n_G-1)+n_Gc}{(n_G-1)c^2+n_G(n_G-2)c-n_G(n_G-1)}
\end{equation}
\begin{equation}
\vec{R}^{(2)}_{G,j}[S_G \rightarrow S_L] =   \frac { \left( c+n_G \right)  \left( n_G-1 \right) }{n_G(n_G-1)-(n_G-1)c^2-n_G(n_G-2)c }
\end{equation}
where $\vec{R}^{(2)}_{G,1}[S_G \rightarrow S_L]$ is PageRank for the node in the complete graph linking to the line and $\vec{R}^{(2)}_{G,j}[S_G \rightarrow S_L]$ is the PageRank of the other nodes in the complete graph.
\end{theorem}

\begin{proof}
For the nodes in the complete graph we get the PageRank immediately from Theorem~\ref{thm:complete_out}. 

For the nodes in the line we get a similar result as when adding a link from a single node to the line in Theorem~\ref{thm:line_add_node}. We get the same PageRank for the nodes we can not reach from the complete graph ($b_{ij}=0,\  j<i$). For the nodes we can reach we need to modify $b_{ij}$. The sum of all probability to reach the node in the complete graph linking to the line is found in equation~\ref{eq:complete_out_1} in Theorem~\ref{thm:complete_out}.  
\begin{displaymath}
\vec{R}^{(2)}_{G,1}[S_G \rightarrow S_L] = -\frac{n_G(n_G-1)+n_Gc}{(n_G-1)c^2+n_G(n_G-2)c-n_G(n_G-1)}
\end{displaymath}
The probability to reach the linked to node in the line $e_j$ is then 
\begin{displaymath}
\left( \frac{c}{n_G} \right) \vec{R}^{(2)}_{e_1 \in S_G}[S_G \rightarrow S_L]
\end{displaymath}
and for any further node in the line we need to multiply with $c$ for every extra step. This gives: 
\begin{equation}
\begin{split}
b_{ij} &= -c^{j-i}\frac{c}{n_G}\frac{n_G(n_G-1)+n_Gc}{(n_G-1)c^2+n_G(n_G-2)c-n_G(n_G-1)}\\
&= -c^{j+1-i}\frac{c+(n_G-1)}{(n_G-1)c^2+n_G(n_G-2)c-n_G(n_G-1)},\quad j \geq i
\end{split}
\end{equation}
And the proof is complete. \qed 
\end{proof}
We give a matrix proof as well: 
\begin{proof}[Proof of Theorem~\ref{thm:line_comp_fg}]
We divide $(\tens{I}-c\tens{A}^\top)$ in blocks: 
\begin{displaymath} (\tens{I}-c\tens{A}^\top) = \begin{bmatrix}
\tens{B} & \tens{C} \\
\tens{D} & \tens{E}
\end{bmatrix}
\end{displaymath}
where $\tens{B}$ is a $n_L\times n_L$ matrix corresponding to the nodes in the line, $\tens{C}$ is a $n_G \times n_L$ matrix of all elements zero except element $C_{jg} = -c/n_G$, where $e_j$ is the node in the line linked to by $e_g$ in the complete graph. $\tens{D}$ is a zero matrix of size $n_L \times n_G$ and $\tens{E}$ is the $n_G \times n_G$ matrix corresponding to a complete graph with node $e_g$ linking out of the graph. We write: 
\begin{displaymath}
 (\tens{I}-c\tens{A}^\top)^{-1} = \begin{bmatrix}
\tens{B}^{inv} & \tens{C}^{inv} \\
\tens{D}^{inv} & \tens{E}^{inv}
\end{bmatrix}
\end{displaymath}
From  Lemma~\ref{lem:blockwise} for blockwise inversion we get: 
\begin{equation}
\begin{split}
\tens{B}^{inv} &= (\tens{B}-\tens{D}\tens{E}^{-1}\tens{C})^{-1} = \tens{B}^{-1} \\
\tens{C}^{inv} &=  -(\tens{B}-\tens{D}\tens{E}^{-1}\tens{C})^{-1}\tens{C}\tens{E}^{-1} = -\tens{B}^{-1}\tens{C}\tens{E}^{-1}\\
\tens{D}^{inv} &= -\tens{E}^{-1}\tens{D}(\tens{B}-\tens{D}\tens{E}^{-1}\tens{C})^{-1} =\tens{0}\\
\tens{E}^{inv} &= \tens{E}^{-1}+\tens{E}^{-1}\tens{D}(\tens{B}-\tens{D}\tens{E}^{-1}\tens{C})^{-1}\tens{C}\tens{E}^{-1} = \tens{E}^{-1}
\end{split}
\end{equation}
Since $\tens{D}=\tens{0}$ and $\tens{E}$ is the matrix for a complete graph with a node linking out we get from Theorem~\ref{thm:complete_out} the PageRank for the nodes in the complete graph: 
\begin{displaymath}
\vec{R}^{(2)}_{G,g}[S_G \rightarrow S_L] = \frac{n(n-1)+nc}{n(n-1)-(n-1)c^2-n(n-2)c}
\end{displaymath}
\begin{displaymath}
\vec{R}^{(2)}_{G,i}[S_G \rightarrow S_L] = \frac { \left( c+n \right)  \left( n-1 \right) }{n(n-1)-(n-1)c^2-n(n-2)c },\quad i \neq g 
\end{displaymath}
For the nodes in the line we need to calculate $C^{inv}$. 
\begin{displaymath}\tens{C}^{inv} = -\tens{B}^{-1}\tens{C}\tens{E}^{-1} = -\begin{bmatrix}
1 & c & c^2 & \ldots & c^{n_L-1}\\
0 & 1 & c & \ldots & c^{n_L-2}\\
0 & 0 & 1 & \ddots & c^{n_L-3}\\
\vdots & \vdots & \ddots & \ddots & \vdots \\
0 & 0 & 0 & \ldots & 1 
\end{bmatrix}\tens{C}\tens{E}^{-1}\end{displaymath}
Calculating $-\tens{B}^{-1}\tens{C}$ we get:
\begin{equation}
 (-\tens{B}^{-1}\tens{C})_{kl} = \left \{ \begin{split}
0 &,~~ l \neq g \\
0 &,~~ k > j\\
c^{j-k}c/n_G &,~~ k \le j,~~ l = g
\end{split}
\right.\end{equation}
In other words zero except for column $g$. Since only one column is non-zero the sum of every row of $-\tens{B}^{-1}\tens{C}\tens{E}^{-1}$ is then easily found as:
\begin{equation}\begin{split}\sum_{l = 1}^{n_G}{(-\tens{B}^{-1}\tens{C}\tens{E}^{-1})_{kl}} \\
= \left \{ \begin{split}
\frac{c^{j-k+1}}{n_G}\sum_{l = 1}^{n_G}{E^{-1}_{kl}} = \frac{c^{j-k+1}}{n_G}\vec{R}^{(2)}_{G,g}[S_G \rightarrow S_L], \ k \le j \\
 0, ~~ k > j \end{split} \right. \end{split} 
\end{equation}
Since $\tens{B}$ is the matrix for a line we get the total PageRank of the nodes in the line as:
\begin{equation}\vec{R}^{(2)}_{L,i}[S_G \rightarrow S_L] = \frac{1-c^{n_L-i+1}}{1-c}+\frac{c^{j-k}c}{n_G}\sum_{l = 1}^{n_G}{E^{-1}_{kl}} \end{equation}
Where the first part is the part corresponding to the PageRank of a line and the second part is the part influenced by the complete graph. Calculating $\frac{c^{j-k}c}{n_G}\sum_{l = 1}^{n_G}{E^{-1}_{kl}}$ we get:
\begin{equation} \frac{c^{j-k}c}{n_G}\sum_{l = 1}^{n_G}{E^{-1}_{kl}}= \left \{ \begin{split}  \frac{-c^{j+1-i}(c+(n_G-1))}{(n_G-1)c^2+n_G(n_G-2)c-n_G(n_G-1)},~~ j \geq i \\
 0,~~ j < i \end{split} \right. \end{equation}
We replace $\frac{c^{j-k}c}{n_G}\sum_{l = 1}^{n_G}{E^{-1}_{kl}}$ with $b_{ij}$ and the proof is complete. 

For reference we include the whole inverse matrix as well (assuming the first node in the complete graph links out): 
\begin{displaymath} (\tens{I}-c\tens{A}^\top)^{-1} = \begin{bmatrix}
\tens{B}^{\text{inv}} & \tens{C}^{\text{inv}} \\
\tens{D}^{\text{inv}} & \tens{E}^{\text{inv}}
\end{bmatrix}, \quad \tens{B}^{\text{inv}} = \begin{bmatrix}
1 & c & c^2 & \ldots & c^{n_L-1}\\
0 & 1 & c & \ldots & c^{n_L-2}\\
0 & 0 & 1 & \ddots & \vdots\\
\vdots & \vdots & \ddots & \ddots & c \\
0 & 0 & \ldots & 0 & 1
\end{bmatrix} \end{displaymath}
\begin{displaymath} \tens{C}^{\text{inv}} = \begin{bmatrix}
\frac{c^{j}(1-(n_G-1)bk_A)}{n_G}  & \frac{c^{j}k_A}{n_G} & \frac{c^{j}k_A}{n_G} & \ldots & \frac{c^{j}k_A}{n_G} \\
\frac{c^{j-1}(1-(n_G-1)bk_A)}{n_G}  & \frac{c^{j-1}k_A}{n_G} & \frac{c^{j-1}k_A}{n_G} & \ldots & \frac{c^{j-1}k_A}{n_G} \\
\vdots & \vdots & \vdots & \ldots & \vdots \\
\frac{c(1-(n_G-1)bk_A)}{n_G} & \frac{ck_A}{n_G} & \frac{ck_A}{n_G} & \ldots & \frac{ck_A}{n_G} \\
0&0&0&\ldots&0 \\
\vdots & \vdots & \vdots & \ldots & \vdots \\
0&0&0&\ldots&0 
\end{bmatrix}\end{displaymath}
\begin{displaymath}
\tens{D}^{\text{inv}} = \begin{bmatrix}
0 & \ldots & 0 \\
\vdots & \ddots & \vdots \\
0 & \ldots & 0
\end{bmatrix},~ \tens{E}^{\text{inv}} = \begin{bmatrix}
E_A & k_A & k_A & \ldots & k_A\\
\frac{bk_A}{a} & E_D & k_D & \ldots & k_D\\
\frac{bk_A}{a} & k_D &E_D & \ddots & k_D\\
\vdots & \vdots & \ddots & \ddots & \vdots \\
\frac{bk_A}{a} & k_D & k_D & \ldots & E_D
\end{bmatrix} 
\end{displaymath}
\begin{displaymath}E_D = 1-(n_G-2)ak_D, ~E_A = 1-(n_G-1)bk_A \end{displaymath}
\begin{displaymath}k_A = \frac{-a}{(1-ab)+(n_G-2)(a-ab)}, \ k_D = \frac{-a}{(1-a^2)+(n_G-3)(a-a^2)}\end{displaymath}
\begin{displaymath}a = \frac{-c}{n_G-1}, \ b = \frac{-c}{n_G}\end{displaymath}
\qed 
\end{proof}

If we want to know the common normalized PageRank we find the normalizing constant as the sum of the PageRank of all the nodes:
\begin{equation}\begin{split} N = {\frac {n_L}{1-c}}-{\frac {c \left( 1-{c}^{n_L} \right) }{ \left( 1-c
 \right) ^{2}}}+{\frac {c \left( 1-{c}^{n_L-i+2} \right)  \left( c+n_G-1
 \right) }{ \left( 1-c \right)  \left(  \left( n_G-1 \right) {c}^{2}+n_G
 \left( n_G-2 \right) c-n_G \left( n_G-1 \right)  \right) }}\\
+{n_G
 \left( n_G-1 \right) +{n_Gc}{ \left( n_G-1 \right) {c}^{2}+n_G \left( n_G-
2 \right) c-n_G \left( n_G-1 \right) }} \\
+n_G \left(  \left( n_G-1
 \right) {c}^{2}+ \left( 2n_G \left( n_G-1 \right) -1 \right) c+ \left( 
n_G \left( n_G-1 \right)  \right) ^{2} \right) \\
\Bigg( c \left(  \left( n_G-1 \right) {c}^{
2}+n_G \left( n_G-2 \right) c-n_G \left( n_G-1 \right)  \right)\\
+n_G \left( 
 \left( n_G-1 \right) {c}^{2}+n_G \left( n_G-2 \right) c-n_G \left( n_G-1
 \right)  \right) -1 \Bigg)^{-1} 
\end{split} \end{equation}
Which can be used to get the normalized PageRank:
\begin{equation}\vec{R}^{(1)}_{i}[S_G \rightarrow S_L] = \vec{R}^{(2)}_{i}[S_G \rightarrow S_L]/N\end{equation}

\subsubsection{Connecting the simple line with a complete graph by adding a link from a node in the line to a node in the complete graph}

When we instead let one node $e_j$ in the simple line link to one node in the complete graph we get a system that could look like the system in Fig.~\ref{GraphLineFromLine}.
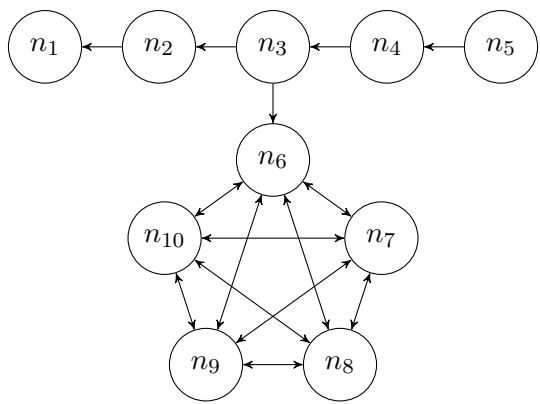
\begin{figure} [!hbt]
\begin{center}
 \begin{tikzpicture}[->,shorten >=0pt,auto,node distance=1.5cm,on grid,>=stealth',
every state/.style={circle,,draw=black}]
\foreach \name/\angle/\text in {A/162/n_{10}, B/90/n_6, C/18/n_7, D/-54/n_8, E/-126/n_9}
\node[state, xshift=6cm,yshift=.5cm] (\name) at (\angle:1.5cm) {$\text$};
\node[state] (H) [above=of B] {$n_3$};
\node[state] (G) [left=of H] {$n_2$};
\node[state] (I) [right=of H] {$n_4$};
\node[state] (F) [left =of G] {$n_1$};
\node[state] (J) [right=of I] {$n_5$};
\foreach \from/\to in {A/B,A/C,A/D,A/E,B/C,B/D,B/E,C/D,C/E,D/E}
\draw [<->] (\from) -- (\to);
\path (G) edge node {} (F)
(H) edge node {} (G)
(I) edge node {} (H)
(J) edge node {} (I)
(H) edge node {} (B);
\end{tikzpicture}
\end{center}
\caption{Simple line with one node in the line linking to a node in a complete graph}
\label{GraphLineFromLine}
\end{figure}

For the PageRank we formulate the following: 

\begin{theorem}\label{thm:line_2_complete}
Observing the nodes in a system $S$ made up of two systems, a simple line $S_L$ with $n_L$ nodes and a complete graph $S_G$ with $n_G$ nodes where we add one link from node $e_j$ in the line to node $e_g$ in the complete graph. Assuming uniform $\vec{u}$ we get the PageRank $\vec{R}^{(2)}_{L,i}[S_L \rightarrow S_G]$ for the nodes in the line after the new link and $\vec{R}^{(2)}_{G,i}[S_L \rightarrow S_G]$ for the nodes in the complete graph after the new link as: 

\begin{equation}\vec{R}^{(2)}_{L,i}[S_L \rightarrow S_G] = \frac{1-c^{n_L+1-i}}{1-c} ,~~ i\ge j\end{equation}
\begin{equation}\vec{R}^{(2)}_{G,g}[S_L \rightarrow S_G] = \left(\frac{c(1-c^{n_L+1-j})}{2(1-c)}\right) \left(\frac{((n_G-1)-c(n_G-2))}{((n_G-1)-c(n_G-2))-c^2} \right) +\frac{1}{1-c} \end{equation}
\begin{equation}\vec{R}^{(2)}_{G,i}[S_L \rightarrow S_G] =  \left( \frac{c^2(1-c^{n_L+1-j)}}{2(1-c)} \right) \left(\frac{1}{((n_G-1)-c(n_G-2)) -c^2} \right) +\frac{1}{1-c} \end{equation}
\begin{equation}\vec{R}^{(2)}_{L,i}[S_L \rightarrow S_G] = \frac{1-c^{j-i}}{1-c} +\left( \frac{c^{j-i}}{2}\right) \frac{1-c^{n_L-j+1}}{1-c},~~ i<j\end{equation}

\end{theorem}

\begin{proof} 
For the nodes above the line we get the same PageRank as for the nodes "above" the linked to node in the line in Theorem~\ref{thm:line_add_node}:
\begin{equation}\vec{R}^{(2)}_{L,i}[S_L \rightarrow S_G] = \frac{1-c^{n_L+1-i}}{1-c} ,~~ i\ge j\end{equation}
In order to find $\vec{R}^{(2)}_{G,g}[S_L \rightarrow S_G]$ we first write it as: 
\begin{equation}\vec{R}^{(2)}_{G,g}[S_L \rightarrow S_G] = \left(\sum_{e_i \in S,\\ e_i \neq e_g}{P(e_i\rightarrow e_g)}+1\right) \left( \sum_{k=0}^{\infty}{(P(e_g \rightarrow e_g))^k}\right)\end{equation}
where $P(e_i \rightarrow e_g)$ is the probability of getting from node $e_i$ to node $e_g$. 
\begin{equation}\begin{split} \sum_{e_i \in S, e_i \neq e_g}{P(e_i\rightarrow e_g)}+1&= 1+\frac{c}{2}\sum_{k=0}^{n_L-j}{c^k} + (n_G-1)\frac{c}{n_G-1}\sum_{k=0}^{\infty}{\left(\frac{c(n_G-2)}{(n_G-1)}\right)^k}\\
&= 1+\frac{c(1-c^{n_L+1-j})}{2(1-c)}+\frac{c(n_G-1)}{(n_G-1)-c(n_G-2)} \end{split}\end{equation}
\begin{equation}P(e_g \rightarrow e_g) = \frac{c^2}{n_G-1}\sum_{k=0}^{\infty}{\left( \frac{c(n_G-2)}{n_G-1}\right)^k} =  \frac{c^2}{(n_G-1)-c(n_G-2)}\end{equation}
\begin{equation}\sum_{k=0}^{\infty}{(P(e_g \rightarrow e_g))^k} =\frac{((n_G-1)-c(n_G-2))}{((n_G-1)-c(n_G-2)) -c^2} \end{equation}
Multiply the two expressions and we get the PageRank for the node linked to in the complete graph.
\begin{displaymath}\vec{R}^{(2)}_{G,g}[S_L \rightarrow S_G] =\end{displaymath}  
\begin{equation}\begin{split} \left(\frac{c(1-c^{n_L+1-j})}{2(1-c)}+1+\frac{c(n_G-1)}{(n_G-1)-c(n_G-2)}\right) \left(\frac{((n_G-1)-c(n_G-2))}{((n_G-1)-c(n_G-2))-c^2} \right)\\
=\left(\frac{c(1-c^{n_L+1-j})}{2(1-c)}\right) \left(\frac{((n_G-1)-c(n_G-2))}{((n_G-1)-c(n_G-2))-c^2} \right) +\frac{1}{1-c} \end{split}\end{equation} 
We previously found the part equal to $\frac{1}{1-c}$ after the proofs of Lemma \ref{lem:complete}.

Below $e_j$ in the line we get a sum of two lines, half the probability of the line starting in $e_L$ and all from the one starting in $e_{j-1}$, since node $e_j$ have two links out where only one links to the line.

Collecting the sum of the probability of all nodes from $e_j$ and above in one term and all below in one gives: 
\begin{equation}\vec{R}^{(2)}_{L,i}[S_L \rightarrow S_G] =\sum_{k=0}^{j-i-1}{c^k} + \sum_{k = j-i}^{n_L-i}{\frac{c^k}{2}} = \frac{1-c^{j-i}}{1-c} +\left( \frac{c^{j-i}}{2}\right) \frac{1-c^{n_L-j+1}}{1-c}, ~ i<j\end{equation}

Last we need to find the PageRank for the nodes in the complete graph not linked to by the node in the line ($\vec{R}^{(2)}_{G,i}[S_L \rightarrow S_G]$). 

We get the same probability of the node getting back to itself as for the node in the complete graph linked to by the line. As such we only need to find the sum of probability of getting to the node once after starting in all nodes once. 

For the other nodes $e_c$ in the complete graph we get:
\begin{equation}P(e_c \rightarrow e_i) = \sum_{k=0}^{\infty}{\frac{c}{n_G-1}\left(\frac{c(n_G-2)}{n_G-1}\right)^k}\end{equation}
We got $n_G-1$ of those nodes giving:
\begin{equation}\sum_{e_c \in S_G, e_c \neq e_i} P(e_c \rightarrow e_i) = \frac{c(n_G-1)}{(n_G-1)-(n_G-2)c}\end{equation}
For the rest of the nodes we can write the sum of probabilities of getting to $e_g$ once from the line as:
\begin{equation}P(e_{line} \rightarrow e_g) = \frac{c}{2}\frac{1-c^{n_L+1-j}}{1-c}\end{equation}
This gives the probability to get from the line to node $e_i$: 
\begin{equation} \begin{split}P(e_{line} \rightarrow e_i) &= \frac{c}{2}\frac{1-c^{n_L+1-j}}{1-c} P(e_g \rightarrow e_i)\\
=\frac{c}{2}\frac{1-c^{n_L+1-j}}{1-c} \left( \frac{c}{(n_G-1)-(n_G-2)c}\right) &=   \frac{c^2(1-c^{n_L+1-j})}{2(1-c)((n_G-1)-(n_G-2)c)}\end{split}\end{equation}
Using this we get the sum of probability to get to a node in the complete graph not linked to by the line as: 
\begin{equation}
\sum_{\stackrel{e_k \in S}{ e_k \neq e_i}}{P(e_{k} \rightarrow e_i)} = \frac{c^2(1-c^{n_L+1-j)}}{2(1-c)((n_G-1)-(n_G-2)c)} + 1+\frac{c(n_G-1)}{(n_G-1)-(n_G-2)c}
\end{equation}
Multiplying with the sum of probability of going from the node in question back to itself gives the PageRank:
\begin{equation}\begin{split}\vec{R}^{(2)}_{G,i}[S_L \rightarrow S_G] = \left(\frac{((n_G-1)-c(n_G-2))}{((n_G-1)-c(n_G-2)) -c^2} \right) \\
 \left( \frac{c^2(1-c^{n_L+1-j)}}{2(1-c)((n_G-1)-(n_G-2)c)} + 1+\frac{c(n_G-1)}{(n_G-1)-(n_G-2)c} \right)\\
 = \left( \frac{c^2(1-c^{n_L+1-j)}}{2(1-c)} \right) \left(\frac{1}{((n_G-1)-c(n_G-2)) -c^2} \right) +\frac{1}{1-c}
\end{split}
\end{equation}
And the proof is complete
\qed 
\end{proof}
We include a matrix proof of Theorem~\ref{thm:line_2_complete} as well:

\begin{proof}[Proof of Theorem~\ref{thm:line_2_complete}]
We divide the matrix $(\tens{I}-c\tens{A}^\top)$ in blocks:
\begin{displaymath}(\tens{I}-c\tens{A}^\top)= \begin{bmatrix}
\tens{B} & \tens{C} \\
\tens{D} & \tens{E} \end{bmatrix} \end{displaymath}
where $\tens{B}$ is the part corresponding to the line, $\tens{C}$ is a zero matrix (since we have no links from nodes in the complete graph to the line). $\tens{D}$ is a zero matrix except for one element $D_{g,j} = -c/2$, where $e_j$ is $j$:th the node in the line linking to the complete graph and $e_g$ is the $g$:th node in the graph linked to by node $e_j$. We note that $j,g$ are the internal number for the complete graph and line respectively and not their "number" in the combined graph. $\tens{E}$ is the part corresponding to the complete graph. 

In the same way we divide the inverse in blocks:
\begin{displaymath}(\tens{I}-c\tens{A}^\top)^{-1}= \begin{bmatrix}
\tens{B}^{inv} & \tens{C}^{inv} \\
\tens{D}^{inv} & \tens{E}^{inv} \end{bmatrix} \end{displaymath}
Using Lemma \ref{lem:blockwise} for blockwise inversion we get:
\begin{equation}
\begin{split}
\tens{B}^{inv} &= (\tens{B}-\tens{D}\tens{E}^{-1}\tens{C})^{-1} = \tens{B}^{-1}\\
\tens{C}^{inv} &=  -(\tens{B}-\tens{D}\tens{E}^{-1}\tens{C})^{-1}\tens{C}\tens{E}^{-1} = \tens{0}\\
\tens{D}^{inv} &=  -\tens{E}^{-1}\tens{D}(\tens{B}-\tens{D}\tens{E}^{-1}\tens{C})^{-1} = \tens{E}^{-1}\tens{D}\tens{B}^{-1}\\
\tens{E}^{inv} &= \tens{E}^{-1} +\tens{E}^{-1}\tens{D}(\tens{B}-\tens{D}\tens{E}^{-1}\tens{C})^{-1}\tens{C}\tens{E}^{-1} = \tens{E}^{-1}\\
\end{split}
\end{equation}
Since one of the nodes in the line links out we get $\tens{B}$ divided in blocks:
\begin{displaymath}
\tens{B} = \begin{bmatrix}
\tens{B}_B & \tens{B}_C \\
\tens{B}_D & \tens{B}_E
\end{bmatrix}
\end{displaymath}
\begin{displaymath}
\tens{B}_B = \begin{bmatrix}
1 & -c & 0 & \ldots & 0\\
0 & 1 & -c & \ddots & \vdots \\
0 & 0 & 1 & \ddots & 0\\
\vdots & \ddots & \ddots & \ddots & -c \\
0 & \ldots& 0& 0& 1
\end{bmatrix} \quad \tens{B}_C = \begin{bmatrix}
0 & \ldots & \ldots & 0\\
\vdots & \ddots & \ddots & \vdots \\
0 & \ddots & \ddots & \vdots \\
-c/2 & 0 & \ldots & 0
\end{bmatrix} 
\end{displaymath}
where $\tens{B}_D$ is a zero matrix and $\tens{B}_E$ looks the same as $\tens{B}_B$ although possibly with a different size. The size of the blocks are: $\tens{B}_B: (j-1) \times (j-1)$, $\tens{B}_C: (j-i) \times (n_L-j+1)$, $\tens{B}_D: ( n_L-j+1) \times (n_L-j+1)$ and $\tens{B}_E:( n_L-j+1 )\times (n_L-j+1)$, where $n_L$ is the total number of nodes in the line. 

For the blocks of the inverse we get: 
\begin{equation}
\begin{split}
 \tens{B}_B^{inv} &= \tens{B}_B^{-1}\\
 \tens{B}_C^{inv} &= -\tens{B}_B^{-1}\tens{B}_C \tens{B}_E^{-1}\\
 \tens{B}_D^{inv} &= \tens{0} \\
 \tens{B}_E^{inv} &= \tens{E}^{-1}
\end{split}
\end{equation}
$\tens{B}_B^{inv}$ and $\tens{B}_E^{inv}$ are found as the inverse for the simple line, leaving $\tens{B}_C^{inv}$ to be computed. The only difference compared to a simple line is that the only non-zero element in $\tens{B}_C$ is $-c/2$ rather than $-c$. In other words $\tens{B}^{-1}$ is exactly as it would have been for a simple line, except block corresponding to $\tens{B}_C^{inv}$ which is multiplied with $0.5$. 

We can now find the PageRank of the nodes in the line:
\begin{equation}
\begin{split}
\vec{R}^{(2)}_{L,i}[S_L \rightarrow S_G] &= \sum_{k=1}^{n_L}{B^{-1}_{i,k}} = \sum_{k=1}^{j-1}{B^{-1}_{i,k}} + \sum_{k=j}^{n_L}{B^{-1}_{i,k}} \\
 = \sum_{k=0}^{j-i-1}{c^k} + \sum_{k=j-i}^{n_L-i}{\frac{c^k}{2}} &=\frac{1-c^{j-i}}{1-c} +\left( \frac{c^{j-i}}{2}\right) \frac{1-c^{n_L-j+1}}{1-c} 
\end{split}
\end{equation}
For the nodes in the complete graph we first need to find $\tens{D}^{inv}$, to do so we start by calculating $\tens{D}\tens{B}^{-1}$. Since only one element $D_{gj}$ of $\tens{D}$ is non-zero, only row $g$ of $\tens{D}\tens{B}^{-1}$ can be non-zero. We get row $g$ as: 
\begin{displaymath}(\tens{D}\tens{B}^{-1})_{\text{row}_g} = \frac{-c}{2}[0 \ \ldots \ 1 \ c \ c^2 \ \ldots \ c^L-j]\end{displaymath}
where there are $j-1$ zeros before the $1$, ($\tens{B}^{-1}$ upper triangular). Multiplying this with $\tens{E}^{-1}$ found in the matrix proof of Lemma~\ref{lem:complete} gives: 
\begin{displaymath} -\tens{E}^{-1}\tens{D}\tens{B}^{-1} =\frac{c}{2} \begin{bmatrix}
0 & \ldots & 0 & s & c s & \ldots & c^{n_L-j}s\\
\vdots & \vdots & \vdots & \vdots & \vdots & \vdots & \vdots \\
0 & \ldots & 0 & s & c s & \ldots & c^{n_L-j}s\\
0 & \ldots & 0 & d & c d & \ldots & c^{n_L-j}d\\
0 & \ldots & 0 & s & c s & \ldots & c^{n_L-j}s\\
\vdots & \vdots & \vdots & \vdots & \vdots & \vdots & \vdots \\
0 & \ldots & 0 & s & c s & \ldots & c^{n_L-j}s
\end{bmatrix}\end{displaymath}
\begin{displaymath} s = \frac{c}{(n_G-1)-c(n_G-2)-c^2},~~	
d = \frac{(n_G-1)-c(n_G-2)}{(n_G-1)-c(n_G-2)-c^2}\end{displaymath}
We can now find the PageRank of the nodes in the complete graph by summation of corresponding row: 
\begin{equation}\vec{R}^{(2)}_{G,i}[S_L \rightarrow S_G] =  \sum_{k=1}^{n_L}{D^{inv}_{ik}} + \sum_{k=1}^{n_L}{E^{inv}_{ik}}  = \sum_{k=1}^{n_L}{D^{inv}_{ik}} + \frac{1}{1-c}\end{equation}
We separate between the node in the complete graph linked to from the line and the other nodes in the complete graph.  
\begin{equation}
\begin{split}
\vec{R}^{(2)}_{G,i}[S_L \rightarrow S_G] = \frac{c}{2}\sum_{k=0}^{n_L-j}{c^k}s +\frac{1}{1-c} \\
= \left(\frac{c(1-c^{n_L-j+1})}{2(1-c)}\right) \left( \frac{c}{(n_G-1)-c(n_G-2)-c^2}\right) +\frac{1}{1-c} ,~i\neq g
\end{split}
\end{equation}
\begin{equation}
\begin{split}
\vec{R}^{(2)}_{G,g}[S_L \rightarrow S_G] =  \frac{c}{2}\sum_{k=0}^{n_L-j}{c^k}d +\frac{1}{1-c} \\
=\left(\frac{c(1-c^{n_L+1-j})}{2(1-c)}\right) \left(\frac{((n_G-1)-c(n_G-2))}{((n_G-1)-c(n_G-2))-c^2} \right) +\frac{1}{1-c}  
\end{split}
\end{equation}
And the proof is complete. 
For completeness we include the complete inverse as well:  
\begin{displaymath}(\tens{I}-c\tens{A}^\top)^{-1}= \begin{bmatrix}
\tens{B}^{inv} & \tens{C}^{inv} \\
\tens{D}^{inv} & \tens{E}^{inv} \end{bmatrix}, \quad \tens{B}^{-1}  \begin{bmatrix}
1 & c & c^2 & \ldots & c^{n_L-1}/2\\
0 & 1 & c & \ldots & c^{n_L-2}/2\\
0 & 0 & 1 & \ddots & \vdots\\
\vdots & \vdots & \ddots & \ddots & c/2 \\
0 & 0 & \ldots & 0 & 1
\end{bmatrix} \end{displaymath}
\begin{displaymath} \tens{C}^{-1} =  \begin{bmatrix}
0 & \ldots & 0 \\
\vdots & \ddots & \vdots \\
0 & \ldots & 0
\end{bmatrix}, \quad \tens{D}^{-1} = -\tens{E}^{-1}\tens{D}\tens{B}^{-1} \text{(seen above)}
\end{displaymath}
\begin{displaymath}
\tens{E}^{-1}  = \begin{bmatrix}
 1-(n_G-1)ak & k & k & \ldots & k \\
 k & 1-(n_G-1)ak & k & \ldots & k \\
 k & k & 1-(n_G-1)ak & \ddots & \vdots \\
  \vdots & \vdots & \ddots & \ddots & k \\
k& k & \ldots  & k&  1-(n_G-1)ak  
\end{bmatrix} \end{displaymath}
\begin{displaymath}
a = -c/(n-1), \ k = \frac{-a}{ (1-a^2 + (n-2)(a-a^2))} \end{displaymath}
\qed 
\end{proof}

The normalizing constant can then be found by summation of the individual PageRank of all the nodes in order to get the normalized PageRank $\vec{R}^{(1)}$.

We note that while the node in the line that links to the complete graph does not lose anything from the new link, the nodes below it in the line do lose quite a lot because of it. 
Likewise the PageRank of the node thats linked to gains more from the new link than the others in the complete graph. 

\subsubsection{Connecting the simple line with a complete graph by letting one node in the line be part of the complete graph}

If we instead let one node in the line be part of the complete graph we get another interesting example to look at. An example of what the system could look like can be seen in Fig.~\ref{GraphLineFromBoth}.
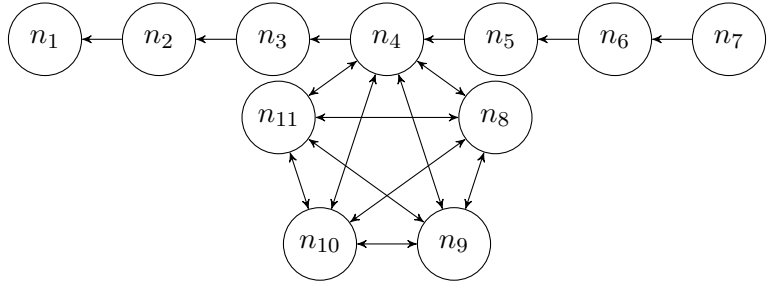
\begin{figure} [!hbt]
\begin{center}
 \begin{tikzpicture}[->,shorten >=0pt,auto,node distance=1.5cm,on grid,>=stealth',
every state/.style={circle,,draw=black}]
\foreach \name/\angle/\text in {A/162/n_{11}, B/90/n_{4}, C/18/n_8, D/-54/n_9, E/-126/n_{10}}
\node[state, xshift=6cm,yshift=.5cm] (\name) at (\angle:1.5cm) {$\text$};
\node[state] (G) [left=of B] {$n_3$};
\node[state] (I) [right=of B] {$n_5$};
\node[state] (F) [left =of G] {$n_2$};
\node[state] (J) [right=of I] {$n_6$};
\node[state] (K) [left =of F] {$n_1$};
\node[state] (L) [right=of J] {$n_7$};
\foreach \from/\to in {A/B,A/C,A/D,A/E,B/C,B/D,B/E,C/D,C/E,D/E}
\draw [<->] (\from) -- (\to);
\path (L) edge node {} (J)
(J) edge node {} (I)
(I) edge node {} (B)
(B) edge node {} (G)
(G) edge node {} (F)
(F) edge node {} (K);
\end{tikzpicture}
\end{center}
\caption{Simple line with one node in the line being a part of a complete graph}
\label{GraphLineFromBoth}
\end{figure}

We formulate the following theorem for the PageRank of the given example:

\begin{theorem}\label{thm:line_part_of_complete}	
The PageRank of the nodes in system $S_L$ made up of a simple line and system $S_G$ made up of a complete graph after one of the nodes $e_j \in S_L$ becomes part of the complete graph assuming uniform $\vec{u}$ can be written: 
\begin{equation}\vec{R}^{(2)}_{L,i}[S_L \leftrightarrow S_G] = \frac{1-c^{n_L+1-i}}{1-c} ,\quad i>j\end{equation}
\begin{equation}\begin{split}\vec{R}^{(2)}_{L,j}[S_L \leftrightarrow S_G] = \left( \frac{1-c^{n_L+1-j}}{1-c}+\frac{c(n_G-1)}{(n_G-1)-c(n_G-2)}\right) \\
\left(\frac{n_G((n_G-1)-c(n_G-2))}{n_G((n_G-1)-c(n_G-2)) -c^2(n_G-1)} \right)\end{split}\end{equation}
\begin{equation}\vec{R}^{(2)}_{L,i}[S_L \leftrightarrow S_G] = \frac{c^{j-i}\vec{R}^{(2)}_{L,j}}{n_G}+\frac{1-c^{j-i}}{1-c}, i < j \end{equation}
\begin{equation}\vec{R}^{(2)}_{G,i}[S_L \leftrightarrow S_G]  = \frac{\left( c+n_G \right)  \left( n_G-1 \right)(1-c) + (n_G-1)c^2(1-c^{n_L-j})}{(1-c)(n_G(n_G-1)-(n_G-1)c^2-n_G(n_G-2)c)}\end{equation}
where $\vec{R}^{(2)}_{G,i}[S_L \leftrightarrow S_G]$ is the PageRank for the nodes in the complete graph (except the node also being a part of the line) and $\vec{R}^{(2)}_{L,i}[S_L \leftrightarrow S_G]$ is the PageRank of nodes in the line.
$n_G,n_L$ is the number of nodes in the complete graph and simple line respectively after making one node in the line part of the complete graph.
\end{theorem}

\begin{proof}
For the proof of the nodes $e_i \in S_L, ~ i>j$ we get the PageRank for a simple line.
In order to find $\vec{R}^{(2)}_{L,j}[S_L \leftrightarrow S_G]$ we first use Theorem~\ref{thm:prProb} to write it as: 
\begin{equation}\vec{R}^{(2)}_{L,j}[S_L \leftrightarrow S_G] =\left( \sum_{e_i \in S,\\ e_i \neq e_j}{P(e_i\rightarrow e_j)}+1\right) \left(\sum_{k=0}^{\infty}{(P(e_j \rightarrow e_j))^k}\right)
\end{equation} 
where $P(e_i \rightarrow e_j)$ is the probability of getting from node $e_i$ to node $e_j$. 
\begin{equation}\label{eq:proof_th8_1}\begin{split}\sum_{e_i \in S, e_i \neq e_j}{P(e_i \rightarrow e_j)} +1&= \sum_{k=0}^{n_L+1-j}{c^k} + (n-1)\frac{c}{n_G-1}\sum_{k=0}^{\infty}{\left(\frac{c(n_G-2)}{(n_G-1)}\right)^k}\\
&= \frac{1-c^{n_L+1-j}}{1-c}+\frac{c(n_G-1)}{(n_G-1)-c(n_G-2)}
\end{split}\end{equation}
\begin{equation}\begin{split}P(e_j \rightarrow e_j) &= \frac{c(n_G-1)}{n_G}\frac{c}{n_G-1}  + \frac{c(n_G-1)}{n_G}\frac{c(n_G-2)}{n_G-1}\frac{c}{n_G-1} \\
&~~+\frac{c(n_G-1)}{n_G}\frac{c^2(n_G-2)^2}{(n_G-1)^2}\frac{c}{n_G-1}+ \ldots  \\
&= \frac{c^2}{n_G}\sum_{k=0}^{\infty}{\left( \frac{c(n_G-2)}{n-1}\right)^k} = \frac{c^2(n_G-1)}{n_G((n_G-1)-c(n_G-2))} 
\end{split}\end{equation}
\begin{equation}\label{eq:proof_th8_2}\sum_{k=0}^{\infty}{(P(e_j \rightarrow e_j))^k} = \frac{n_G((n_G-1)-c(n_G-2))}{n_G((n_G-1)-c(n_G-2)) -c^2(n_G-1)} \end{equation}
Multiplication of the results from equation~(\ref{eq:proof_th8_1}) and (\ref{eq:proof_th8_2}) gives
\begin{equation}\begin{split}
\vec{R}^{(2)}_{L,j}[S_L \leftrightarrow S_G] = \left( \frac{1-c^{n_L+1-j}}{1-c}+\frac{c(n_G-1)}{(n_G-1)-c(n_G-2)}\right)\\
\left(\frac{n_G((n_G-1)-c(n_G-2))}{n_G((n_G-1)-c(n_G-2)) -c^2(n_G-1)} \right)
\end{split}\end{equation}
For the nodes below the one in the complete graph we can write the PageRank as: 
\begin{equation}\vec{R}^{(2)}_{L,i}[S_L \leftrightarrow S_G] = \vec{R}^{(2)}_{L,j}[S_L \leftrightarrow S_G] P(e_j \rightarrow e_i) +\sum_{k=0}^{j-i-1}{c^k}, \quad i<j\end{equation}
\begin{equation}P(e_j \rightarrow e_i) = \frac{c^{j-i}}{n_G} \end{equation}
This gives:
\begin{equation}\vec{R}^{(2)}_{L,i}[S_L \leftrightarrow S_G] = \frac{c^{j-i}\vec{R}^{(2)}_{L,j}[S_L \leftrightarrow S_G]}{n_G}+\frac{1-c^{j-i}}{1-c} ,i < j\end{equation}

Left to prove we have the formula for the nodes in the complete graph not directly connected to the line $\vec{R}^{(2)}_{G,i}[S_L \leftrightarrow S_G]$. 
We do not need to consider the part of the line following the complete graph, since we can not get from this part of the graph back to the complete graph. 
We already have the PageRank for the nodes not linking out in the complete graph in the case where we have no line of nodes linking to the complete graph from Theorem \ref{thm:complete_out}. In the Matrix proof there we found the PageRank of the nodes in the complete graph not linking out to be: 
\begin{equation}\begin{split}\vec{R}^{(2)}_{G,i}[S_G] = D_i^{inv}+E_d^{inv}+(n_G-2)E_e^{inv}  \\
=\left( P(e_j \rightarrow e_i) +1+\sum_{\stackrel{e_k \in S_GS}{ e_k \neq e_j,e_i}}{P(e_k\rightarrow e_i)}  \right) \left( \sum_{k=0}^{\infty}{(P(e_i \rightarrow e_i))^k}\right)\end{split}\end{equation}
We identify $D^{inv}_i$ in the expression: 
\begin{equation}D_i^{inv} = P(e_j \rightarrow e_i)\left( \sum_{k=0}^{\infty}{(P(e_i \rightarrow e_i))^k}\right)\end{equation}
Since all paths $P(e_k \rightarrow e_i)$ where $e_k \in S_L$ need to go through node $e_j$ we can write these as a product of the probability to get to node $e_j$ times the probability to from there get to node $e_i$ for which we want to calculate PageRank. 
\begin{equation}\begin{split}\vec{R}^{(2)}_{G,i}[S_L \leftrightarrow S_G] = \left( \sum_{k=0}^{\infty}{(P(e_i \rightarrow e_i))^k}\right)\\
 \left( \left( \sum_{\stackrel{e_k \in S_L}{ e_k \neq e_j}}{P(e_k \rightarrow e_j)}+1\right) P(e_j \rightarrow e_i) +1+\sum_{\stackrel{e_k \in S_G}{\ e_k \neq e_j,e_i}}{P(e_k\rightarrow e_i)}  \right) \end{split}\end{equation}
Using the expressions found in the matrix proof of Theorem \ref{thm:complete_out} we get PageRank. 
\begin{equation}\begin{split}
\vec{R}^{(2)}_{G,i}[S_L \leftrightarrow S_G] &= (\sum_{k=0}^{n_L-j}c^k){D_i}^{inv}+E_d^{inv}+(n_G-2)E_e^{inv}  \\
&= \frac{(n_G-1)(n_G+c)}{n_G(n_G-1) -n_G(n_G-2)c -(n_G-1)c^2}\\
+ \left( \frac{c(1-c^{n_L-j})}{1-c}\right)& \frac{c(n_G-1)}{n_G(n_G-1) -n_G(n_G-2)c -(n_G-1)c^2}\\
&=  \frac{\left( c+n_G \right)  \left( n_G-1 \right)(1-c) + (n_G-1)c^2(1-c^{n_L-j})}{(1-c)(n_G(n_G-1)-(n_G-1)c^2-n_G(n_G-2)c)}\end{split}\end{equation}
\qed 
\end{proof}
We give a matrix proof of Theorem~\ref{thm:line_part_of_complete} as well. 

\begin{proof}[Proof of Theorem~\ref{thm:line_part_of_complete}]
For the proof we consider the structure in Fig.~\ref{GraphLineFromBoth} when we talk about nodes below/above other nodes. 

For the proof we once again start by dividing the matrix $(\tens{I}-c\tens{A}^\top)$ in blocks: 
\begin{displaymath}(\tens{I}-c\tens{A}^\top)= \begin{bmatrix}
\tens{B} & \tens{C} \\
\tens{D} & \tens{E} \end{bmatrix} \end{displaymath}
And the same for the inverse: 
\begin{displaymath}(\tens{I}-c\tens{A}^\top)^{-1}= \begin{bmatrix}
\tens{B}^{inv} & \tens{C}^{inv} \\
\tens{D}^{inv} & \tens{E}^{inv} \end{bmatrix} \end{displaymath}
We let $\tens{B}$ be the matrix for the nodes "below" the complete graph as well as the complete graph. In other words we can divide $\tens{B}$ itself in blocks: 
\begin{displaymath} \tens{B} = \begin{bmatrix}
\tens{B}_B & \tens{B}_C \\
\tens{B}_D & \tens{B}_E \end{bmatrix} \end{displaymath}
where $\tens{B}_B$ corresponds to the part of the line we can reach from the complete graph (nodes below the complete graph), $\tens{B}_C$ is a zero matrix except for the bottom left element which is equal to $-c/n_G$ corresponding to the link from the complete graph to the linked to node in the line. $\tens{B}_D$ is a zero matrix and $\tens{B}_E$ corresponds to the matrix of a complete graph with the first node in the complete graph linking out (any other could of course be used and would result in the same calculations after multiplication with a suitable permutation matrix.).

Then $\tens{B}$ have the same structure as the matrix for our example system with one node in a complete graph linking to a node in the line in Fig.~\ref{GraphLineFromGraph} (except we have no nodes "above" the complete graph). We note that we have already looked at this system in the matrix proof of Theorem~\ref{thm:line_comp_fg} and will leave it until we will need it later. 

We let $\tens{E}$ correspond to the part of the line "above" the complete graph, which is also something we have looked at multiple times already. $\tens{C}$ is a zero matrix except for element $C_{j1}$ corresponding to the link from the line to the complete graph. $\tens{D}$ is a zero matrix. 

For completeness we include the sizes of the different blocks as well, $\tens{B}: (j-1+n_G) \times (j-1+n_G)$, $\tens{C}: (j-1+n_G) \times (n_L-j)$, $\tens{D}:   (n_L-j)\times (j-1+n_G) $, $\tens{E}:  (n_L-j) \times  (n_L-j)$.

From Lemma~\ref{lem:blockwise} we get:
\begin{equation}\begin{split}
\tens{B}^{inv} &= (\tens{B}-\tens{D}\tens{E}^{-1}\tens{C})^{-1} = \tens{B}^{-1}\\
\tens{C}^{inv} &=  -(\tens{B}-\tens{D}\tens{E}^{-1}\tens{C})^{-1}\tens{C}\tens{E}^{-1} = \tens{B}^{-1}\tens{C}\tens{E}^{-1}\\
\tens{D}^{inv} &=  -\tens{E}^{-1}\tens{D}(\tens{B}-\tens{D}\tens{E}^{-1}\tens{C})^{-1} = \tens{0}\\
\tens{E}^{inv} &= \tens{E}^{-1} +\tens{E}^{-1}\tens{D}(\tens{B}-\tens{D}\tens{E}^{-1}\tens{C})^{-1}\tens{C}\tens{E}^{-1} = \tens{E}^{-1}
\end{split}\end{equation}
We already know most of $\tens{B}^{-1}$ from the matrix proof of Theorem~\ref{GraphLineFromGraph} and $\tens{E}^{-1}$ from the matrix proof of Theorem~\ref{thm:line_add_node}. Left to find is $\tens{C}^{inv} $ before we can find the PageRank. 

Calculating $\tens{C}\tens{E}^{-1}$ we get a $ (j-1+n_G) \times (n_L-j)$ zero matrix except for row $j$ corresponding to the non-zero element $C_{j1}=-c$ in $\tens{C}$. This gives  $\tens{C}\tens{E}^{-1}$:
\begin{displaymath} \tens{C}\tens{E}^{-1} = -c \begin{bmatrix}
0 & 0 & 0 & \ldots & 0 \\
\vdots &\vdots &\vdots &\ldots &\vdots \\
0 & 0 & 0 & \ldots & 0 \\
1 & c & c^2 & \ldots & c^{n_L-j-1}\\
0 & 0 & 0 & \ldots & 0 \\
\vdots &\vdots &\vdots &\ldots &\vdots \\
0 & 0 & 0 & \ldots & 0 
\end{bmatrix}\end{displaymath}
where row $j$ is the non-zero row. When multiplying this with $\tens{B}^{-1}$ we are only interested in column $j$ of $\tens{B}^{-1}$ since only row $j$ of $\tens{C}\tens{E}^{-1}$ is non-zero. From the matrix proof of Theorem~\ref{GraphLineFromGraph} we get column $j$ as:
\begin{equation}B^{-1}_{kj} = \left\{ \begin{split}
&\frac{c^{j-k}}{n_G} \left(1-(n_G-1)bk_A \right), ~~ 1 \le k \ < j\\
&\left(1-(n_G-1)bk_A \right), ~~ k = j\\
&\frac{bk_A}{a}, ~~ j < k \le j-1+n_G  \end{split} \right. \end{equation}
\begin{displaymath} k_A = \frac{-a}{(1-ab)+(n_G-2)(a-ab)}\end{displaymath}
\begin{displaymath} a = \frac{-c}{n_G-1}, \quad b = \frac{-c}{n_G}\end{displaymath}
We note that we get $c^{j-k}$ rather than $c^{j-k+1}$ since node $e_j$ in this case is the node linking to node $e_{j-1}$ rather than node $e_j$ being the one linked to. We are now ready to calculate $C^{inv}$: 
\begin{equation}\tens{C}^{inv}_{\text{row}_k} =  B^{-1}_{kj}[c \ c^2 \ \ldots \ c^{n_L-j}],~~ (1 \le k \le j-1+n_G)\end{equation}
To get the PageRank of a node we need to sum all the elements of corresponding row of $(\tens{I}-c\tens{A}^\top)^{-1}$, for the nodes "above" the complete graph we get the simple line: 
\begin{equation}\vec{R}^{(2)}_{L,i}[S_L \leftrightarrow S_G] = \sum^{n_L-i}_{k=0}{c^k} = \frac{1-c^{n_L-i+1}}{1-c}, \ i > j\end{equation}
For the node part of both the complete graph and the line we get: 
\begin{equation}\begin{split}
\vec{R}^{(2)}_{L,j}[S_L \leftrightarrow S_G] = \sum^{n_G}_{k = 1}{B^{-1}_{jk}} + \sum^{n_L-j}_{k = 1}{C^{inv}_{jk}}\\
=  \frac{n_G(n_G-1)+n_Gc}{n_G(n_G-1)-(n_G-1)c^2-n_G(n_G-2)c} + B^{-1}_{jj}\sum^{n_L-j}_{k=1}{c^k}\\
= \frac{n_G(n_G-1)+n_Gc}{n_G(n_G-1)-(n_G-1)c^2-n_G(n_G-2)c} \\
+\left( \frac{n_G(n_G-1) - n_G(n_G-2)c}{n_G(n_G-1) - n_G(n_G-2)c-2(n_G-2)c^2}\right)\left( \sum^{n_L-j}_{k=0}{c^k} -1\right)\\
 = \left( \frac{1-c^{n_L+1-j}}{1-c}+\frac{c(n_G-1)}{(n_G-1)-c(n_G-2)}\right) \\
\left(\frac{n_G((n_G-1)-c(n_G-2))}{n_G((n_G-1)-c(n_G-2)) -c^2(n_G-1)} \right) \end{split}\end{equation}
For the nodes "below" the complete graph we get: 
\begin{equation}\begin{split}
\vec{R}^{(2)}_{L,i}[S_L \leftrightarrow S_G] = \sum^{n_G}_{k = 1}{B^{-1}_{ik}} + \sum^{n_L-j}_{k = 1}{C^{inv}_{ik}}, ~~ i < j\\
= \sum^{j-i-1}_{k=0}{c^k} + \frac{c^{j-i}(c+(n_G-1))}{n_G(n_G-1) - n_G(n_G-2)c-(n_G-1)c^2} + B^{-1}_{ij}\sum^{n_L-j}_{k=1}{c^k}\end{split}\end{equation}
where we once again note that we get $c^{j-i}$ rather than $c^{j-i+1}$ since we consider node $j-1$ the node linked to by the graph rather than node $j$. 
\begin{equation}\begin{split}
\vec{R}^{(2)}_{L,i}[S_L \leftrightarrow S_G] = \frac{1-c^{j-i}}{1-c} \\
+  \frac{c^{j-i}(c+(n_G-1))}{n_G(n_G-1) - n_G(n_G-2)c-(n_G-1)c^2} + \frac{c^{j-i}}{n_G} B^{-1}_{jj} \sum^{n_L-j}_{k=1}{c^k}\\
= \frac{1-c^{j-i}}{1-c}  + \frac{c^{j-i}}{n_G} \vec{R}^{(2)}_{L,j}[S_L \leftrightarrow S_G]\end{split}\end{equation}

For the nodes in the complete graph not part of the line we get: 
\begin{equation}
\vec{R}^{(2)}_{G,i}[S_L \leftrightarrow S_G] = \sum^{n_G}_{k = 1}{B^{-1}_{ik}} + \sum^{n_L-j}_{k = 1}{C^{inv}_{jk}}, \ i \neq j 
\end{equation}
where $e_j \in S_G$ is the node in the complete graph also part of the line. 
\begin{equation}\begin{split}
\vec{R}^{(2)}_{G,i} &= \frac { \left( c+n_G \right)  \left( n_G-1 \right) }{n_G(n_G-1)-(n_G-1)c^2-n_G(n_G-2)c } + \frac{bk_A}{a}\sum^{n_L-j}_{k=1}{c^k}\\
&=  \frac { \left( c+n_G \right)  \left( n_G-1 \right) }{n_G(n_G-1)-(n_G-1)c^2-n_G(n_G-2)c} \\
&~~+ \frac{(n_G-1)c^2(1-c^{n_L-j})}{(1-c)(n_G(n_G-1)-(n_G-1)c^2-n_G(n_G-2)c)}\\
&= \frac{\left( c+n_G \right)  \left( n_G-1 \right)(1-c) + (n_G-1)c^2(1-c^{n_L-j})}{(1-c)(n_G(n_G-1)-(n_G-1)c^2-n_G(n_G-2)c)}\end{split}\end{equation}
And the proof is complete.

\qed 
\end{proof}

For reference we include the complete inverse matrix once again. 
\begin{displaymath}
(\tens{I}-c\tens{A}^\top)^{-1} = \begin{bmatrix}
\tens{B}_B^{inv} & \tens{B}_C ^{inv} & \tens{C}_1^{inv} \\
\tens{B}_D^{inv} & \tens{B}_E^{inv} & \tens{C}_2^{inv} \\
\tens{D}_1^{inv} & \tens{D}_2^{inv} & \tens{E}^{inv}
\end{bmatrix} \end{displaymath}
\begin{displaymath}\tens{B}_B^{inv} = \begin{bmatrix}
1 & c & c^2 & \ldots & c^{j-2}\\
0 & 1 & c & \ldots & c^{j-3}\\
0 & 0 & 1 & \ddots & \vdots\\
\vdots & \vdots & \ddots & \ddots & c \\
0 & 0 & \ldots & 0 & 1
\end{bmatrix}, ~ \tens{B}_C^{inv} = \begin{bmatrix}
\frac{c^{j-1}s_A}{n_G}  & \frac{c^{j-1}k_A}{n_G} & \frac{c^{j-1}k_A}{n_G} & \ldots & \frac{c^{j-1}k_A}{n_G} \\
\frac{c^{j-2}s_A}{n_G}  & \frac{c^{j-2}k_A}{n_G} & \frac{c^{j-2}k_A}{n_G} & \ldots & \frac{c^{j-2}k_A}{n_G} \\
\vdots & \vdots & \vdots & \ldots & \vdots \\
\frac{cs_A}{n_G} & \frac{ck_A}{n_G} & \frac{ck_A}{n_G} & \ldots & \frac{ck_A}{n_G} \\
\end{bmatrix}  \end{displaymath}
\begin{displaymath} \tens{B}_D^{inv} = \begin{bmatrix}
0 & \ldots & 0 \\
\vdots & \ddots & \vdots \\
0 & \ldots & 0
\end{bmatrix},~  \tens{B}_E^{inv} = \begin{bmatrix}
s_A & k_A & k_A & \ldots & k_A\\
\frac{bk_A}{a} & s_D & k_D & \ldots & k_D\\
\frac{bk_A}{a} & k_D & s_D & \ddots & k_D\\
\vdots & \vdots & \ddots & \ddots & \vdots \\
\frac{bk_A}{a} & k_D & k_D & \ldots & s_D \end{bmatrix} \end{displaymath}
%
%
\begin{displaymath}s_A =  1-(n_G-1)bk_A, ~ s_D =  1-(n_G-2)ak_D\end{displaymath}
Where $\tens{B}$ is taken from the matrix proof of Theorem~\ref{thm:line_comp_fg} where $\tens{B}$ is the part of the graph consisting of a line with $j-1$ nodes and a complete graph with $n_G$ nodes with the first node in the complete graph linking to node $j-1$ in the line. 
\begin{displaymath}\tens{C}^{inv} = \begin{bmatrix} 
\tens{C}_1^{inv} \\
\tens{C}_2^{inv}
\end{bmatrix} = \begin{bmatrix} 
\frac{-c^{j-1}cs_A}{n_G} & \frac{-c^{j-1}c^2s_A}{n_G} & \ldots & \frac{-c^{j-1}c^{n_L-j}s_A}{n_G} \\
\frac{-c^{j-2}cs_A}{n_G} & \frac{-c^{j-2}c^2s_A}{n_G} & \ldots & \frac{-c^{j-2}c^{n_L-j}s_A}{n_G} \\
\vdots & \vdots & \ldots & \vdots \\
\frac{-c^{1}cs_A}{n_G} & \frac{-c^{1}c^2s_A}{n_G} & \ldots & \frac{-c^{1}c^{n_L-j}s_A}{n_G} \\
{-cs_A} & {-c^2s_A} & \ldots & {-c^{n_L-j}s_A} \\
\frac{-cbk_A}{a} & \frac{-c^2bk_A}{a}  & \ldots & \frac{-c^{n_L-j}bk_A}{a} \\
\frac{-cbk_A}{a} & \frac{-c^2bk_A}{a}  & \ldots & \frac{-c^{n_L-j}bk_A}{a} \\
\vdots & \vdots & \ldots & \vdots \\
\frac{-cbk_A}{a} & \frac{-c^2bk_A}{a}  & \ldots & \frac{-c^{n_L-j}bk_A}{a} 
\end{bmatrix} \end{displaymath}
%
%
\begin{displaymath}\tens{D}^{inv} = [\tens{D}_1^{inv} \ \tens{D}_2^{inv}] = \begin{bmatrix}
0 & \ldots & 0 \\
\vdots & \ddots & \vdots \\
0 & \ldots & 0
\end{bmatrix}, \quad E^{inv} =  \begin{bmatrix}
1 & c & c^2 & \ldots & c^{n_L-j-1}\\
0 & 1 & c & \ldots & c^{n_L-j-2}\\
0 & 0 & 1 & \ddots & \vdots \\
\vdots & \vdots & \ddots & \ddots & c \\
0 & 0 & \ldots & 0 & 1
\end{bmatrix} \end{displaymath}
%
\begin{theorem}
The normalizing constant $N$ for the simple line with one node being part of a complete graph using uniform $\vec{u}$ can be written as: 
\begin{equation}\begin{split}
N = \left( n_G-1 \right) {\vec{R}^{(2)}_{G,i \notin L}[S_L \leftrightarrow S_G]+{\vec{R}^{(2)}_{L,j}[S_L \leftrightarrow S_G]}}
+ {\frac {n_L-1}{1-c}}\\-{\frac {c
 \left( 1-{c}^{n_L-j} \right) }{ \left( 1-c \right) ^{2}}}-{\frac {c \left( 1-{c}^{j-1} \right) }{ \left( 1-c \right) ^{2}}
}+{\frac {c \left( 1-{c}^{j-1} \right) {\vec{R}^{(2)}_{L,j}[S_L \leftrightarrow S_G]}}{n_G \left( 1-c \right) 
}}\end{split}\end{equation}
where $\vec{R}^{(2)}_{G,i}[S_L \leftrightarrow S_G]$ is the PageRank of nodes in the complete graph, \\$\vec{R}^{(2)}_{L,j}[S_L \leftrightarrow S_G]$ is for the node in both the line and complete graph and \\ $\vec{R}^{(2)}_{L,i}[S_L \leftrightarrow S_G]$ is for the nodes in the line.
\end{theorem}

\begin{proof}
The normalizing constant is equal to the sum of the non-normalized PageRank of all nodes. 

We got $n_G$ nodes in the complete graph, $(n-1)$ not directly connected to the line and one connected to the line. This gives: 
\begin{equation}
N = (n-1)\vec{R}^{(2)}_{G,i}[S_L \leftrightarrow S_G] +\vec{R}^{(2)}_{L,j}[S_L \leftrightarrow S_G] +\sum_{i \neq j}{\vec{R}^{(2)}_{L,j}[S_L \leftrightarrow S_G]}\end{equation}
where $\vec{R}^{(2)}_{L,j}[S_L \leftrightarrow S_G]$ is the PageRank of individual nodes in the line except for the node node $j$ in the line for which we have $\vec{R}^{(2)}_{L,j}[S_L \leftrightarrow S_G]$.
For those nodes we got PageRank: 
\begin{equation}\vec{R}^{(2)}_{L,j}[S_L \leftrightarrow S_G] = \left\{ \begin{split}
\frac{1-c^{n_L+1-i}}{1-c} ,~~ i>j\\
\frac{c^{j-i}\vec{R}^{(2)}_{L,j}}{n_G}+\frac{1-c^{j-i}}{1-c},~~ i < j \end{split} \right.\end{equation}
The sum of all nodes for which $i>j$ can be written: 
\begin{equation}
\sum_{i = j+1}^{n_L}{\vec{R}^{(2)}_{L,j}[S_L \leftrightarrow S_G]} = {\frac {n_L-j}{1-c}}-{\frac {c \left( 1-{c}^{n_L-j} \right) }{ \left( 1-c
 \right) ^{2}}}\end{equation}
where we use that the second part $\sum_{i=j}^{n_L}{\frac{-c^{n_L+1-i}}{1-c}}$ is a geometric sum. Calculating the sum for $i<j$ in the same way we get:
\begin{equation}
\sum_{i = 1}^{j-1}{\vec{R}^{(2)}_{L,j}[S_L \leftrightarrow S_G]} = {\frac {j-1}{1-c}}-{\frac {c \left( 1-{c}^{j-1} \right) }{ \left( 1-c
 \right) ^{2}}}+{\frac {c \left( 1-{c}^{j-1} \right) {\it \vec{R}^{(2)}_{L,j}}}{n_G
 \left( 1-c \right) }}\end{equation}
Summation of all individual parts completes the proof. 

\qed 
\end{proof}

Now that we have an explicit formula for this example we can look at what happens when we change various parameters like $c$ or the size of either the line or complete graph. 

\subsection{A closer look at the formulas for PageRank in our examples} \label{closer_look_formula}
Now that we have formulas for the PageRank of a couple different graph structures we are going to take a short look at what happens when we change some parameters. We will also take a look at the partial derivative with respect to $c$. 
\subsubsection{Partial derivatives with respect to c}
In the case of the simple line with formula as seen earlier we get the derivative with respect to $c$ as: 
\begin{equation}
\frac{\partial}{\partial c}\vec{R}^{(2)}_{i}[S_L] = \left( 1-c \right) ^{-2}-{\frac {{c}^{n_L-i+1} \left( n_L-i+1 \right) }{c
 \left( 1-c \right) }}-{\frac {{c}^{n_L-i+1}}{ \left( 1-c \right) ^{2}}}\end{equation}
Rewriting it and looking to see if it is positive we get:
\begin{equation}{\frac {1+{c}^{n_L-i}(i-n_L)(1-c)}{
 \left( -1+c \right) ^{2}}} \geq 0 \Leftrightarrow c^{n_L-i}((i-n_L)(1-c) +\frac{1}{c^{n_L-i}}) \ge 0\end{equation}
\begin{equation}\Leftrightarrow \frac{1}{c^{n_L-i}} \ge (n_L-i)(1-c) \Leftrightarrow \frac{1}{1-c} \ge (n_L-i)c^{n_L-i} \Leftrightarrow 
\sum_{k=0}^{\infty}{c^k} \ge (n_L-i)c^{n_L-i} \end{equation}
Since we have $0<c<1,~ n_L\ge i$ we have that $c^k > c^{k+1}$ the first $n_L-i$ elements of the left sum is at least as large as $c^{n_L-i}$, this gives: 
\begin{equation}\sum_{k=0}^{\infty}{c^k} \ge \sum_{k=1}^{n_L-i}{c^{k}} \ge (n_L-i)c^{n_L-i}\end{equation}

For our case with a line connected to a complete graph by letting one node in the complete graph be part of the line we get the following derivative with respect to $c$: 
\begin{displaymath}
\frac{\partial}{\partial c} \vec{R}^{(2)}_{L,j}[S_L \leftrightarrow S_G] \end{displaymath}
\begin{equation}\begin{split}= {\frac { \left(  \left(  \left( -1+c \right) {n_G}^{2}+ \left( -1+c
 \right) ^{2}n_G-{c}^{2} \right)  \left(  \left( -1+c \right) n_G-2c+1
 \right) {\frac {\partial}{\partial c}}G \left( c \right)\right)n_G}{ \left(  \left( -1+c \right) {n_G}^{2}+ \left( -1+c \right) ^{2}n_G-{c}^{2} \right) ^{2}}} \\
\frac{-\left( \left( n_G-1 \right) 
 \left( c \left(  \left( -2+c \right) n_G+2-2c \right) G \left( c
 \right) - \left( n_G-1 \right)  \left( n_G+{c}^{2} \right)  \right) 
 \right) n_G}{ \left(  \left( -1+c \right) {n_G}^{2}+ \left( -1+c \right)^{2}n_G-{c}^{2} \right) ^{2}} \end{split}\end{equation}
\begin{displaymath} G(c) =  \frac{1-c^{n_L-j+1}}{1-c}\end{displaymath}
\begin{displaymath}\frac{\partial}{\partial c} G(c) = \left( 1-c \right) ^{-2}-{\frac {{c}^{n_L-j+1} \left( n_L-j+1 \right) }{c
 \left( 1-c \right) }}-{\frac {{c}^{L-j+1}}{ \left( 1-c \right) ^{2}}}
\end{displaymath}
The derivative have about the same shape as the original function. As $c$ gets large so does the derivative and as $n_G$ increases the slope get steeper for large $c$. 

Looking at the other nodes in the complete graph we get the derivative with respect to $c$ as: 
\begin{displaymath}
 \frac{\partial}{\partial c} \vec{R}^{(2)}_{G,i}[S_L \leftrightarrow S_G]  
\end{displaymath}
\begin{equation}\begin{split}
=\frac { \left( n_{G}-1 \right)  \left( 1-c \right) - \left( c-n_{G} \right)  \left( n_{G}-1 \right) +2 \left( n_{G}-1 \right) c \left( 1-{c}^{n_{{L}}-j} \right) }
{\left( 1-c \right) \left( n_{{G}} \left( n_{{G}}-1 \right) - \left( n_{{G}}-1 \right) {c}^{2}-n_{{G}} \left( n_{{G}}-2 \right) c \right)}  \\
-\frac {  \left( n_{{G}}-1 \right) {c}^{1+n_{{L}}-j} \left( n_{{L}}-j \right)}
{\left( 1-c \right) \left( n_{{G}} \left( n_{{G}}-1 \right) - \left( n_{{G}}-1 \right) {c}^{2}-n_{{G}} \left( n_{{G}}-2 \right) c \right)}  \\
 + \frac { \left( c+n_{{G}} \right)  \left( n_{{G}}-1 \right)  \left( 1-c \right) + \left( n_{{G}}-1 \right) {c}^{2} \left( 1-{c}^{n_{{L}}-j} \right) }
{ \left( 1-c \right) ^{2} \left( n_{{G}} \left( n_{{G}}-1 \right) - \left( n_{{G}}-1 \right) {c}^{2}-n_{{G}} \left( n_{{G}}-2 \right) c \right) } \\
- \frac { \left(  \left( c+n_{{G}} \right)  \left( n_{{G}}-1 \right)  \left( 1-c \right) + \left( n_{{G}}-1 \right) {c}^{2} \left( 1-{c}^{n_{{L}}-j} \right)  \right)  \left(2c  +\left(2-2c -n_{{G}}\right)n_{{G}}   \right) }
{ \left( 1-c \right)  \left( n_{{G}} \left( n_{{G}}-1 \right) - \left( n_{{G}}-1 \right) {c}^{2}-n_{{G}} \left( n_{{G}}-2 \right) c \right) ^{2}}
\end{split}\end{equation}

\subsubsection{Changes in the size of the complete graph for our last example}
When we change the size of the complete graph we can see for example what size would be the most effective for increasing ones PageRank. In all these examples we will use $n_L=10,j=6,c=0.85$ and $n_G$ will wary between $1$ and $50$.
First we note that the part above the complete graph is unaffected by the change of $n_G$. It is obvious however that as $n_G$ increases the normalizing constant in the normalized PageRank will likely get larger resulting in a lower PageRank as long as it is part of a small system. 

For the nodes in the complete graph except for the one thats part of the line we get the result in Fig.~\ref{pg}.
\begin{figure}[!hbt]
\begin{center}
\includegraphics[width=.4\textwidth]{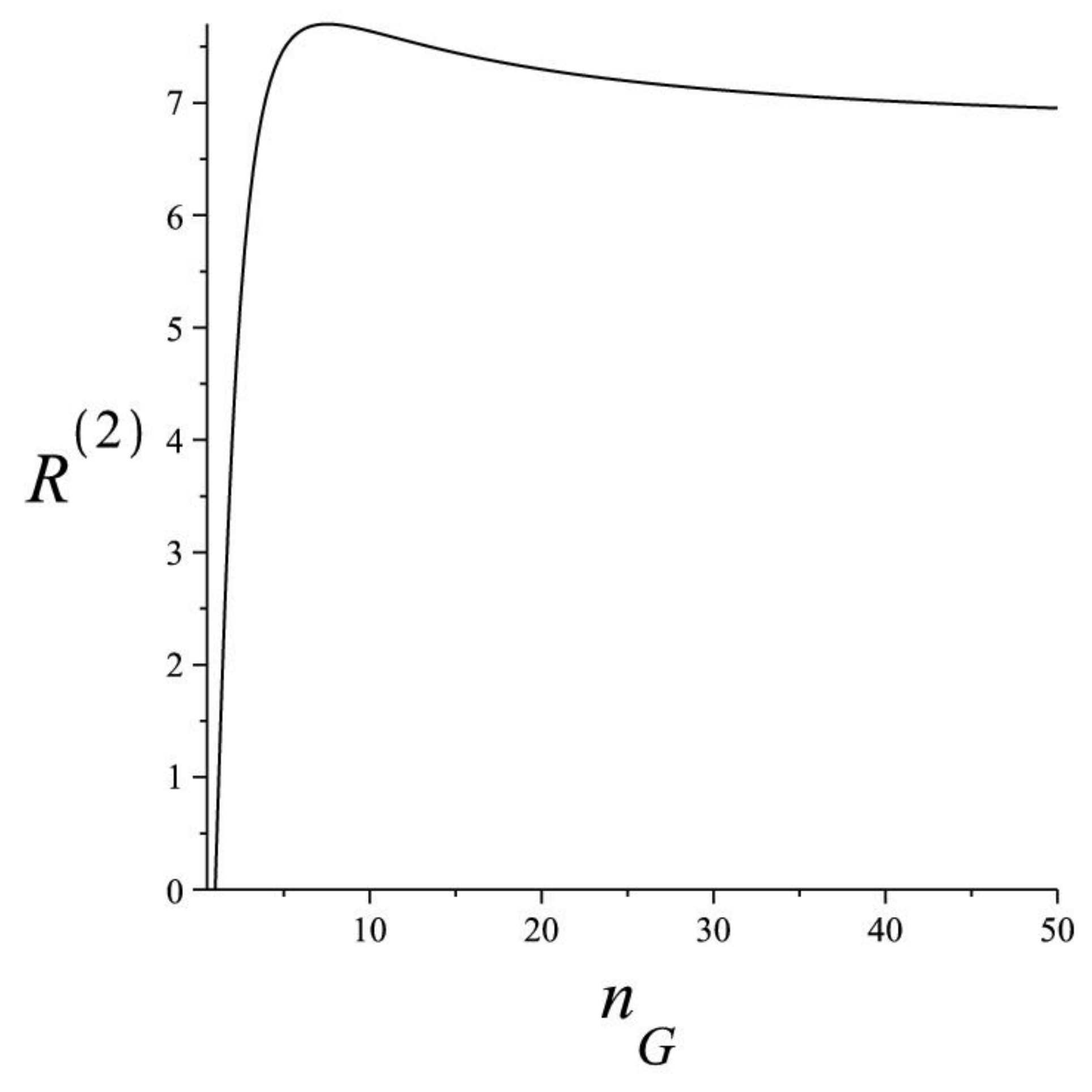}
\end{center}
\caption{$\vec{R}^{(2)}$ of the nodes in the complete graph not part of the line as a function of $n_G$}
\label{pg}
\end{figure} 

Looking at the function we can see two things, first of all a larger number of nodes in the graph will increase the rank of the nodes in it. We can also see a hint that it seems to be converging towards a value as $n_G$ gets large. Since the chance of escaping the graph decreases as $n_G$ increases we can expect it to eventually keep nearly all of it resulting in the PageRank of all the nodes in the complete graph approaching $1/(1-c) \approx 6.67, c=0.85$.

For the node in the complete graph thats part of the line as well we get the result in Fig.~\ref{pj}.
\begin{figure}[!hbt]
\begin{center}
\includegraphics[width=.4\textwidth]{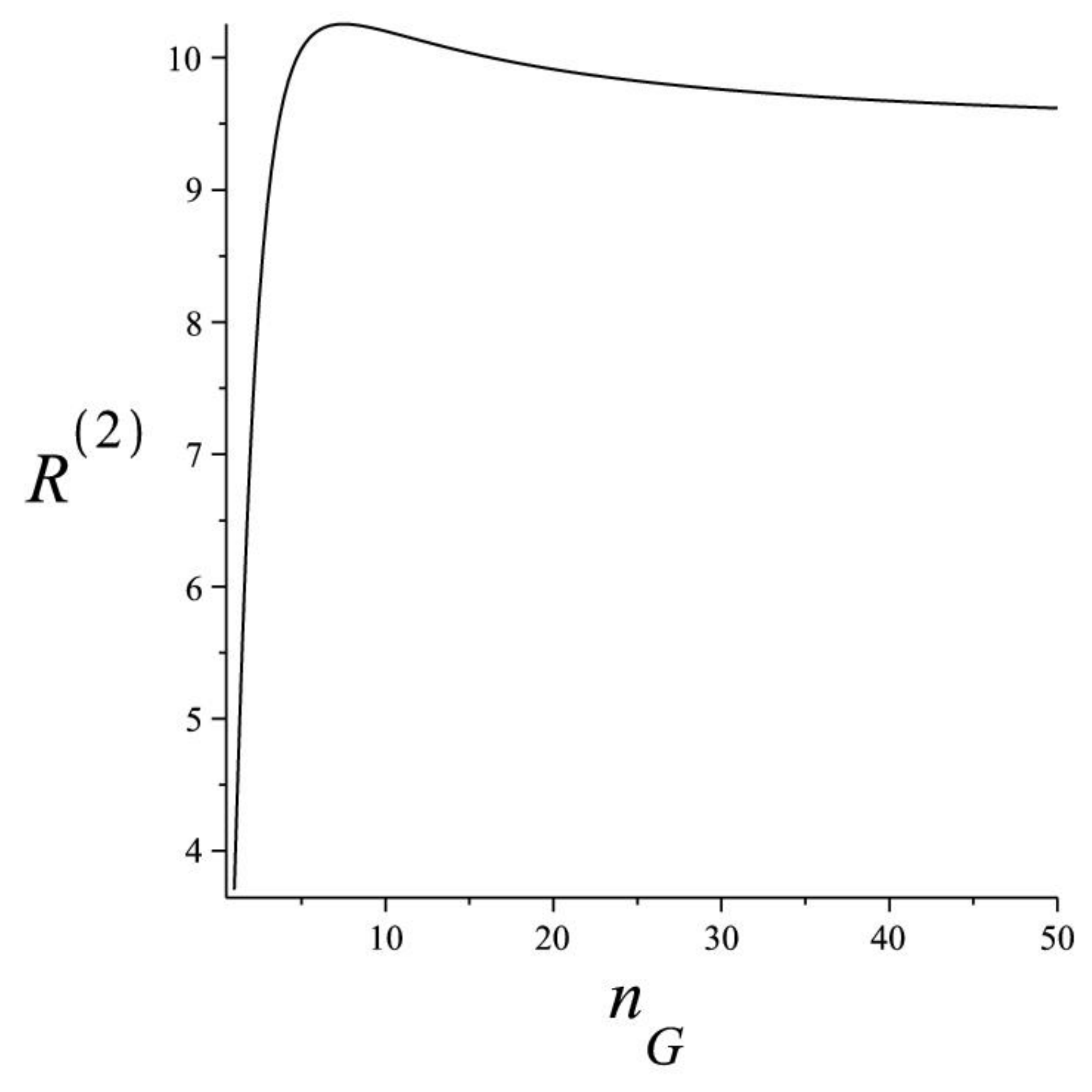}
\end{center}
\caption{$\vec{R}^{(2)}$ of the node in the complete graph thats part of the line as a function of $n_G$}
\label{pj}
\end{figure} 

Here we see something curious, the node seems to be gaining rank in the beginning while starting to fall after a while and possibly converging towards a value in the same way as the other nodes in the complete graph. The reason we get a local maximum is the fact that for a moderately large $n_G$ we maximize the probability of $\vec{R}^{(2)}_{L,j}$ getting back to itself while keeping the complete graph lare enough to keep most probability for itself. Here we can see that its not always a good idea for an individual node to join a complete graph. If the node in question already have larger PageRank than the other nodes in the complete graph it actually might lose PageRank from joining it. 

The result for the node below the complete graph we get the result in Fig.~\ref{pi5}.
\begin{figure}[!hbt]
\begin{center}
\includegraphics[width=.4\textwidth]{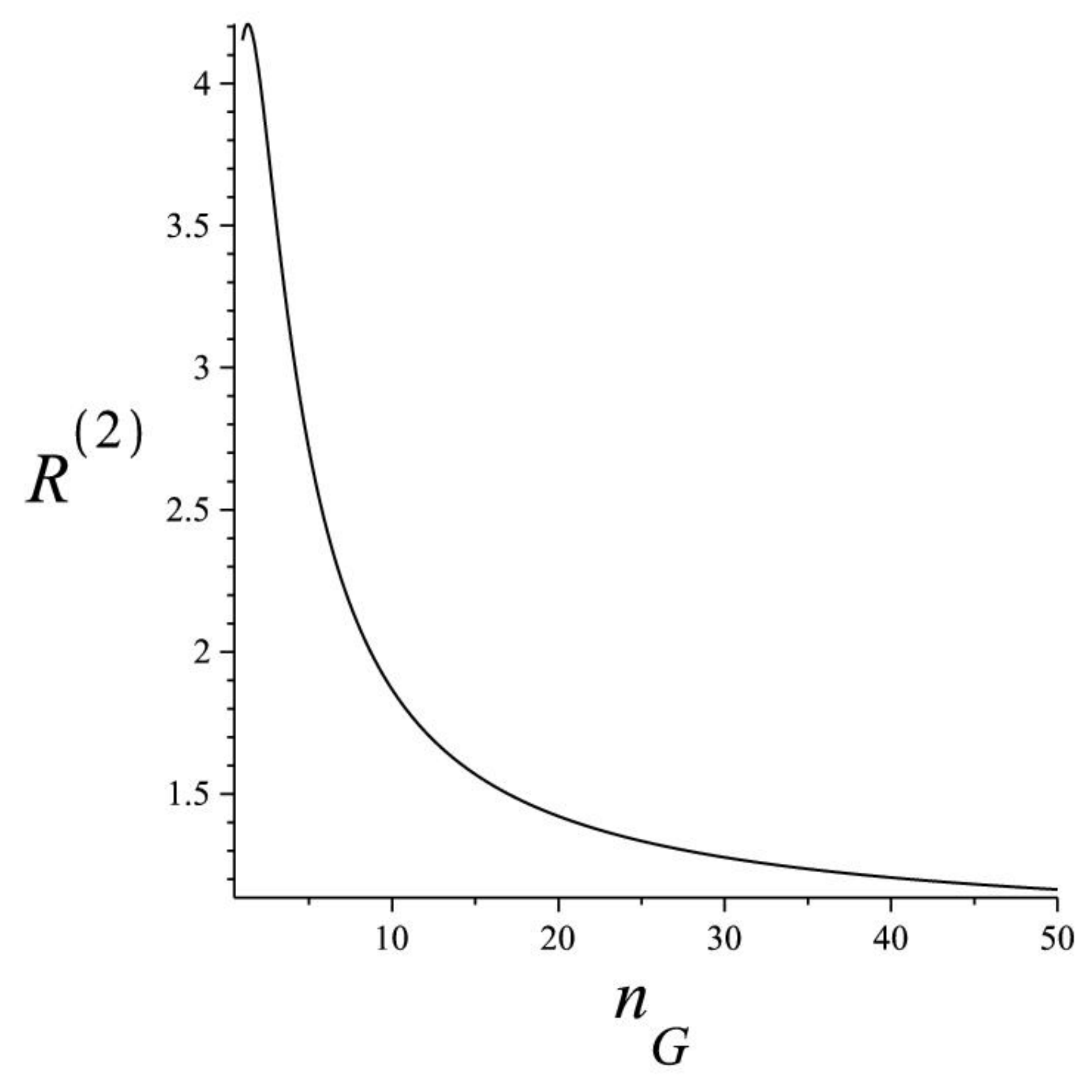}
\end{center}
\caption{$\vec{R}^{(2)}$ of the nodes in the line below the complete graph as a function of $n_G$}
\label{pi5}
\end{figure} 

Here we see the great loser as $n_G$ increases. Since the chance of escaping the complete graph depends on $\vec{R}^{(2)}_{S_G,j}[S_L \leftrightarrow S_G]/n_G$ as $n_G$ increases so does this nodes PageRank as well. From this we see a clear example of the effects of complete graphs on its surrounding nodes. A complete graph can be seen as a type of sink, all links to the complete graph will be used to maximum effect within the complete graph. And even worse, even if the complete graph have some nodes that point out of it their influence will be very small since the nodes in the complete graph having a large number of links the chance of escaping is low.

\subsection{A look at the normalized PageRank for the line connected with a complete graph}\label{formula_norm}
Looking at the normalized PageRank in our last example with a simple line with one node being part of a complete graph we want to see how the PageRank changes as $c$ or the relation between the size of the line or complete graph changes. 

\subsubsection{Dependence on c}
Plotting the PageRank with $n_G=10,n_L=10,j=6$ and $c \in [0.01,0.99]$ we get the following results. For the node just above the complete graph $\vec{R}^{(1)}_{S_L,i}[S_L \leftrightarrow S_G],i=7$ we get the result in Fig.~\ref{pi7N}.
\begin{figure}[!hbt]
\begin{center}
\includegraphics[width=.4\textwidth]{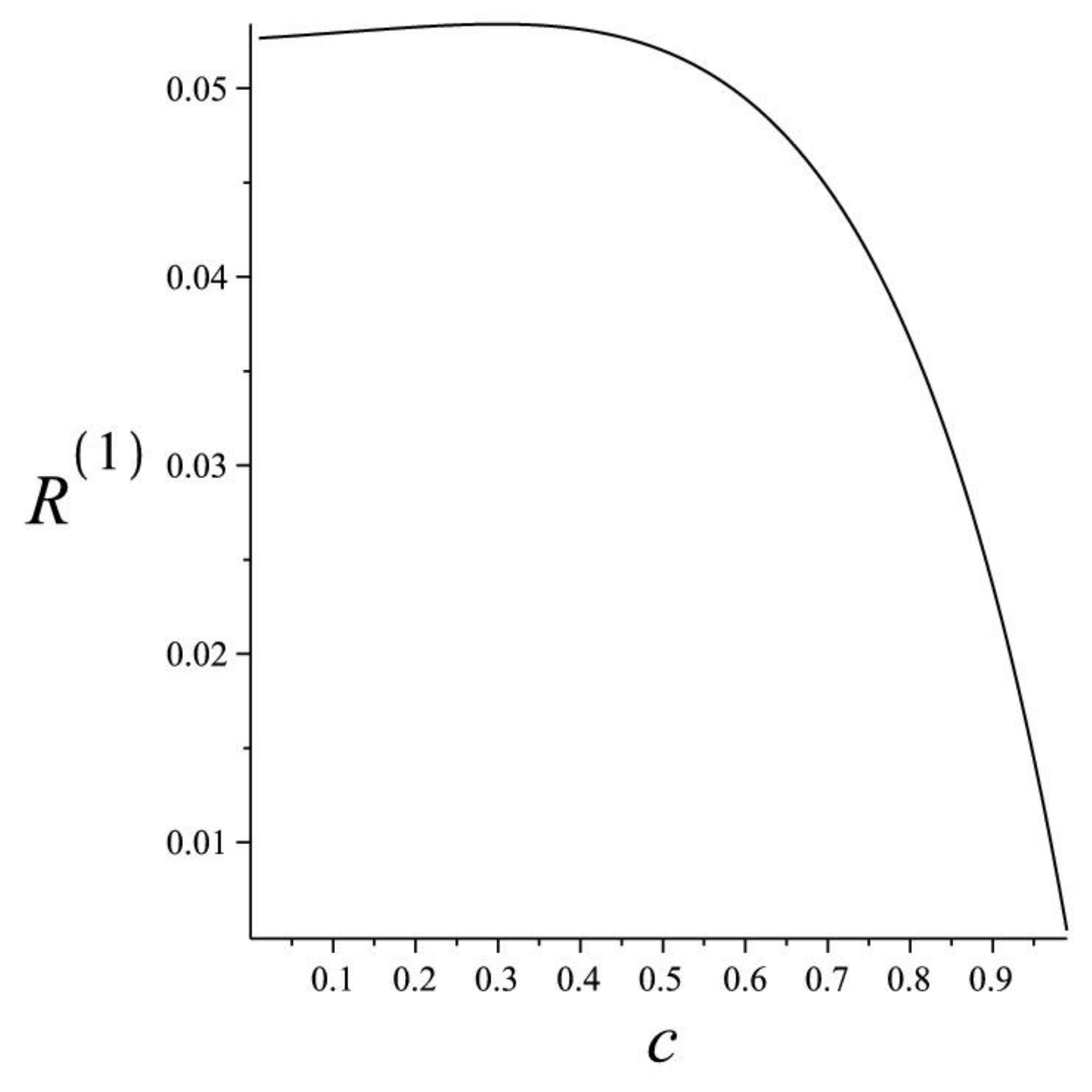}
\end{center}
\caption{$\vec{R}^{(1)}$ of the node above the complete graph as a function of $c$}
\label{pi7N}
\end{figure}

Here we see that the function seems to have a max at about $c=0.55$ after which it decreases faster the closer to $c=1$ it gets. We find the $c$ which maximize the function for some other different parameters $n_G,n_L,j,i$ in the table below. All the local max/min is calculated using the optimization tool in Maple $15$. 
\begin{table}
\caption{Maximum PageRank $\vec{R}^{(1)}$ of node $i$ "above" the complete graph depending on $c$ for various changes in the graph where one node in a simple line is part of a complete graph.}
\label{tab:1}       
\begin{tabular}{llllll}
\hline\noalign{\smallskip}
 $n_G$ 	&$n_L$& $j$&$i$&$c_{\text{max}}$&max\\
\noalign{\smallskip}\hline\noalign{\smallskip}
 5			&10&6&7& 0.349& 0.073\\
 10			&10&6&7& 0.300& 0.053\\
 20			&10&6&7& 0.248& 0.035\\
 10			&20&6&7& 0.751& 0.370\\
 20			&20&6&7& 0.721& 0.027\\
 50			&50&6&7& 0.874& 0.010\\
 10			&10&9&10& 0.000& 0.053\\
 10			&10&3&4& 0.515& 0.054\\
10			&10&6&9& 0.300& 0.053\\
\noalign{\smallskip}\hline
\end{tabular}
\end{table}

As seen the location of the maximum seems to be moving towards the left as $n_G$ increases and towards the right as $n_L$ increases. In the same manner it moves towards the left as $i$ get closer to $n_L$. The value of the maximum is only included out of completeness, it is natural that they decrease as either $n_G$ or $n_L$ increases as we in those cases get a larger number of total nodes in the system. It is interesting to note that the max seems to be going towards the right as both $n_G,n_L$ increases as well. 
Looking at the node in the line being a part of the complete graph we get the result in Fig.~\ref{pjN}.
\begin{figure}[!hbt]
\begin{center}
\includegraphics[width=.4\textwidth]{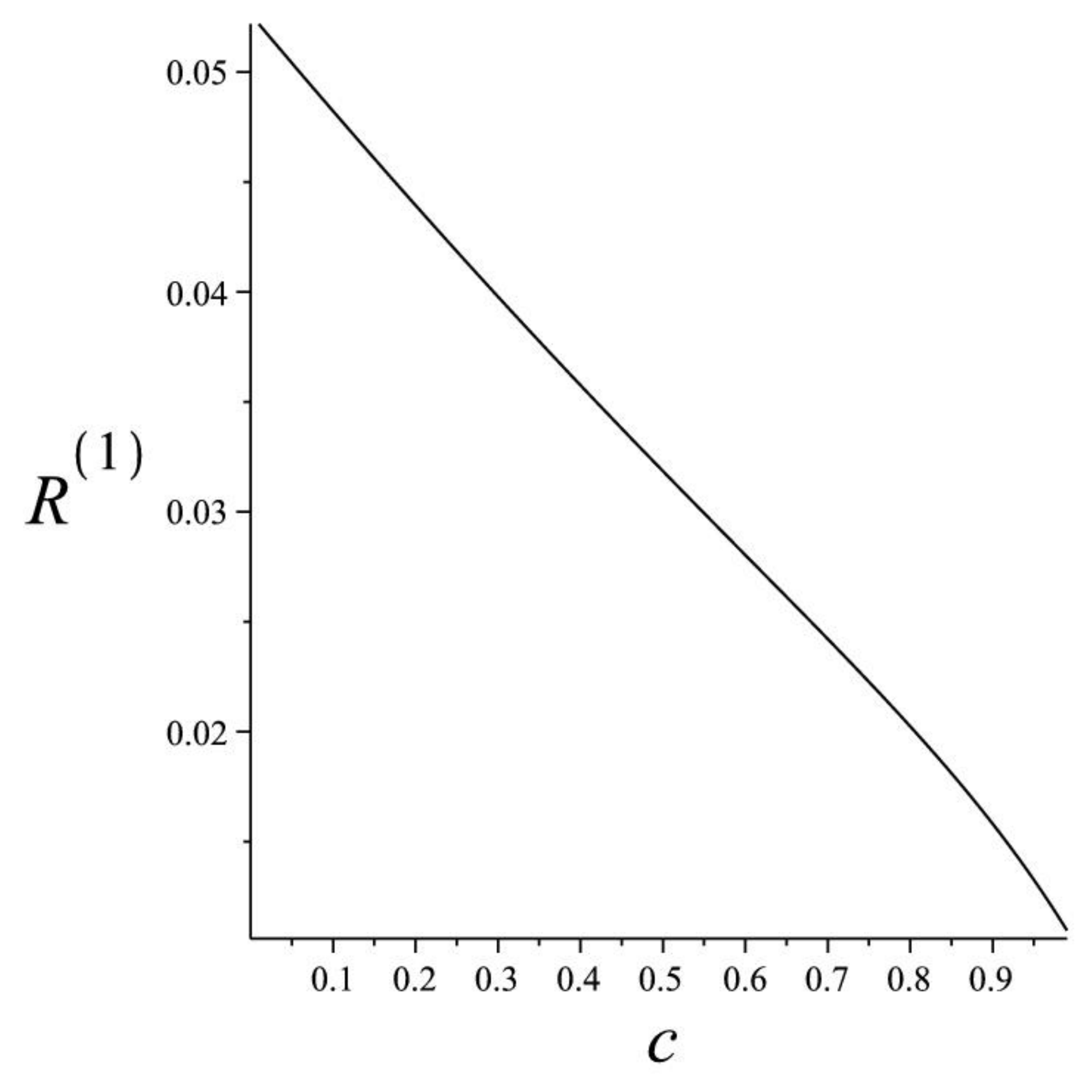}
\end{center}
\caption{$\vec{R}^{(1)}$ of the node in the line being a part of the complete graph as a function of $c$}
\label{pjN}
\end{figure}  
\begin{table}
\caption{Maximum PageRank $\vec{R}^{(1)}$ of node $j$ depending on $c$ for various changes in the graph where node $j$ in a simple line is part of a complete graph.}
\label{tab:2}       
\begin{tabular}{llllll}
\hline\noalign{\smallskip}
 $n_G$ 	&$n_L$& $j$&$c_{\text{max}}$&max\\
\noalign{\smallskip}\hline\noalign{\smallskip}
 5			&10&6& 1& 0.164\\
 10			&10&6& 0.894& 0.099\\
 20			&10&6& 0.776& 0.059\\
 10			&20&6& 1& 0.096\\
 20			&20&6& 0.929& 0.056\\
 50			&50&6& 0.965& 0.023\\
 10			&10&9& 1& 0.091\\
 10			&10&3& 0.893& 0.107\\
\noalign{\smallskip}\hline
\end{tabular}
\end{table}

Here we see the great "winner" as $c$ increases. Do note the difference in the axis for the different images, since this at its lowest point is actually about the same as the highest for the node above the complete graph. The PageRank of this node is the largest when $c$ is large, sometimes with a local maximum and sometimes not. It seems to be that as the number of nodes in the complete graph increases we are more likely to find a local maximum than not. For the node just below the complete graph we get the result in Fig.~\ref{pi5N}.
\begin{figure}[!hbt]
\begin{center}
\includegraphics[width=.4\textwidth]{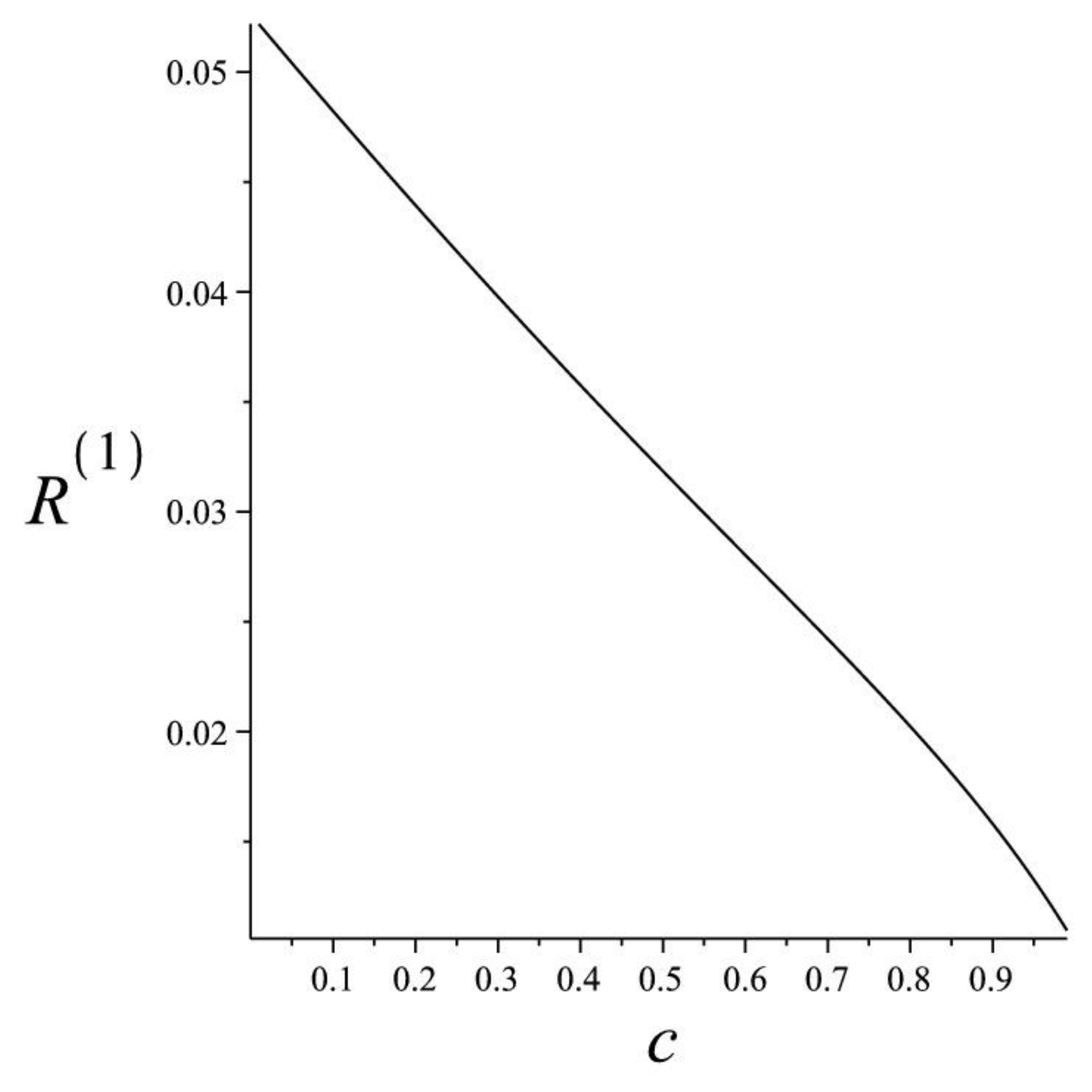}
\end{center}
\caption{$\vec{R}^{(1)}$ of the node below the complete graph as a function of $c$}
\label{pi5N}
\end{figure}

PageRank decreases as $c$ increases, but compared to the nodes above the complete graph not as fast for large $c$. This since the PageRank of the nodes in the complete graph increase so fast for large $c$ that even the comparatively small influence it have on the nodes out of it is enough to at least stop the extremely rapid loss of rank as for the nodes above the complete graph. 

Last we got the PageRank of the other nodes in the complete graph in Fig.~\ref{pgN}.
\begin{figure}[!hbt]
\begin{center}
\includegraphics[width=.4\textwidth]{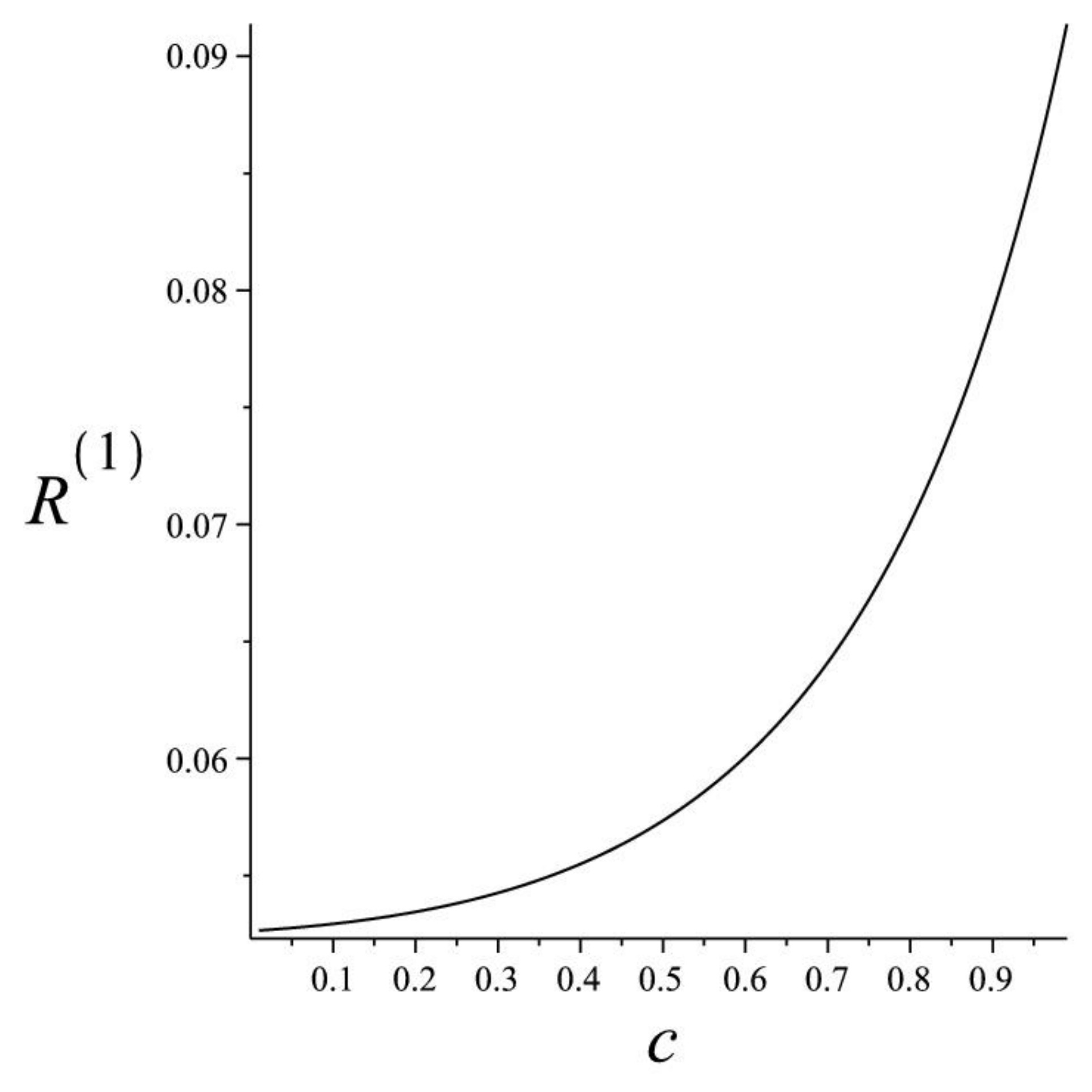}
\end{center}
\caption{$\vec{R}^{(1)}$ of a node in the complete graph as a function of $c$}
\label{pgN}
\end{figure} 

As with the node in both the line and complete graph, PageRank increases very fast for large $c$. We once again see a hint to why a to large $c$ could be problematic, it is for large $c$ we get the largest relative changes in PageRank between nodes. We have no min/max here, instead PageRank increases faster and faster the larger $c$ gets. 

We note that these local maximum and minimum are not always present. In these cases we have a PageRank thats decreasing as $c$ increases for the whole interval. If the one exist we can expect the other to as well (since we expect the rank to decrease at the end of the interval). It is hard to say anything conclusive about the location or existence of local maximum or minimum points, but we do note that they exist. There is also a large difference in how PageRank changes for different (especially large) $c$, we can therefor expect $c$ to have an effect not only in the final rank and the computational aspect, but also the final ranking order of pages. 

\subsubsection{A look at the partial derivatives with respect to c}\label{ddcpN}
Since we have the formulas for the normalized PageRank it is also possible to find the partial derivatives. Since the partial derivatives result in very large expressions (multiple pages each) they are not included here. By setting $n_G=n_L=10,j=6$ we get the result after taking the partial derivative with respect to $c$ for $0.05<c<0.95$ for the node $e_{L,7}$ above the complete graph in Fig.~\ref{ddcnormPI7}. 
\begin{figure}[!!hbt]
\begin{center}
\includegraphics[width=.38\textwidth]{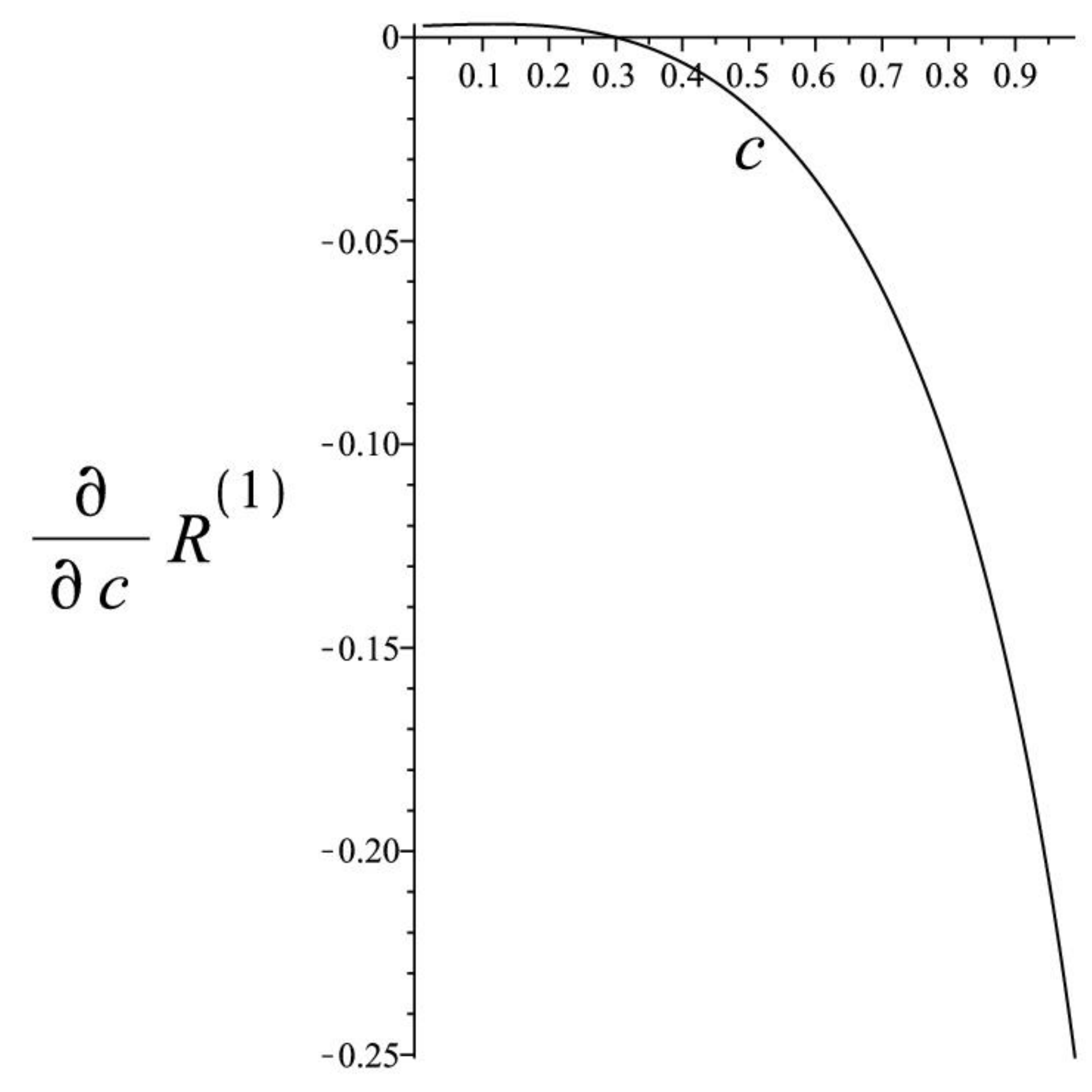}
\end{center}
\caption{Partial derivative with respect to $c$ of normalized PageRank of the node in the line above the complete graph}
\label{ddcnormPI7}
\end{figure} 

We see the derivative falling faster as $c$ increases. Here as well we see the more dramatic changes in large $c$ above about $0.8$. Apart from seeing the maximum at around $c=0.3$ in the original function we can also see that the derivative seems to briefly increase in the beginning, reaching a maximum at about $c=0.1$. For the node part of both the line and the complete graph we get the result in Fig.~\ref{ddcnormPJ}.
\begin{figure}[!!hbt]
\begin{center}
\includegraphics[width=.38\textwidth]{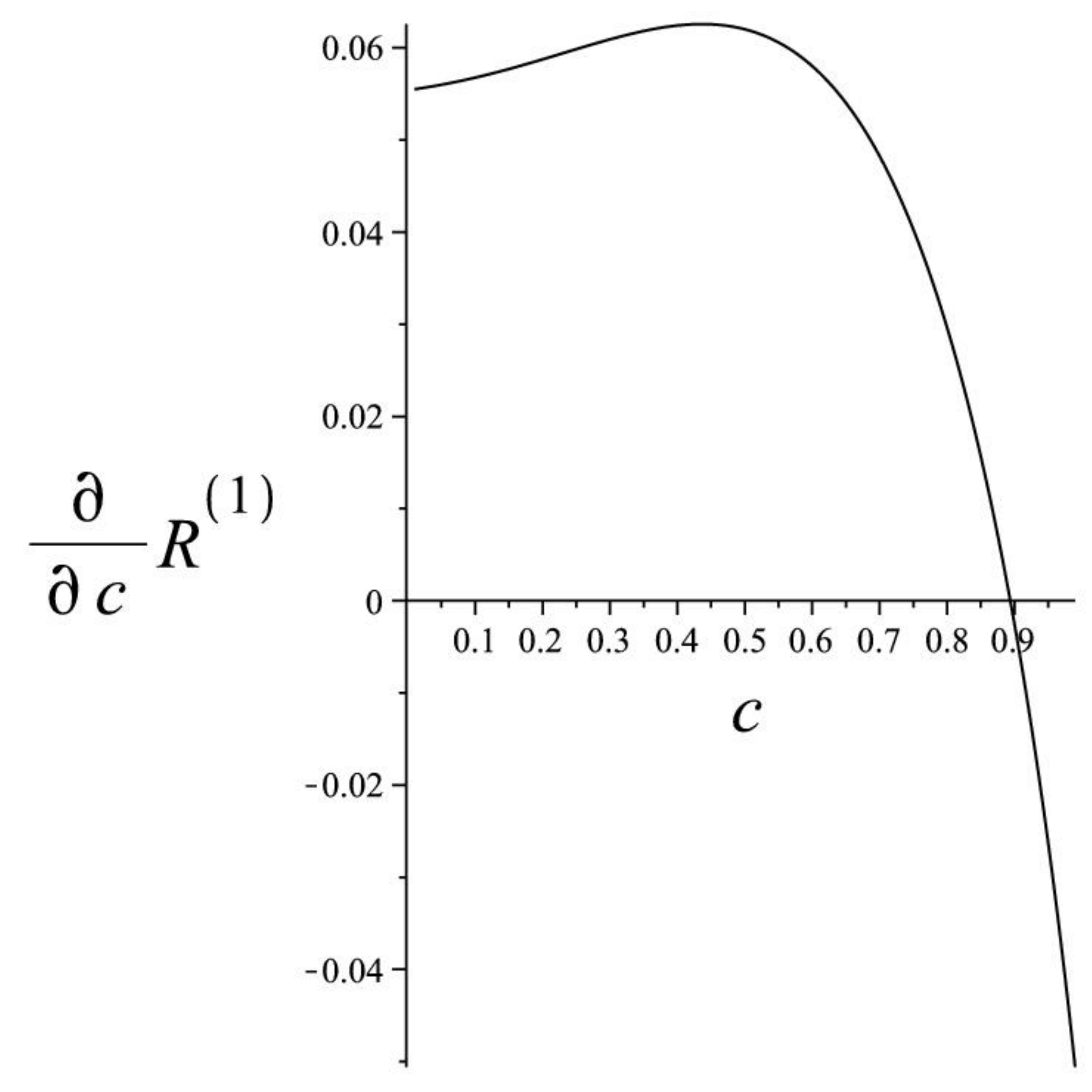}
\end{center}
\caption{Partial derivative with respect to $c$ of normalized PageRank of the node in the line part of the complete graph}
\label{ddcnormPJ}
\end{figure} 

We can see a high derivative all the way until we get to very large $c$ where it finally starts going down. We can clearly see the maximum at about $c\approx 0.9$ in the original function. For the node on the line below the complete graph we get the result in Fig.~\ref{ddcnormPI5}.
\begin{figure}[!!hbt]
\begin{center}
\includegraphics[width=.4\textwidth]{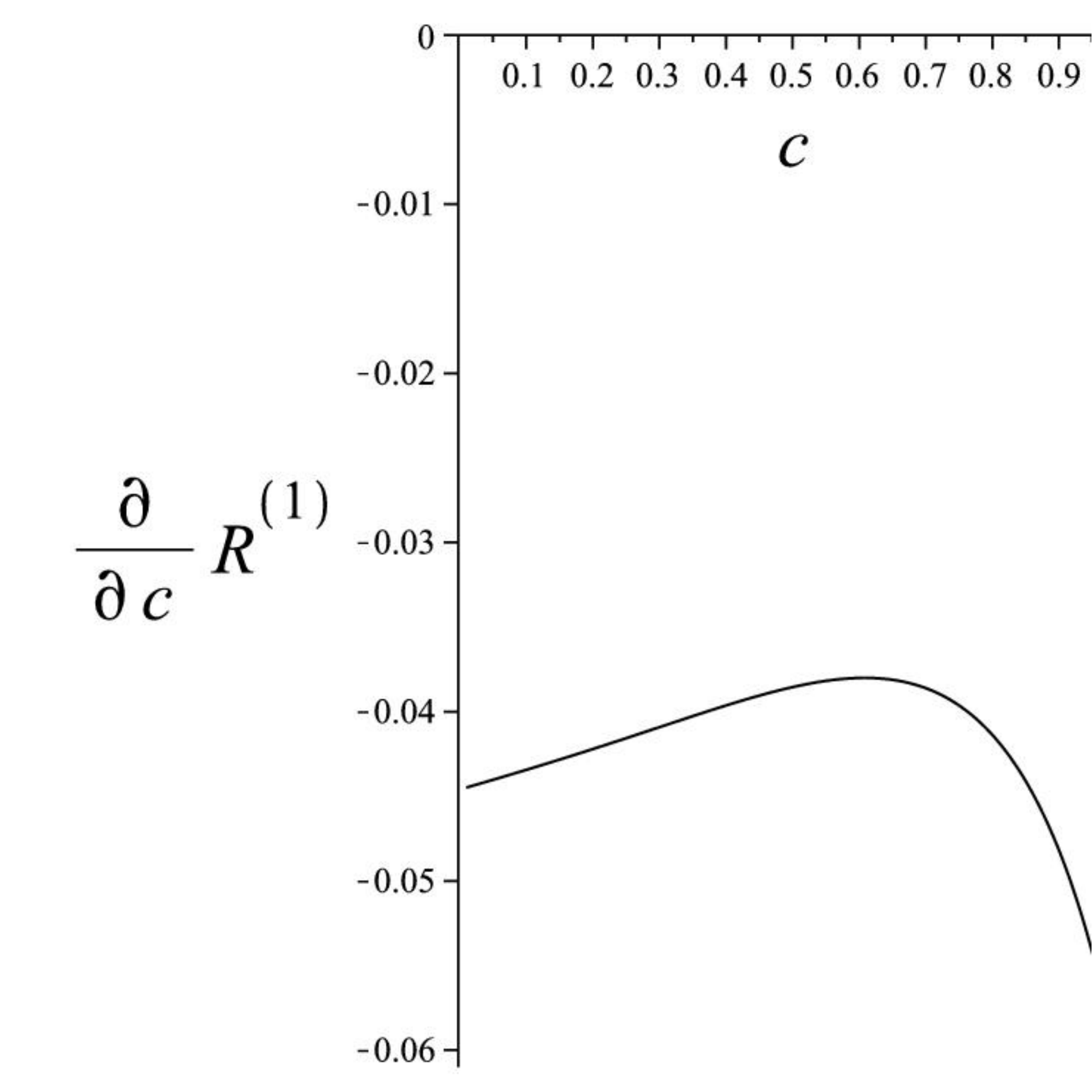}
\end{center}
\caption{Partial derivative with respect to $c$ of normalized PageRank of the node in the line below the complete graph}
\label{ddcnormPI5}
\end{figure} 

Although the derivative is decreasing for all $c$, the derivative have a local maximum at about $c\approx 0.6$.Worth to note is that the axis can be a little misleading, the partial derivative is in fact not that close to $0$ at the local maximum. As before the largest changes are at high $c$. Worth to note that the derivative is decreasing for all $c$. For the nodes in the complete graph not part of the line we get the result found in Fig.~\ref{ddcnormPG}
\begin{figure}[!!hbt]
\begin{center}
\includegraphics[width=.4\textwidth]{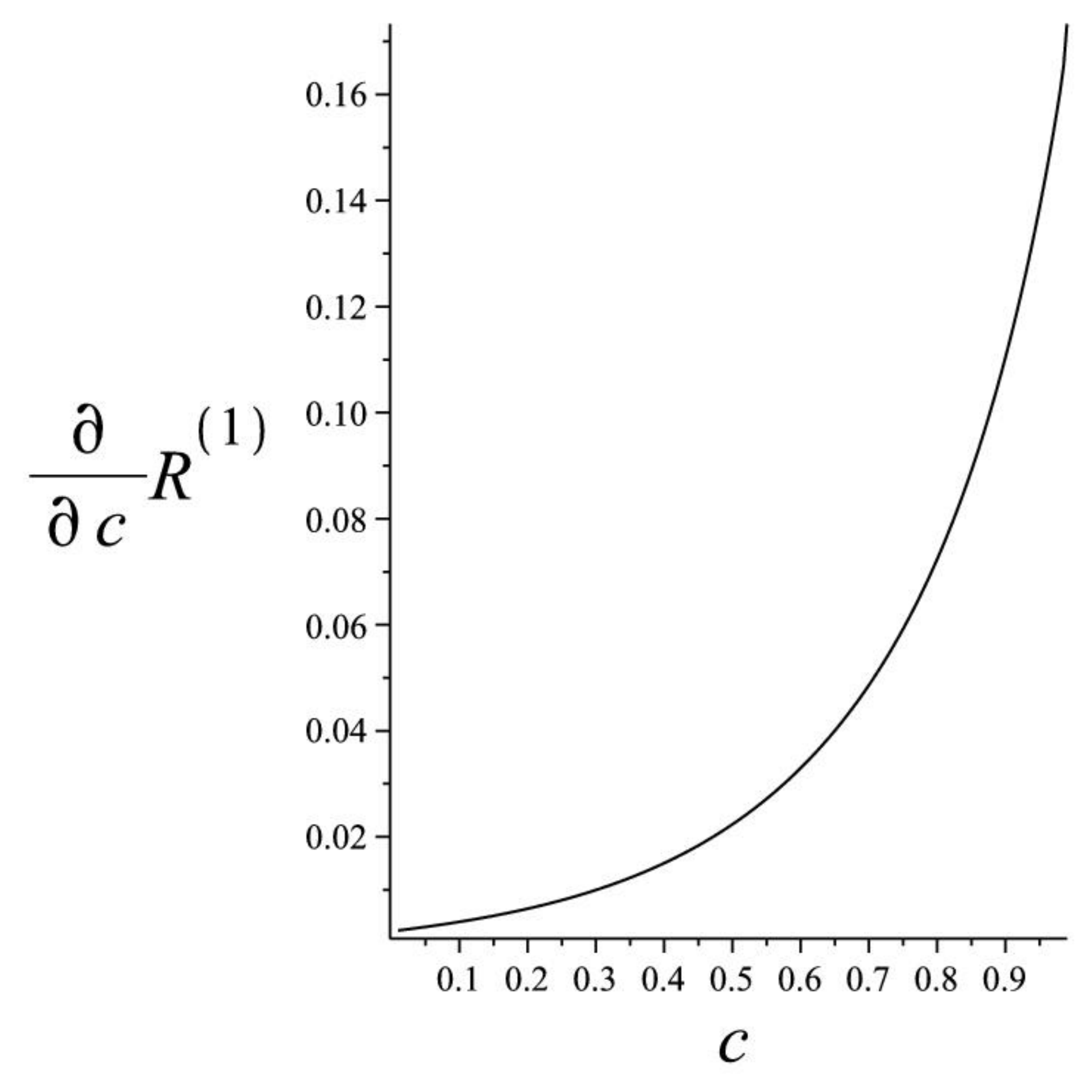}
\end{center}
\caption{Partial derivative with respect to $c$ of normalized PageRank of a node in the complete graph not part of the line}
\label{ddcnormPG}
\end{figure} 

As before the largest changes are found at large $c$. Compared to the node part of both the complete graph and the line the derivative for the ones only in the complete graph continue to increase as $c$ increases, however the PageRank itself is not actually ever higher for the ones not part of the line. We have seen that although it is possible to find symbolic expressions for the PageRank and derivative for some simple graphs, as the complexity of the graph increases it becomes very hard to do. Already for these simple examples the partial derivatives a rather large and complicated expressions.  Finding more general symbolic expressions for when the derivative is $zero$ should be possible although problematic given the constraints and size of the problem. 

\subsection{The effect of changing the weight vector V}\label{diffV}
From the equation system we see that the inverse matrix in the solution $\vec{R}^{(2)}$ does not depend on $\vec{V}$. While there is usually the system matrix $\tens{A}$ that rapidly changes making calculating PageRank in this way unpractical since we need to calculate the inverse of a huge matrix when doing changes (in the case of Internet pages). If we instead would have a mostly static system but with varying weight vector $V$ it might be useful to use this representation instead since calculating the new PageRank would then be a simple matrix-vector multiplication. We can also see that changing an element in $\vec{V}$ to zero from the uniform case $1/n_S$ has two effects. First of all we know that it would change the nodes PageRank by at least a constant amount ($1$ from uniform $\vec{V}$) the rest of the probability $c/1-c \approx 5.67 , c=0.85$ might be lost from other nodes in the vicinity. If the node has no outgoing links the PageRank of all other nodes will be unaffected. In the other case where the node in question is not a dangling node nor is it possible to reach any dangling nodes from it, all of it will be removed somewhere. We find the maximum that can be lost by setting a nodes weight to zero assuming a previous weight of $1$ as $1$ plus what we get if all nodes it link to link directly back to it and do not link to anything else as: 
\begin{equation}\sum_{k=0}^{\infty}{c^{2k}} = \frac{1}{1-c^2} \approx 3.6 ,\quad c=0.85 \end{equation}
In the same way doubling $V_i$ for one node increases the PageRank of those same nodes by the same amount as they would otherwise lose hade we instead set it to zero.

Especially effective it seems to simply change $\vec{V}$ for nodes in a complete graph if they are believed to be cheating, since the complete graph is so effective in keeping its probability to itself changing $\vec{V}$ to zero for those nodes should have a very little effect in surrounding nodes apart from possibly scaling the PageRank for all the nodes in the system with a different constant in the case of the normalized PageRank $\vec{R}^{(1)}$.   

\subsection{A comparison of normalized and non-normalized PageRank}
Here we will take a short look at the difference between normalized $(\vec{R}^{(1)})$ and non normalized $(\vec{R}^{(2)})$ PageRank in order to get a bigger understanding of the differences between them. We already know that $\vec{R}^{2} \propto \vec{R}^{(1)}$ so there will always be the same relation between the PageRank of two nodes. Here we will take a look at how the absolute difference between nodes and the two types of PageRank differ instead. 

Since the PageRank is normalized to one in $\vec{R}^{(1)}$ we obviously get that the PageRank will decrease as the number of nodes increases, potentially making for problems with number-representation for extremely large graphs unless it is taken into account when making the implementation. This problem is not as large a problem for $\vec{R}^{(2)}$ since most nodes will have approximately the same size regardless of the size of the graph. However the possible huge relative difference between nodes is still needed to take into consideration.. We note however that with the current way to calculate $\vec{R}^{(2)}$ by solving the equation system such large systems that could potentially be a problem in $\vec{R}^{(1)}$ is simply to large for us to solve in a timely manner. 

We also have one other main difference between the normalized and non-normalized PageRank and that is with dangling nodes and how they effect the global PageRank. In $\vec{R}^{(2)}$ a dangling nodes means some of the "probability" escape the graph resulting in a lower total PageRank (but still proportional to $\vec{R}^{(1)}$). In $\vec{R}^{(1)}$ however dangling nodes can be seen as linking to all nodes and in fact behaves exactly as if they did. We illustrate the difference in a rather extreme example with a graph composed of only four dangling nodes as well as a complete graph composed of four nodes. 

An image of the systems can be seen in Fig.~\ref{CompNormNNorm} below.
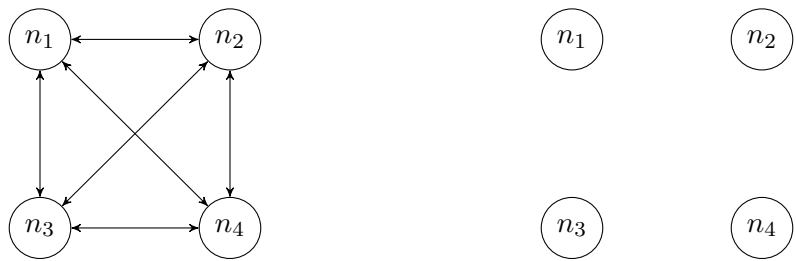
\begin{figure} [!hbt]
\begin{center}
 \begin{tikzpicture}[->,shorten >=0pt,auto,node distance=2.5cm,on grid,>=stealth',
every state/.style={circle,,draw=black,,minimum size=17pt}]
\node[state] (A)  {$n_1$};
\node[state] (B) [right=of A] {$n_2$};
\node[state] (C) [below=of A] {$n_3$};
\node[state] (D) [below =of B] {$n_4$};
\foreach \from/\to in {A/B,A/C,A/D,B/C,B/D,C/D}
\draw [<->] (\from) -- (\to);
\node[state] (E) [right =of B, xshift = 2cm] {$n_1$};
\node[state] (F) [right=of E] {$n_2$};
\node[state] (G) [below=of E] {$n_3$};
\node[state] (H) [below =of F] {$n_4$};
\end{tikzpicture}
\end{center}
\caption{A complete graph (left) and a system made of four dangling nodes (right)} 
\label{CompNormNNorm}
\end{figure}
When computing $\vec{R}^{(1)}$ of both systems assuming uniform weightvector $\vec{u}$ they are both obviously equal with PageRank $\vec{R}^{(1)} = [1/4,1/4,1/4,1/4]$, it does not even matter what $c$ we chose as long as it is between zero and one for convergence. However for the non normalized PageRank we get a large difference between the PageRank of the two systems where we for the complete graph get the PageRank $\vec{R}^{(2)}_a = [1/1-c, 1/1-c, 1/1-c, 1/1-c ]$ as seen in Sect.~\ref{CCompGraph}. However for the graph made up of only dangling nodes we get the PageRank $\vec{R}^{(2)}_b = [1,1,1,1]$ regardless of $c$. We see that while they might be proportional to each other, the non normalized version behaves differently for dangling nodes making a distinction between dangling nodes and nodes that link to all nodes (including itself which we normally do not allow). While this distinction might seem unnecessary since nodes that link to all nodes do not normally exist or similar nodes such as a node that links to all or most other nodes should either be extremely uncommon or plain do not exist as well, this might not be the case if working with smaller link structures where such a distinction might be useful. It is also this distinction that makes it possible to make comparisons of PageRank between different systems in $\vec{R}^{(2)}$ while not generally possible in $\vec{R}^{(1)}$.

\section{Conclusions}
We have seen that we can solve the resulting equation system instead of using the definition directly or using the Power method. While this method is significantly slower it has made it possible to get a bigger understanding of the different roles of the link matrix $\tens{A}$ and the weight vector $\vec{u}$. We have seen how PageRank changes when doing some small changes in a couple of simple systems and when connecting said systems. For these systems we also found explicit expressions for the PageRank and in particular two ways to find these. Either by solving the equation system itself or by calculating:
\begin{displaymath}
\left( \sum_{e_i \in S,\  e_i \neq e_g}{P(e_i\rightarrow e_g)}+1\right) \left(\sum_{k=0}^{\infty}{(P(e_g \rightarrow e_g))^k}\right)\end{displaymath}
where $P(e_i\rightarrow e_g)$ is the sum of probability of all paths from node $e_i$ to node $e_g$ and the weight vector $\vec{u}$ is uniform.

Given the expressions for PageRank we looked at the results when changing some parameters. While it is hard to say anything specific, two things seem to be true overall: The most dramatic changes happens as $c$ get large, usually somewhere where $c>0.8$ some nodes get dramatically larger PageRank compared to the other. We also see that complete graphs, while not gaining a larger rank if the graph is larger, it becomes a lot more reliable (as in not as effected in changes of individual nodes) in keeping its large PageRank as the structure get larger. 

We saw that if using uniform $\vec{V}$ it is possible to split a large system $S$ into multiple disjoint systems $S_1,S_2,\ldots S_N$ it is possible to calculate $\vec{R}^{(2)}$ for every subsystem itself and they will not differ from $\vec{R}^{(1)}$ apart from a normalizing constant that is the same across all subsystems. This is a property we would like to if possible have when using the power method as well. This since it could potentially greatly reduce the work needed primary when doing updates in the system. 

For the last part we looked at what happens in $\vec{R}^{(2)}$ when changing the weight vector $\vec{V}$. Especially we could see some guaranteed change in the constant change and we could find an upper bound in how much total difference the change can have overall. Especially effective it seems to be in lowering the PageRank of nodes in complete graphs since they keep most of their probability to themselves.

\section{Acknowledgments}
This research was supported in part by the Swedish Research Council (621-
2007-6338), Swedish Foundation for International Cooperation in Research
and Higher Education (STINT), Royal Swedish  Academy of Sciences, Royal Physiographic Society in Lund and Crafoord
Foundation.

\bibliographystyle{abbrv}      


\bibliography{ce-biblio}

\begin{thebibliography}{10}

\bibitem{Kamvar200451}
Adaptive methods for the computation of pagerank.
\newblock {\em Linear Algebra and its Applications}, 386(0):51 -- 65, 2004.
\newblock Special Issue on the Conference on the Numerical Solution of Markov
  Chains 2003.

\bibitem{FAndersson:MThesis}
F.~Andersson.
\newblock Estimation of the quality of hyperlinked documents using a series
  formulation of pagerank.
\newblock Master's thesis, Mathematics, Centre for Mathematical sciences, Lund
  Institute of Technology, Lund University, May 2006:E22.
\newblock LUTFMA-3132-2006.

\bibitem{FAndersson_art_PR}
F.~Andersson and S.~Silvestrov.
\newblock The mathematics of internet search engines.
\newblock {\em Acta Appl. Math.}, 104, 2008.

\bibitem{berman1994nonnegative}
A.~Berman and R.~Plemmons.
\newblock {\em Nonnegative Matrices in the Mathematical Sciences}.
\newblock Number del 11 in Classics in Applied Mathematics.

\bibitem{B_Dennis_Matrix}
D.~Bernstein.
\newblock {\em Matrix Mathematics}.
\newblock Princeton University Press, 2005.

\bibitem{Bianchini:2005:IP:1052934.1052938}
M.~Bianchini, M.~Gori, and F.~Scarselli.
\newblock Inside pagerank.
\newblock {\em ACM Trans. Internet Technol.}, 5(1):92--128, Feb. 2005.

\bibitem{Brin1998107}
S.~Brin and L.~Page.
\newblock The anatomy of a large-scale hypertextual web search engine.
\newblock {\em Computer Networks and ISDN Systems}, 30(1-7):107 -- 117, 1998.
\newblock Proceedings of the Seventh International World Wide Web Conference.

\bibitem{A-25bilEig}
K.~Bryan and T.~Leise.
\newblock The $\$25,000,000,000$ eigenvector: The linear algebra behind google.
\newblock {\em SIAM Review}, 48(3):569--581, 2006.

\bibitem{Dhyani:2003:DVS:942051.942054}
D.~Dhyani, S.~S. Bhowmick, and W.-K. Ng.
\newblock Deriving and verifying statistical distribution of a hyperlink-based
  web page quality metric.
\newblock {\em Data Knowl. Eng.}, 46(3):291--315, Sept. 2003.

\bibitem{CEngstrom:MThesis}
C.~Engstr{\"o}m.
\newblock Pagerank as a solution to a linear system, pagerank in changing
  systems and non-normalized versions of pagerank.
\newblock Master's thesis, Mathematics, Centre for Mathematical sciences, Lund
  Institute of Technology, Lund University, May 2011:E31.
\newblock LUTFMA-3220-2011.

\bibitem{gantmacher1959theory}
F.~Gantmacher.
\newblock {\em The Theory of Matrices. Gantmacher}.

\bibitem{ilprints582}
T.~Haveliwala and S.~Kamvar.
\newblock The second eigenvalue of the google matrix.
\newblock Technical Report 2003-20, Stanford InfoLab, 2003.

\bibitem{5399514}
H.~Ishii, R.~Tempo, E.-W. Bai, and F.~Dabbene.
\newblock Distributed randomized pagerank computation based on web aggregation.
\newblock In {\em Decision and Control, 2009 held jointly with the 2009 28th
  Chinese Control Conference. CDC/CCC 2009. Proceedings of the 48th IEEE
  Conference on}, pages 3026--3031, 2009.

\bibitem{ilprints597}
S.~Kamvar and T.~Haveliwala.
\newblock The condition number of the pagerank problem.
\newblock Technical Report 2003-36, Stanford InfoLab, June 2003.

\bibitem{Kamvar:2003:EAR:775152.775242}
S.~D. Kamvar, M.~T. Schlosser, and H.~Garcia-Molina.
\newblock The eigentrust algorithm for reputation management in p2p networks.
\newblock In {\em Proceedings of the 12th international conference on World
  Wide Web}, WWW '03, pages 640--651, 2003.

\bibitem{lancaster1969theory}
P.~Lancaster.
\newblock {\em Theory of Matrices}.

\bibitem{NorrisMC}
J.~R. Norris.
\newblock {\em Markov chains}.
\newblock Cambridge University Press, New York, 2009.

\bibitem{Ryden2000m}
R.~Tobias and L.~Georg.
\newblock {\em Markovprocesser}.
\newblock Univ., Lund, 2000.

\end{thebibliography}

%
%

\end{document}